\documentclass[11pt,a4paper]{article}
\usepackage{amssymb}
\usepackage{makeidx}
\usepackage{amsfonts}
\usepackage{amsmath}
\usepackage{amsthm}
\usepackage{graphicx}
\usepackage[dvips]{epsfig}
\usepackage{longtable}
\usepackage{graphics}
\usepackage{dsfont}
\usepackage{xfp} 
\usepackage{pgfplotstable}
\usepackage{subcaption}
\usepackage{comment}
\usepackage{xcolor}

\usepackage{graphicx}				
\usepackage{amssymb}
\usepackage{mathtools}
\usepackage{tensor}

\newtheorem{theorem}{Theorem}[section]
\newtheorem{lemma}[theorem]{Lemma}
\newtheorem{proposition}[theorem]{Proposition}
\newtheorem{corollary}[theorem]{Corollary}

\newtheorem{assumption}{Assumption}
\newtheorem{example}{Example}[section]
\newtheorem{remark}{Remark}[section]

\newcommand{\E}{\mathbb{E}}
\newcommand{\Q}{\mathbb{Q}}
\renewcommand{\P}{\mathbb{P}}

\newcommand{\VaR}{\operatorname{VaR}}

\newcommand{\ES}{\operatorname{ES}}

\newcommand{\var}{\operatorname{var}}

\newcommand{\Cor}{\operatorname{Cor}}

\renewcommand{\phi}{\varphi}

\title{Cost-of-capital valuation with risky assets}
\author{Hansj\"org Albrecher\footnote{Department of Actuarial Science, Faculty of Business and Economics (HEC), University of Lausanne and Swiss Finance Institute, UNIL-Chamberonne, CH-1015 Lausanne} \and Filip Lindskog\footnote{Department of Mathematics, Stockholm University, SE-106 91 Stockholm} \and Hervé Zumbach\footnote{Department of Actuarial Science, Faculty of Business and Economics (HEC), University of Lausanne, UNIL-Chamberonne, CH-1015 Lausanne}}

\date{}

\begin{document}
\maketitle

\begin{abstract}
Cost-of-capital valuation is a well-established approach to the valuation of liabilities and is one of the cornerstones of current regulatory frameworks for the insurance industry. Standard cost-of-capital considerations typically rely on the assumption that the required buffer capital is held in risk-less one-year bonds. The aim of this work is to analyze the effects of allowing investments of the buffer capital in risky assets, e.g.~in a combination of stocks and bonds. In particular, we make precise how the decomposition of the buffer capital into contributions from policyholders and investors varies as the degree of riskiness of the investment increases, and highlight the role of limited liability in the case of heavy-tailed insurance risks. With a focus on non-life insurance, we present a combination of general theoretical results, explicit results for certain stochastic models and numerical results that emphasize the key findings. 
\end{abstract}


\section{Introduction}\label{sec1}
Risk-based solvency principles for insurance companies are at the heart of modern insurance regulation, see for instance the currently enforced Solvency II regulation in the European Union \cite{SolvencyII,EU_Directive_2025_2} and the Swiss Solvency Test (SST) in Switzerland \cite{fopi}. According to these principles, the valuation of insurance liabilities must reflect not only expected future cash flows but also the uncertainty surrounding them. A central component in this process is the cost of capital, which represents the return that investors demand for providing capital to absorb unexpected losses, where the necessary amount of capital to ensure solvency is set by the regulator on the basis of the underwritten policies and a risk measure. While financial risks can typically be priced in a market-consistent way (cf.\ \cite{koch2020market}), this is more subtle for insurance risk (see e.g.\ \cite{salahnejhad2023market}, \cite{engsner2023multiple} and \cite{schmidt23,oberpriller2024robust}). In the absence of a liquid market for insurance risks, one typically needs to adhere to a \textit{mark-to-model} approach. The economic value of the insurance liabilities is then defined as the expected claim costs ('\textit{best estimate}') plus a margin (called \textit{risk margin} in Solvency II, and \textit{market value margin} in the SST) which reflects these capital costs. The justification is that for this monetary amount another insurance undertaking should then be willing to take over these liabilities, as it covers the costs it generates for integrating it into one's business (an \textit{arm's length transaction}, cf.\ \cite{SolvencyII}). The suggested cost-of-capital rate (above the risk free rate) for such considerations, and the standard value currently implemented, is $6\%$. Recently, the European union decided to lower this value to 4.75\%, cf.\ \cite{reut,EU_Directive_2025_2}, which after national transposition is expected to be enforced in the member countries by January 2027. See e.g.\ \cite{floreani2011risk,albrecher2022cost} for economic considerations to justify particular cost-of-capital rates. \\

It may be natural to interpret that the insurance premiums should match the economic value of liabilities that the issuance of these insurance policies generates (as this is the value for which also another party should be willing to accept these risks). This would mean that the safety loading (on top of the pure premium reflecting expected claim costs) corresponds to the risk margin discussed above. From an actuarial perspective, this is a natural starting point for the concrete pricing of the policies, although several other factors will eventually play a role (including competition, demand/supply patterns, solidarity considerations etc.). More than that, at the general level of asset-liability management of the company, the company is exposed to several other sources of risks, and the regulatory rules ask to determine the solvency capital requirement that each of these generate, and the resulting values then need to be added up for the overall solvency capital satisfying the regulatory demand. That is, for simplicity, these risks are considered independently at first, and based on some (often very coarse) dependence assumptions between the different risk categories, the obtained sum can in a second step potentially be reduced by a diversification benefit, taking into account the fact that not all of these risks are likely to lead to losses at the same time. \\

In the approach as described above, the insurance liabilities are considered in a stand-alone fashion, with the capital costs they generate (and therefore the necessary insurance premiums) possibly over-estimated, as they rely on the assumption that the solvency capital is invested solely in risk-less bonds. But in practice some of that capital can be invested into risky assets, generating higher returns than the risk-less bond and providing additional income that reduces the claim costs. At the same time, this position in risky assets introduces additional risk, which itself translates into the need for further solvency capital.\\
\indent In this paper, we would like to specifically investigate this trade-off directly by considering those two risks together. Concretely, considering a mix of risky and risk-less investment of the solvency capital, we would like to study up to which degree of 'riskiness' of the investment the policyholders can benefit from smaller premiums for the same level of solvency of the insurance undertaking (under the assumption that the safety loading is determined by the generated capital costs). Our focus is on short-tailed risks in the non-life domain. Despite extensive literature on the joint valuation of actuarial and financial risks within the insurance product itself (particularly with life insurance applications in mind, see e.g.\ \cite{dhaene2017fair,delong2019fair,barigou2019fair,barigou2022insurance}), to the best of our knowledge an explicit analysis of the trade-off between risky and risk-less assets of a non-life insurer for a policyholder's perspective was not considered before. We therefore deliberately decide to keep the underlying model assumptions simple, in order not to blur the main analysis by overlays with other factors. This includes the restriction to a one-period model, independence between insurance and financial risk as well as a focus on the Value-at-Risk and Expected Shortfall as a risk measure, see Section \ref{sec2} for details. \\

It turns out that one can specify conditions under which the needed insurance premiums decrease with increasing weight in risky investment, up to a certain limit weight, up to which also the overall capital requirement is reduced. That is, a mild weight in risky investment is of advantage to all involved parties. Beyond that limit, the needed insurance premiums may decrease even further, but at the expense of overall increased capital requirements for the additionally introduced investment risk. Such considerations may also add to reflections about the justifications of currently implemented solvency capital charges in the standard model of Solvency II for equity risk (e.g.\ 39\% shock for Type 1 assets and 49\% shock for Type 2 assets, cf. \cite{EU_Delegated_Regulation_2015_35}). We reiterate that the purpose of the paper is to establish some concrete quantitative insights into the matter under very concrete and simple model assumptions. \\

The rest of the paper is organized as follows. Section \ref{sec2} lays out the concrete model assumptions on which we base the considerations of the paper and establishes a number of general monotonicity results for the needed insurance premium, the invested amount of the shareholders and the overall solvency capital requirement when increasing the proportion of risky assets in the management of insurance risks. For the more particular case of normally distributed insurance and financial risks we derive explicit formulas for the involved quantities in Section \ref{sec3}. 
In Section \ref{sec5} we then give and discuss concrete numerical illustrations of the effects of risky investments on the needed insurance premiums and solvency requirements, and Section \ref{sec6} concludes. All proofs are delegated to an appendix.

\section{A cost-of-capital valuation with risky assets}\label{sec2}
Assume that at time $0$ liabilities corresponding to an aggregate claim amount $X_1$ at time $1$ are transferred from an insurance company to an empty company whose purpose is to carry out the runoff of the liabilities. $X_1$ is not replicable and no replicating portfolio is considered. For simplicity, we assume that all payoffs at time $1$ are discounted to monetary values at time $0$ (alternatively, that the risk-less interest rate is zero). 
Along with the liabilities, the following cash amounts are transferred: a cash amount $C_0$ from the shareholders and 
a cash amount $V_0$ from the insurance company. Since only the liabilities and the cash amount $V_0$ are transferred from the insurance company, $V_0$ should be interpreted as a theoretical premium for these liabilities. 
The cash transfers are necessary because the new entity receiving the liabilities needs capital in order to comply with the solvency standards of insurance regulation. 

Assume now that the (deterministic) cash amount $R_0=V_0+C_0$ is immediately invested into an asset (or a collection of assets) with gross return $Z_1$ giving the risky amount $R_0Z_1$ at time $1$. 
At time $1$, the payoffs to the shareholders and to the policyholders, respectively, are: 
\begin{align*}
Z_{\text{sh}}:=(R_0Z_1-X_1)^+, \quad Z_{\text{ph}}:=(R_0Z_1)\wedge X_1,
\end{align*}
where $x\wedge y$ means $\min(x,y)$ and $x^+=\max(x,0)$. 
The interpretation is that the policyholders receive what they are entitled to if there is sufficient capital available at that time, and any remaining available capital goes to the shareholders. For a discussion on the effects of the limited liability of the shareholders, see e.g.\ \cite{filipovic2015optimal,albrecher2022cost}.

The shareholders assign at time $0$ a value $C_0$ to their payoff $Z_{\text{sh}}$ at time $1$, and the aggregate theoretical premium is the remaining amount $V_0=R_0-C_0$ needed to finance $R_0$. Cost-of-capital valuation corresponds to 
\begin{align}\label{C0coc} 
C_0=\E[Z_{\text{sh}}]/(1+\eta),
\end{align}
where $\eta>0$ represents the cost-of-capital rate, which is the spread over the risk-free rate for the more risky investment 
(which corresponds to betting on a favorable runoff result). 
Writing $(R_0Z_1-X_1)^+=R_0Z_1-(R_0Z_1)\wedge X_1$, we get in the cost-of-capital case that 
\begin{align}\label{eq:V0formula} 
V_0=\frac{1}{1+\eta}\,\E[(R_0Z_1)\wedge X_1]+\frac{\eta}{1+\eta}\,R_0+\frac{1}{1+\eta}R_0\,\E[1-Z_1].
\end{align}
If the investment is into a risk-less bond, then $Z_1\equiv 1$ which gives as a special case of the valuation formula \eqref{eq:V0formula}
\begin{align*} 
V_0=\frac{1}{1+\eta}\,\E[R_0\wedge X_1]+\frac{\eta}{1+\eta}\,R_0,
\end{align*} 
which appears as Equation (8.16) in Mildenhall and Major \cite[p.199]{mildenhall}. It also corresponds to the one-period case of  \cite[Eq.10]{engler2023approximations} and to \cite[Eq.12]{AlDa} in the absence of limited liability of the shareholders.

The cash amount $R_0$ is the smallest amount such that the new entity managing the liability runoff is allowed to operate (is considered solvent). 
$R_0$ cannot be determined unless the associated investment strategy is given.   
Concretely, $R_0$ must satisfy 
\begin{align}\label{r0}
\rho(R_0Z_1-X_1)=0
\end{align} 
so that the final net worth $R_0Z_1-X_1$ is acceptable according to a solvency criterion given by a risk measure $\rho$. That is, each choice of the risk measure $\rho$ entails an overall needed cash amount $R_0$ for each given insurance risk $X_1$ and chosen asset strategy $Z_1$. In terms of risk measure, in this paper we focus on the Value-at-Risk and the Expected Shortfall. For $\alpha\in (0,1)$ small, e.g.~$0.5\%$ or $1\%$, these are defined as, with $Y=R_0Z_1-X_1$,  
\begin{align*} 
\VaR_{\alpha}(Y)=F_{-Y}^{-1}(1-\alpha) \quad\text{and}\quad \ES_{\alpha}(Y)=\frac{1}{\alpha}\int_0^{\alpha}\VaR_{\beta}(Y)d\beta. 
\end{align*}
Since $Z_1$ and $X_1$ are discounted, no discount factor appears in front of the quantile $F_{-Y}^{-1}(1-\alpha)$.  

\begin{remark}\normalfont 
We could consider the more general $\min\{r:\rho(rZ_1-X_1)\leq 0\}$ as definition of $R_0$  
instead of \eqref{r0} since there are stochastic models for which no $R_0$ satisfies \eqref{r0}. 
However, under reasonable assumptions (see Proposition \ref{prop1} below) $R_0$ is the unique solution to $\rho(R_0Z_1-X_1)=0$. 
\hfill $\diamond$
\end{remark}

\begin{remark}\label{rem:V0bounds}\normalfont
A model-independent upper bound for $V_0$ is obtained from using $\E[Y^+]\geq \E[Y]$ which, applied to \eqref{eq:V0formula}, gives 
\begin{align*}
V_0\leq R_0-\frac{1}{1+\eta}\E[R_0Z_1-X_1]=\frac{1+\eta-\E[Z_1]}{1+\eta}R_0+\frac{1}{1+\eta}\E[X_1].
\end{align*}
If $Z_1$ and $X_1$ have finite variances and if the risk measure $\VaR_{\alpha}$ is used to determine $R_0$, then 
a lower bound for $V_0$ follows from combining the Cauchy-Schwarz inequality $\E[|AB|]\leq \E[A^2]^{1/2}\E[B^2]^{1/2}$ with the identity $Y^+=Y\mathds{1}_{[0,\infty)}(Y)$. Concretely,  
\begin{align*}
\E[(R_0Z_1-X_1)^+]\leq \E[(R_0Z_1-X_1)^2]^{1/2}(1-\alpha)^{1/2}
\end{align*}
and $\E[Y^2]=\var(Y)+\E[Y]^2$ then lead to the lower bound 
\begin{align*}
V_0\geq R_0-\frac{(1-\alpha)^{1/2}}{1+\eta}\Big(R_0^2\var(Z_1)+\var(X_1)+\big(R_0\E[Z_1]-\E[X_1]\big)^2\Big)^{1/2}.
\end{align*}
\hfill $\diamond$
\end{remark}

\begin{remark}\label{rem:llo}\normalfont
Limited liability means that shareholders do not have to inject capital at time $1$ to offset a possible deficit at that time. Limited liability increases the value $C_0$ for shareholders from $\E[R_0Z_1-X_1]/(1+\eta)$ to $\E[(R_0Z_1-X_1)^+]/(1+\eta)$. The difference 
\begin{align*}
\frac{1}{1+\eta}\big(\E[(R_0Z_1-X_1)^+]-\E[R_0Z_1-X_1]\big)=-\frac{1}{1+\eta}\E[(R_0Z_1-X_1)^-]
\end{align*} 
is referred to as the value of the limited liability option. Note that this value coincides with the difference between the upper bound for $V_0$ in Remark \ref{rem:V0bounds} and the value of $V_0$ in case of limited liability. 
\hfill $\diamond$
\end{remark}

\begin{example}\label{ex:llo_pareto}\normalfont 
As an illustration, consider a Pareto-distributed insurance risk $X_1$ with cumulative distribution function 
\begin{equation}\label{parparam}F(x)=1-(x/x_m)^{-\beta}, \quad x>x_m>0\end{equation} and suppose that $\beta>1$,  ensuring a finite mean $\E[X_1]=x_m\beta/(\beta-1)$.  
Consider a purely risk-less investment and $R_0$ determined by $\VaR_{\alpha}$, 
\begin{align*}
R_0=\VaR_{\alpha}(-X_1)=x_m\alpha^{-1/\beta}=\frac{\beta-1}{\beta}\alpha^{-1/\beta}\E[X_1]. 
\end{align*}
In this Pareto model with only risk-less investment we can compute the value of the limited liability option (cf.~Remark \ref{rem:llo}) explicitly. It is given by the expression 
\begin{align*}
\frac{1}{1+\eta}\frac{\alpha^{-1/\beta+1}}{\beta}\E[X_1]. 
\end{align*}
For nonnegative $R_0,Z_1,X_1$, the (typically crude) upper bound $(R_0Z_1-X_1)^+\leq R_0Z_1$ gives the general upper bound $\E[X_1]/(1+\eta)$ for the value of the limited liability option. Here, with $Z\equiv 1$ and $X_1$ Pareto distributed, we see that we can actually come arbitrarily close to this upper bound by letting $\beta$ approach $1$. \\ 
\indent Consider $\alpha=0.005$. 
For $\beta=2$ and $\beta=1.1$ the two values of the limited liability option are approximately $0.03\cdot \E[X_1]$ and $0.53\cdot \E[X_1]$, respectively. The two corresponding values of $R_0$ are approximately $7.07 \cdot \E[X_1]$ and $11.23 \cdot \E[X_1]$, respectively. The upper bound for $V_0$ in Remark \ref{rem:V0bounds} (corresponding to unlimited liability) is here
\begin{align*}
\frac{\eta}{1+\eta}\frac{\beta-1}{\beta}\alpha^{-1/\beta}\E[X_1]+\frac{1}{1+\eta}\E[X_1]
=\frac{\E[X_1]}{1+\eta}\bigg(1+\frac{\alpha^{-1/\beta}}{\beta}\eta(\beta-1)\bigg).
\end{align*}
For $\beta=2$ and $\beta=1.1$ the two values of the upper bound for $V_0$ are therefore approximately $1.34\cdot \E[X_1]$ and $1.58\cdot \E[X_1]$, respectively.  
We obtain the actual value of $V_0$ by subtracting the value of the limited liability option: 
\begin{align*}
V_0=\frac{\E[X_1]}{1+\eta}\bigg(1+\frac{\alpha^{-1/\beta}}{\beta}\big(\eta(\beta-1)-\alpha\big)\bigg). 
\end{align*} 
For $\beta=2$ and $\beta=1.1$ the two values of $V_0$ are therefore approximately $1.31\cdot \E[X_1]$ and $1.05\cdot \E[X_1]$, respectively.  
The fact that, for a fixed $\E[X_1]$, the heavier tail implies a 
considerably smaller value of $V_0$ is due to the relatively large value of the limited liability option. \hfill $\diamond$
\end{example}
Investing in a risk-less bond means $Z_1\equiv 1$ and then (for translation-invariant $\rho$) we have $R_0=\rho(-X_1)$. Our interest lies in the consequences of investing the capital $R_0$ in (at least partially) risky assets. Therefore, the case $Z_1\equiv 1$ is only a benchmark  here, and we will in general  consider random variables $Z_1$ of the form 
\begin{align}\label{convc}
Z_1=Z^{w}_1:=wS_1+1-w\quad \text{for}\;w\in [0,1].
\end{align} 
This corresponds to investing a fraction $w$ in a risky asset (with value $S_0=1$ at time $0$ and discounted value $S_1$ at time $1$) and the remainder in a risk-less bond. 
When considering $Z_1$ of the form \eqref{convc} we will sometimes write $R_0^{w}, C_0^{w}, V_0^{w}$ to emphasize the dependence on $w$ for fixed $S_1$ and $X_1$. 
In the sequel, we will always assume the following: 

\begin{assumption}\label{basicassumption}
	The risk measure $\rho$ is either $\VaR_{\alpha}$ or $\ES_{\alpha}$. $X_1$ and $S_1$ are independent, and $S_1$ is absolutely continuous (having a density). There exists a unique solution $R_0\geq 0$ to $\rho(R_0Z_1-X_1)=0$. 
\end{assumption}
The independence assumption between $X_1$ and $S_1$ is for simplicity of exposition (see also Remark \ref{rem3.5} and Section \ref{sec6}). 
The assumption on the existence of a unique solution $R_0$ is in fact not very restrictive. For instance, one can derive the following result (the proof of which is given in Appendix \ref{app1}). 

\begin{proposition}\label{prop1}
	If $\rho$ is either $\VaR_{\alpha}$ or $\ES_{\alpha}$, $X_1$ and $S_1$ are independent and take nonnegative values only, $S_1$ is absolutely continuous and $\P(X_1=0) < 1-\alpha$, then there exists a unique $R_0>0$ solving \eqref{r0}.
\end{proposition}

\begin{remark}\label{rem2}\normalfont
Since we are considering positively homogeneous risk measures, $0=\rho(R_0Z_1-X_1)=\rho(aR_0Z_1-aX_1)$ for any $a>0$. Hence, for a fixed $Z_1$, replacing $X_1$ by $aX_1$ changes the necessary invested amount from $R_0$ to $aR_0$. Correspondingly, $C_0=\E[(R_0Z_1-X_1)^+]/(1+\eta)$ changes to 
\begin{align*}
\E[(aR_0Z_1-aX_1)^+]/(1+\eta)=a\E[(R_0Z_1-X_1)^+]/(1+\eta)=aC_0. 
\end{align*}
As a result, $V_0$ also changes to $aV_0$. 
Hence, it is sufficient to restrict the analysis to insurance liability variables $X_1$ satisfying $\E[X_1]=1$. 
\hfill $\diamond$
\end{remark}

\subsection{Monotonicity properties}
Throughout this paper, for any two random variables $Y_1,Y_2$ the notation $Y_1\leq_{\text{st}} Y_2$ refers to first-order stochastic dominance, i.e.\ for the respective cumulative distribution functions we have $F_{Y_1}(x)\ge F_{Y_2}(x)$ for all $x\in{\mathbb R}$. 

We need to establish some overall soundness of the functionals $R_0$ and $V_0$. 

\begin{proposition}\label{prop23}
Assume that Assumption \ref{basicassumption} holds and that $Z_1$ takes only nonnegative values.  
	\begin{itemize}
		\item[(i)] For $Z_1$ fixed, $X_1\mapsto R_0$ is increasing with respect to increasing first-order stochastic dominance.
		\item[(ii)] For $X_1$ fixed, $Z_1\mapsto R_0$ is decreasing with respect to increasing first-order stochastic dominance.
		\item[(iii)] For $Z_1$ fixed and and $\E[Z_1]\leq 1+\eta$, $X_1\mapsto V_0$ is increasing with respect to increasing first-order stochastic dominance.
		\item[(iv)] For $X_1$ fixed and $\E[Z_1]\leq 1+\eta$, $Z_1\mapsto V_0$ is decreasing with respect to increasing first-order stochastic dominance.
	\end{itemize}
\end{proposition}

Let $S_1\leq_{\text{icx}}\widetilde{S}_1$ mean that $\widetilde{S}_1$ is larger than $S_1$ in increasing convex order (which is equivalent to stop-loss order, and also equivalent to $\int_{\beta}^{1}F_{S_1}^{-1}(u)du\leq \int_{\beta}^{1}F_{\widetilde{S}_1}^{-1}(u)du$ for all $\beta\in [0,1]$, see e.g.\ \cite{denuit2006actuarial}). Take $S_1\geq_{\text{icx}} 1$.  Then for $\widetilde{w}\geq w$
\begin{align*}
&\int_{\beta}^{1}F_{\widetilde{w}S_1+1-\widetilde{w}}^{-1}(u)du-\int_{\beta}^{1}F_{wS_1+1-w}^{-1}(u)du\\
&\quad = (\widetilde{w}-w)\int_{\beta}^{1}F_{S_1}^{-1}(u)du+(1-\beta)(w-\widetilde{w})\\
&\quad = (\widetilde{w}-w)\bigg(\int_{\beta}^{1}F_{S_1}^{-1}(u)du-(1-\beta)\bigg)\geq 0.
\end{align*}
Hence, for  $S_1\geq_{\text{icx}} 1$ we see that $Z_1^w=wS_1+1-w$ increases in increasing convex order with $w$. The condition $S_1\geq_{\text{icx}} 1$ is easily checked. For a random variable with symmetric density it simply means that the mode (mean if it exists) is greater than one.

Increasing convex order enables results on how the fraction of the invested amount in the risky asset affects the value of liabilities. 
Let $\mu_w=w\E[S_1]+1-w$ and let $R_0^w$ refer to the quantity $R_0$ in \eqref{r0} under \eqref{convc} (in particular, $R_0^0=\rho(-X_1)$). 

\begin{proposition}\label{prop:icxV0}
Assume that Assumption \ref{basicassumption} holds, that $S_1$ takes only nonnegative values, and that $S_1\geq_{\text{icx}} 1$. 
If $\widetilde{w}\geq w$, $R_0^{\widetilde{w}}\leq R_0^{w}$ and $1+\eta\geq \mu_w$, then $V_0^{\widetilde{w}}\leq V_0^{w}$. 
\end{proposition}

\begin{proposition}\label{prop:icxV02}
Assume that Assumption \ref{basicassumption} holds. If $R_0^{w}\mu_w\geq R_0^{0}$, then $C_0^{w}\geq C_0^0$. 
If $R_0^{w}\mu_w\geq R_0^{0}\geq R_0^{w}$, then $V_0^{w}\leq V_0^0$. 
\end{proposition}

The following result says that, under mild conditions, some risky investment, i.e.~a small $w>0$ compared to $w=0$, is always beneficial in the sense that $R_0^w<R_0^0$ and $V_0^w<V_0^0$. 

\begin{proposition}\label{prop:init_der_base_case}
Let $\rho=\VaR_{\alpha}$ and take $S_1$ and $X_1$ to be nonnegative and independent. Assume that $\E[S_1]>1$ and that $X_1$ has a bounded and continuous density that is nonvanishing in a neighborhood of $R_0^0=F_{X_1}^{-1}(1-\alpha)$. 
If $w\mapsto R_0^w$ is continuously differentiable, then 
\begin{align*}
\lim_{w\to 0}\frac{d}{dw}R_0^w=\lim_{w\to 0}\frac{d}{dw}V_0^w=-R_0^0(\E[S_1]-1)<0, \quad \lim_{w\to 0}\frac{d}{dw}C_0^w=0. 
\end{align*}
\end{proposition}

\begin{remark}\normalfont 
If $X_1=x>0$ is degenerate, then the assumptions in Proposition \ref{prop:init_der_base_case} do not hold. 
In this case $0=\VaR_{\alpha}(R_0^w(wS_1+1-w)-x)=R_0^w(w(\VaR_{\alpha}(S_1)+1)-1)+x$. 
If further $\VaR_{\alpha}(S_1)\in (-1,0)$, then 
\begin{align*}
\lim_{w\to 0}\frac{d}{dw}R_0^w=x(\VaR_{\alpha}(S_1)-1)>0. 
\end{align*}
The conclusion is that there is no $w$ which gives $R_0^w<R_0^0$ for all possible insurance risks $X_1$. In order for risky investment to reduce $R_0^w$ some uncertainty about the outcome of $X_1$ is needed.
This conclusion is consistent with Theorem 1.1 in \cite{Filipovic-MF-2008}, although \cite{Filipovic-MF-2008} focuses on convex risk measures.  
\hfill $\diamond$
\end{remark}

\begin{remark}\normalfont
We focus on cost-of-capital valuation in this paper, but we could consider replacing $C_0=\E[Z_{\text{sh}}]/(1+\eta)$ in \eqref{C0coc} by risk-neutral valuation $C_0=\E^{\Q}[Z_{\text{sh}}]$ and still obtain the qualitative conclusion that some risky investment is always beneficial. 
Under the additional assumptions that $\E^{\Q}[S_1]=1$, that independence between $S_1$ and $X_1$ holds also with respect to $\Q$, and that $X_1$ has a continuous density with respect to $\Q$ which is nonvanishing in a neighborhood of $R_0^0$,  
the expression for the limit of the derivative for $V_0^w$ takes the form 
\begin{align*}
\lim_{w\to 0}\frac{d}{dw}V_0^w=R_0^0(F^{\Q}_{X_1}(R_0^0)-1)(\E[S_1]-1)<0,
\end{align*}
where $F^{\Q}_{X_1}(R_0^0)=\Q(X_1\leq R_0^0)$. 
A proof of this fact follows along the same lines as the proof of Proposition \ref{prop:init_der_base_case}. 
\hfill $\diamond$
\end{remark}

\section{The Gaussian model}\label{sec3}
Let us now assume that $X_1$ and $Z_1$ are normally distributed. In that case one obtains explicit expressions for $R_0$, $C_0$ and $V_0$. Let us focus on the Value-at-Risk first. 

\subsection{Value-at-Risk}
For $Y\sim N(\mu,\sigma^2)$, it is well-known that  
\begin{align}\label{var}
	\VaR_{\alpha }(Y) = -\mu  + \sigma \Phi ^{-1}(1 - \alpha). 
\end{align}
We will always assume that $\alpha\in (0,1/2)$ so that $ \Phi ^{-1}(1-\alpha)>0$. 

\begin{proposition}\label{prop:capital_req_normal}
	Suppose that $X_1\sim N(\gamma,\nu^2)$ and $Z_1\sim N(\mu,\sigma^2)$ are independent with $\mu,\gamma,\sigma,\nu>0$. Then
	\begin{align}\label{eqr}
			\VaR_{\alpha }(R_0Z_1-X_1)=\gamma-R_0\mu+\Phi^{-1}(1-\alpha)\sqrt{R_0^2\sigma^2+\nu^2}.
	\end{align}
	\begin{itemize}
		\item[(i)]
		If $\mu>\sigma\Phi^{-1}(1-\alpha)$, then  
		\begin{align}\label{er0}
			R_0=\frac{\mu\gamma+\Phi^{-1}(1-\alpha)\sqrt{\gamma^2\sigma^2+\mu^2\nu^2-\sigma^2\nu^2\Phi^{-1}(1-\alpha)^2}}{\mu^2-\sigma^2\Phi^{-1}(1-\alpha)^2}
		\end{align}
		is the unique $R_0>0$ solving Equation \eqref{r0}. 
		\item[(ii)]
		If $\mu\le \sigma\Phi^{-1}(1-\alpha)$, then Equation \eqref{r0} has no positive solution $R_0$. 
	\end{itemize}
\end{proposition}

\begin{remark}\label{rem:equivR0}\normalfont 
If $\gamma\neq \nu\Phi^{-1}(1-\alpha)$, then $R_0$ in \eqref{er0} is equivalent to  
\begin{align*}
R_0=\frac{\gamma^2-\nu^2\Phi^{-1}(1-\alpha)^2}{\mu\gamma-\Phi^{-1}(1-\alpha)\sqrt{\gamma^2\sigma^2+\mu^2\nu^2-\sigma^2\nu^2\Phi^{-1}(1-\alpha)^2}}.
\end{align*}
This follows by multiplying both the numerator and denominator of \eqref{er0} by 
\begin{align*}
\mu\gamma-\Phi^{-1}(1-\alpha)\sqrt{\gamma^2\sigma^2+\mu^2\nu^2-\sigma^2\nu^2\Phi^{-1}(1-\alpha)^2},
\end{align*}
and factoring out $\mu ^2  - \sigma ^2 \Phi ^{-1}(1 - \alpha )^{2} $ from the numerator. 
The usefulness comes from the fact that we will consider $\gamma$ and $\nu$ fixed while allowing $\mu$ and $\sigma$ to vary as a result of considering different positions in a risky asset.   
\hfill $\diamond$
\end{remark}

\begin{remark}\normalfont 
Taking derivatives in \eqref{er0}, we immediately get the (intuitive) result that $\mu \mapsto R_{0} $  is decreasing and the functions $\sigma  \mapsto R_{0} $,  $\gamma \mapsto R_{0}$ and $\nu \mapsto R_{0}$ are increasing. \hfill $\diamond$
\end{remark}
\begin{remark}\normalfont 
 Note that for the existence of $R_0>0$, the mean return $\mu$ of the assets needs to exceed the threshold $\sigma\Phi^{-1}(1-\alpha)$. The latter expression depends on the volatility $\sigma$ of the assets as well as the security level $\alpha$. Only in case of a purely risk-less investment ($\mu=1,\sigma=0$) the existence of $R_0>0$ is guaranteed for all parameter values (and then trivially reduces to $R_0={\gamma+\nu\Phi^{-1}(1-\alpha)}$).\hfill $\diamond$
\end{remark}

Investing a portion $0<w<1$ of the wealth in the risky asset according to \eqref{convc}, in the Gaussian model this simply translates into replacing the parameters $(\mu,\sigma)$ by $(\mu_{w},\,\sigma_{w}):=(w\mu+1-w,w\sigma)$. 
The following result shows that, under very mild conditions, the function $w\mapsto R_0^{w}$ is strictly convex and in addition strictly decreasing for $w$ sufficiently small.   

\begin{proposition}\label{prop:R0_is_convex}
Suppose that $X_1\sim N(\gamma,\nu^2)$ and $S_1\sim N(\mu,\sigma^2)$ are independent with   
\begin{align*}
\nu>0,\sigma>0,\gamma>\nu\Phi^{-1}(1-\alpha) \text{ and } \mu>\max(1,\sigma\Phi^{-1}(1-\alpha)). 
\end{align*}
For $w\in [0,1]$, let $R_0^{w}$ be the positive solution to 
\begin{align*}
\VaR_{\alpha}(R_0^w(wS_1+1-w)-X_1)=0.
\end{align*}  
Then $w\mapsto R_0^{w}$ is strictly convex. Moreover, $R_0^{w}<R_0^0$ for all $w\in (0,\widehat{w})$, where $\widehat{w}=1$ if $\mu\geq 1+\sigma\Phi^{-1}(1-\alpha)$, and otherwise 
\begin{align*}
\widehat{w}=\frac{2(\mu-1)\nu\Phi^{-1}(1-\alpha)}{(1+\sigma\Phi^{-1}(1-\alpha)-\mu)(\mu-1+\sigma\Phi^{-1}(1-\alpha))(\gamma+\nu\Phi^{-1}(1-\alpha))}.
\end{align*}
\end{proposition}
Note that the quantity $\widehat{w}$ is of particular interest, since in view of Proposition \ref{prop:icxV0} it represents the limit weight of risky assets until which the overall capital requirement is not larger than $R_0^0$ (the one for purely risk-less assets), and the needed premium $V_0^w$ is smaller than the one for $w=0$. 
If $\max(1,\sigma\Phi^{-1}(1-\alpha))<\mu<1+\sigma\Phi^{-1}(1-\alpha)$, then it is easily verified that $\nu\mapsto \widehat{w}$ is strictly increasing. Hence, more insurance risk allows for more risky investment without increasing the total capital requirement above the level corresponding to purely risk-less investment. \\

The capital $C_0$ given by \eqref{C0coc} can be computed explicitly in the Gaussian setting. 
Together with the explicit expression for $R_0$ in Proposition \ref{prop:capital_req_normal} this means that all the quantities $R_0,C_0,V_0$ can be computed explicitly. 

\begin{proposition}\label{prop:C0V0_normal_model}
Suppose that $X_1\sim N(\gamma,\nu^2)$ and $Z_1\sim N(\mu,\sigma^2)$ are independent with $\mu,\gamma,\sigma,\nu>0$. 
If there exists an $R_0>0$ solving $\VaR_{\alpha}(R_0Z_1-X_1)=0$, then 
\begin{align*}
C_0&=\frac{R_0\mu-\gamma}{1+\eta}(1+\delta(\alpha))\quad \text{and} \\
V_0&=\frac{\gamma(1+\delta(\alpha))}{1+\eta}+R_0\frac{1+\eta-\mu(1+\delta(\alpha))}{1+\eta},
\end{align*}
where 
\begin{align*}
	0<\delta(\alpha):=\frac{\phi(\Phi^{-1}(1-\alpha))}{\Phi^{-1}(1-\alpha)}-\alpha<\frac{\alpha}{\Phi^{-1}(1-\alpha)^2}.
\end{align*}
\end{proposition}

\begin{remark}\label{rem:llo_gaussian}\normalfont
From Proposition \ref{prop:C0V0_normal_model} it follows that the value of the limited liability option (cf.~Remark \ref{rem:llo}) in the Gaussian model and $R_0$ determined by the risk measure $\VaR_{\alpha}$ is 
\begin{align*}
\delta(\alpha)\frac{R_0\mu-\gamma}{1+\eta}. 
\end{align*}
For $\alpha$ small, e.g.~$\alpha=0.005$, the value of the limited liability option is very small due to the light Gaussian tails.  
\hfill $\diamond$
\end{remark}

	The following result shows that investment in the (sufficiently attractive) risky asset always makes the payoff for a capital provider more attractive compared to the case with only risk-less investment. Hence, the contribution $C_0$ to the financing of the capital requirement should increase if risky investments are allowed.  
	
\begin{proposition}\label{prop:C0greater_normal_model}
Suppose that $X_1\sim N(\gamma,\nu^2)$ and $S_1\sim N(\mu,\sigma^2)$ are independent with $\gamma,\sigma,\nu>0$ and  
$\mu > \sigma  \Phi ^{-1}(1-\alpha )$.  
For $w\in [0,1]$, let $R_0^{w}$ be the positive solution to 
\begin{align*}
\VaR_{\alpha}(R_0^w(wS_1+1-w)-X_1)=0.
\end{align*}  
Then $C_{0}^{0} < C_{0}^{w}$ for all $w \in (0,1]$. 
\end{proposition}

If $R_0^w<R_0^0$ and $C_0^w>C_0^0$, then $V_0^w:=R_0^w-C_0^w<R_0^0-C_0^0=:V_0^0$.  
Hence, by combining Propositions \ref{prop:R0_is_convex} and \ref{prop:C0greater_normal_model} we immediately obtain the following result. 

\begin{corollary}\label{cor:Vw_smaller_than_V0}
Suppose that $X_1\sim N(\gamma,\nu^2)$ and $S_1\sim N(\mu,\sigma^2)$ are independent with $\sigma,\nu>0$,  
$\gamma>\nu\Phi^{-1}(1-\alpha)$ and $\mu>\max(1,\sigma\Phi^{-1}(1-\alpha))$.  
For $w\in [0,1]$, let $R_0^{w}$ be the positive solution to 
\begin{align*}
\VaR_{\alpha}(R_0^w(wS_1+1-w)-X_1)=0.
\end{align*}  
Then $V_0^{w}\leq V_0^0$ for all $w\in (0,\widehat{w})$, where $\widehat{w}$ is given in Proposition \ref{prop:R0_is_convex}. 
\end{corollary}

Note that the statement of Corollary \ref{cor:Vw_smaller_than_V0} is stronger than what we would get by combining Propositions \ref{prop:icxV0} and \ref{prop:R0_is_convex}, since we do not need any condition involving the cost-of-capital rate $\eta$. 
The statement of Corollary \ref{cor:Vw_smaller_than_V0} is a statement for $w$ sufficiently small where $R_0^w<R_0^0$. For $w$ larger with $R_0^w>R_0^0$ there is no hope for a simple statement to hold regardless of the parameter $\eta$. Note that if we let $\eta\to\infty$, then $V_0^w\to R_0^w$ and $R_0^w$ is typically not monotone, although decreasing for $w$ sufficiently small. However, by evaluating $V_0^w$ numerically we do observe that $V_0^{w}\leq V_0^0$ for all $w$ for realistic parameter values, see e.g.~Figure \ref{fig1} in Section \ref{sec5} for an illustration. We emphasize that the statement about $V_0^w$ for realistic parameter values implicitly assumes that $\mu\geq 1$ is not too small. That is, for $\mu=1$ it follows immediately from Proposition \ref{prop:C0V0_normal_model} that 
\begin{align*}
V_0^w=\frac{\gamma(1+\delta(\alpha))}{1+\eta}+R_0^w\frac{\eta-\delta(\alpha)}{1+\eta}.
\end{align*} 
So for $\mu=1$ we see that $V_0^{w}>V_0^0$ is equivalent to $R_0^{w}>R_0^0$. From Remark \ref{rem:equivR0} it follows that $R_0^{w}>R_0^0$ holds for all $w\in (0,1]$ if $\mu=1$ and $\gamma>\nu\Phi^{-1}(1-\alpha)$. 

\begin{remark}\normalfont \label{rem3.5}
In this paper we restrict our attention to the setting where $S_1$ and $X_1$ are independent. If the correlation coefficient $\Cor(S_1,X_1)>0$, then risky investments would not only generate a more favorable expectation (assuming $\mu>1$) but also partly hedge the insurance liability. One may wonder whether risky investments could be advisable even in the case $\Cor(S_1,X_1)<0$. In this case the expressions for $V_0$ and $C_0$ in Proposition \ref{prop:C0V0_normal_model} remain the same, but the expression for $R_0$ changes. With $c:=\Cor(S_1,X_1)\leq 0$, and under the assumption that $\mu>\sigma\Phi^{-1}(1-\alpha)$ and $\gamma\neq\nu\Phi^{-1}(1-\alpha)$,    
\begin{align*}
C^w_0&=\frac{R^w_0\mu_w-\gamma}{1+\eta}(1+\delta(\alpha)), \\
V^w_0&=\frac{\gamma(1+\delta(\alpha))}{1+\eta}+R^w_0\frac{1+\eta-\mu_w(1+\delta(\alpha))}{1+\eta}, \\
R^w_0&=\frac{\gamma^2-\nu^2\Phi^{-1}(1-\alpha)^2}{h(w)}, 
\end{align*}
where $\mu_w:=w\mu+1-w$, $\sigma_w:=w\sigma$, and 
\begin{align*}
h(w)&=\mu_w\gamma-c\sigma_w\nu\Phi^{-1}(1-\alpha)^2 
-\Big((\mu_w\gamma-c\sigma_w\nu\Phi^{-1}(1-\alpha)^2)^2 \\
&\quad-(\mu_w^2-\sigma_w^2\Phi^{-1}(1-\alpha)^2)(\gamma^2-\nu^2\Phi^{-1}(1-\alpha)^2)\Big)^{1/2}.
\end{align*}
The only difference from the case $c=0$ is that a term involving $c$ appears twice in the expression for $h(w)$.  
Straightforward computations give   
\begin{align*}
\frac{dR^w_0}{dw}(0)&=-(\gamma^2-\nu^2\Phi^{-1}(1-\alpha))\frac{h'(0)}{h^2(0)}\\
&=-(\gamma+\nu\Phi^{-1}(1-\alpha))(\mu-1+c\sigma\Phi^{-1}(1-\alpha)).
\end{align*}
We conclude that as long as $\mu>1-c\sigma\Phi^{-1}(1-\alpha)$, small positions $w>0$ in the risky asset will still imply $R_0^w<R_0^0$.  
\hfill $\diamond$
\end{remark}

\subsection{Expected Shortfall}
In the case of normal distributions, the expected shortfall of $Y\sim N(\mu,\sigma^2)$ at safety level $\alpha$ can simply be expressed as 
\begin{align*}
	\text{ES}_{\alpha }(Y) = -\mu + \sigma \frac{\phi (\Phi ^{-1}(1-\alpha ) ) }{\alpha }, 
\end{align*}
so the constant $\Phi ^{-1}(1-\alpha)$ for the Value-at-Risk in \eqref{var} is just replaced by another constant 
\begin{align*}
\psi=\frac{1}{\alpha}\int_0^{\alpha}\Phi^{-1}(1-\beta)d\beta
=\frac{1}{\alpha}\int_{\Phi^{-1}(1-\alpha)}^{\infty}z\phi(z)dz=\frac{\phi(\Phi^{-1}(1-\alpha))}{\alpha}.
\end{align*}
Correspondingly, we can directly adapt the results previously obtained for 
the Value-at-Risk. For instance, Proposition \ref{prop:capital_req_normal} turns into the following result. 

\begin{proposition}\label{prop:R0_ES}
	Suppose that $X_1\sim N(\gamma,\nu^2)$ and $Z_1\sim N(\mu,\sigma^2)$ are independent with $\mu,\gamma,\sigma,\nu>0$. Then, for $\rho(\cdot) = \ES_{\alpha}(\cdot)$, 
	\begin{itemize}
		\item[(i)]
		if $\mu>\sigma\psi$, then  
		\begin{align}\label{er0ES}
			R_0=\frac{\mu\gamma+\psi\sqrt{\gamma^2\sigma^2+\mu^2\nu^2-\sigma^2\nu^2\psi^2}}{\mu^2-\sigma^2\psi^2}
		\end{align}
		is the unique $R_0>0$ solving Equation \eqref{r0}. 
		\item[(ii)]
		if $\mu\le \sigma\psi$, then Equation \eqref{r0} has no positive solution $R_0$. 
	\end{itemize}
\end{proposition} 

Proposition \ref{prop:R0_is_convex} is easily adjusted to $\ES_{\alpha}$ instead of $\VaR_{\alpha}$. Replacing $\Phi^{-1}(1-\alpha)$ by $\psi$ gives the following result.

\begin{proposition}\label{prop:R0_is_convex_ES}
Suppose that $X_1\sim N(\gamma,\nu^2)$ and $S_1\sim N(\mu,\sigma^2)$ are independent with   
$\nu>0,\sigma>0$, $\gamma>\nu\psi$ and $\mu>\max(1,\sigma\psi)$. 
For $w\in [0,1]$, let $R_0^{w}$ be the positive solution to 
\begin{align*}
\ES_{\alpha}(R_0^w(wS_1+1-w)-X_1)=0.
\end{align*}  
Then $w\mapsto R_0^{w}$ is strictly convex. Moreover, $R_0^{w}<R_0^0$ for all $w\in (0,\widehat{w})$, where $\widehat{w}=1$ if $\mu\geq 1+\sigma\psi$, and otherwise 
\begin{align*}
\widehat{w}=\frac{2(\mu-1)\nu\psi}{(1+\sigma\psi-\mu)(\mu-1+\sigma\psi)(\gamma+\nu\psi)}.
\end{align*}
\end{proposition}

Proposition \ref{prop:C0V0_normal_model} is also easily adjusted to $\ES_{\alpha}$ instead of $\VaR_{\alpha}$. Replacing $1+\delta(\alpha)$ in Proposition \ref{prop:C0V0_normal_model} by $\Phi (\psi) + {\phi(\psi)}/{\psi}$ gives the following result. 

\begin{proposition}\label{prop:C0_V0_ES}
	Suppose that $X_1\sim N(\gamma,\nu^2)$ and $Z_1\sim N(\mu,\sigma^2)$ are independent with $\mu,\gamma,\sigma,\nu>0$. 
	If there exists an $R_0>0$ solving $\ES_{\alpha}(R_0Z_1-X_1)=0$, then 
	\begin{align*}
	C_{0} &= (R_0 \mu  - \gamma  ) \frac{\Phi (\psi  ) + \phi (\psi  ) /{\psi  }  }{1 + \eta  },\\
	V_{0} &= \gamma \frac{\Phi (\psi) + {\phi(\psi)}/{\psi}}{1+\eta} + R_0 \frac{1+\eta -\mu(\Phi (\psi) + {\phi(\psi)}/{\psi})}{1+\eta}.
\end{align*}
\end{proposition}

\noindent Finally, Corollary \ref{cor:Vw_smaller_than_V0} has the following version for $\ES_{\alpha}$ instead of $\VaR_{\alpha}$. 

\begin{corollary}
Suppose that $X_1\sim N(\gamma,\nu^2)$ and $S_1\sim N(\mu,\sigma^2)$ are independent with $\sigma,\nu>0$,  
$\gamma>\nu\psi$ and $\mu>\max(1,\sigma\psi)$.  
For $w\in [0,1]$, let $R_0^{w}$ be the positive solution to 
\begin{align*}
\ES_{\alpha}(R_0^w(wS_1+1-w)-X_1)=0.
\end{align*}  
Then $V_0^{w}\leq V_0^0$ for all $w\in (0,\widehat{w})$, where $\widehat{w}$ is given in Proposition \ref{prop:R0_is_convex_ES}.
\end{corollary}

\section{Numerical illustrations}\label{sec5}
Let us now put the results of the previous sections into concrete numerical conclusions for realistic parameter choices. Throughout this section we assume the cost-of-capital rate to be $\eta=0.06$ (see Section \ref{sec1} for a discussion of that parameter).

\subsection{The Gaussian model}
Assuming that both the insurance risk $X_1$ and the financial risk $Z_1$ are normally distributed, we can use the explicit formulas from Section \ref{sec3} to study the effect of the choice of $w$ on the resulting requirements for insurance premium $V_0$, solvency capital requirement $C_0$ and their sum $R_0$. Let us focus on the case of Value-at-Risk with safety level $\alpha = 0.005$, which is the risk measure used in Solvency II. Let us further assume  $\mu = 1.05$, $\sigma =0.2$, which may be considered a realistic return and volatility for an institutional investor, as well as the variations $\sigma=0.1$ and $\sigma=0.3$. For the insurance risk $X_1$, we assume a mean of $\gamma  = 1$ (cf.\ Remark \ref{rem2}) and a standard deviation of $\nu =0.3$ (which corresponds to a coefficient of variation of the total claim size of 0.3), together with a few variations of $\nu$. Figure \ref{fig1} plots the resulting values of $R_{0}$, $C_{0}$ and $V_{0}$ as functions of $w$ (recall that $w=0$ corresponds to purely risk-less investment and $w=1$ corresponds to purely risky investment, cf.\ \eqref{convc}). 

\begin{figure}[h]
	\centering
	\begin{subfigure}{0.32\textwidth}
		\includegraphics[width=\textwidth]{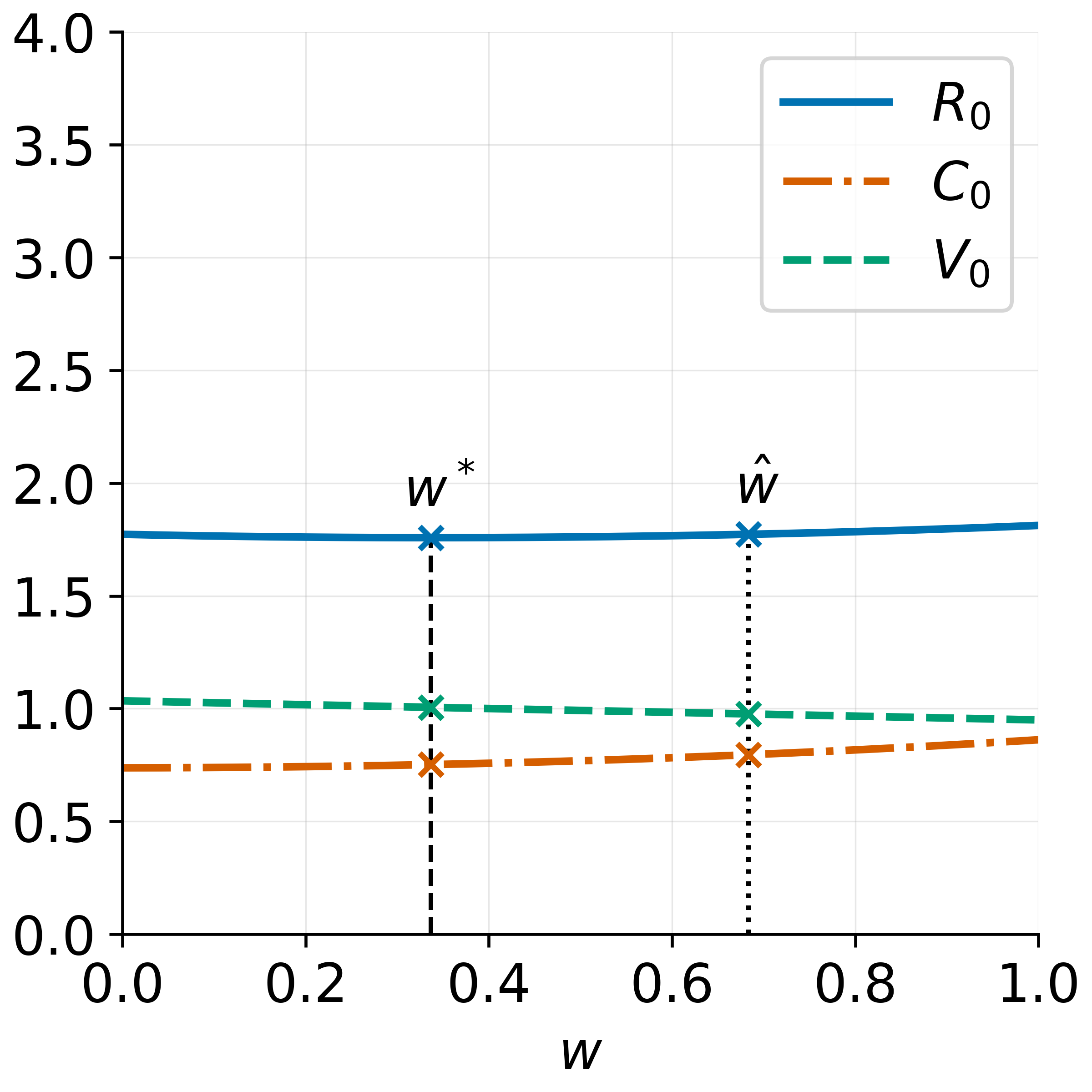}
		\caption{$\sigma=0.1$}
		\label{fig:first}
	\end{subfigure}
	\hfill
	\begin{subfigure}{0.32\textwidth}
		\includegraphics[width=\textwidth]{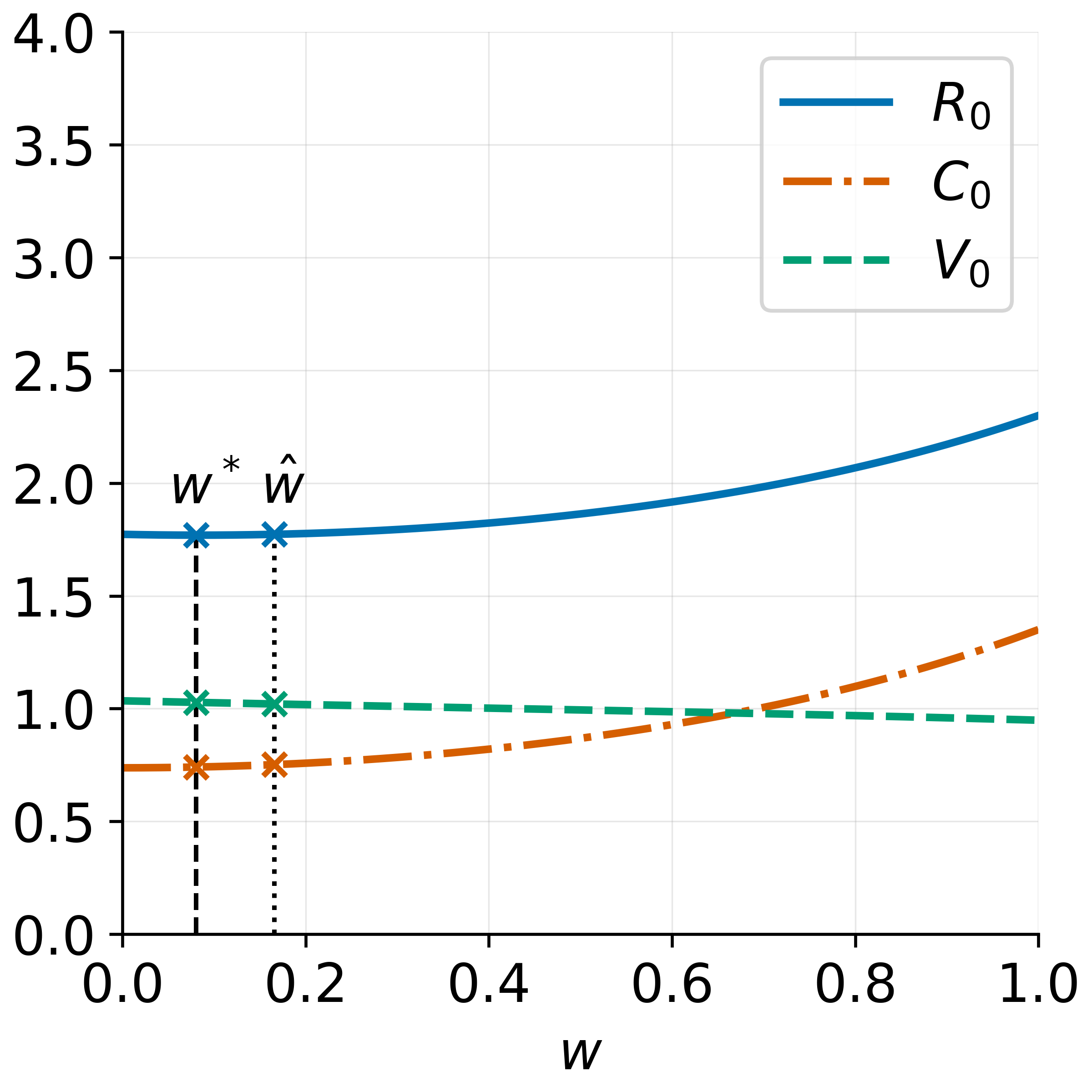}
		\caption{$\sigma=0.2$}
		\label{fig:second}
	\end{subfigure}
	\hfill
	\begin{subfigure}{0.32\textwidth}
		\includegraphics[width=\textwidth]{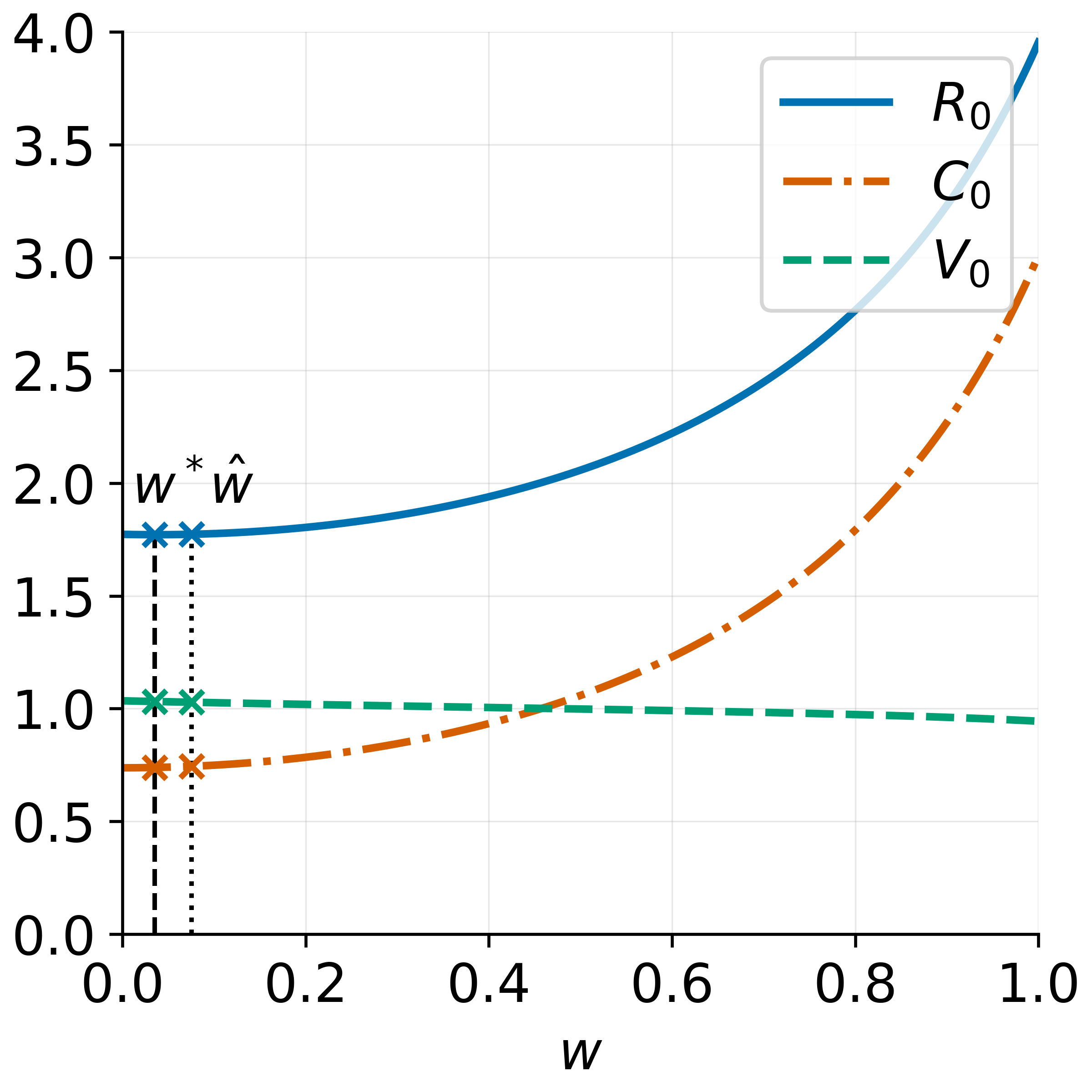}
		\caption{$\sigma=0.3$}
		\label{fig:third}
	\end{subfigure}
	\caption{$R_{0}^{w}, C_{0}^{w}$ and $V_{0}^{w}$ for the Gaussian model with $\mu = 1.05, \gamma  = 1,\nu =0.3 $, $\rho=\VaR_{0.005}$ and various values of $\sigma$}
	\label{fig1}
\end{figure}

\begin{figure}[h]
	\centering
	\begin{subfigure}{0.32\textwidth}
		\includegraphics[width=\textwidth]{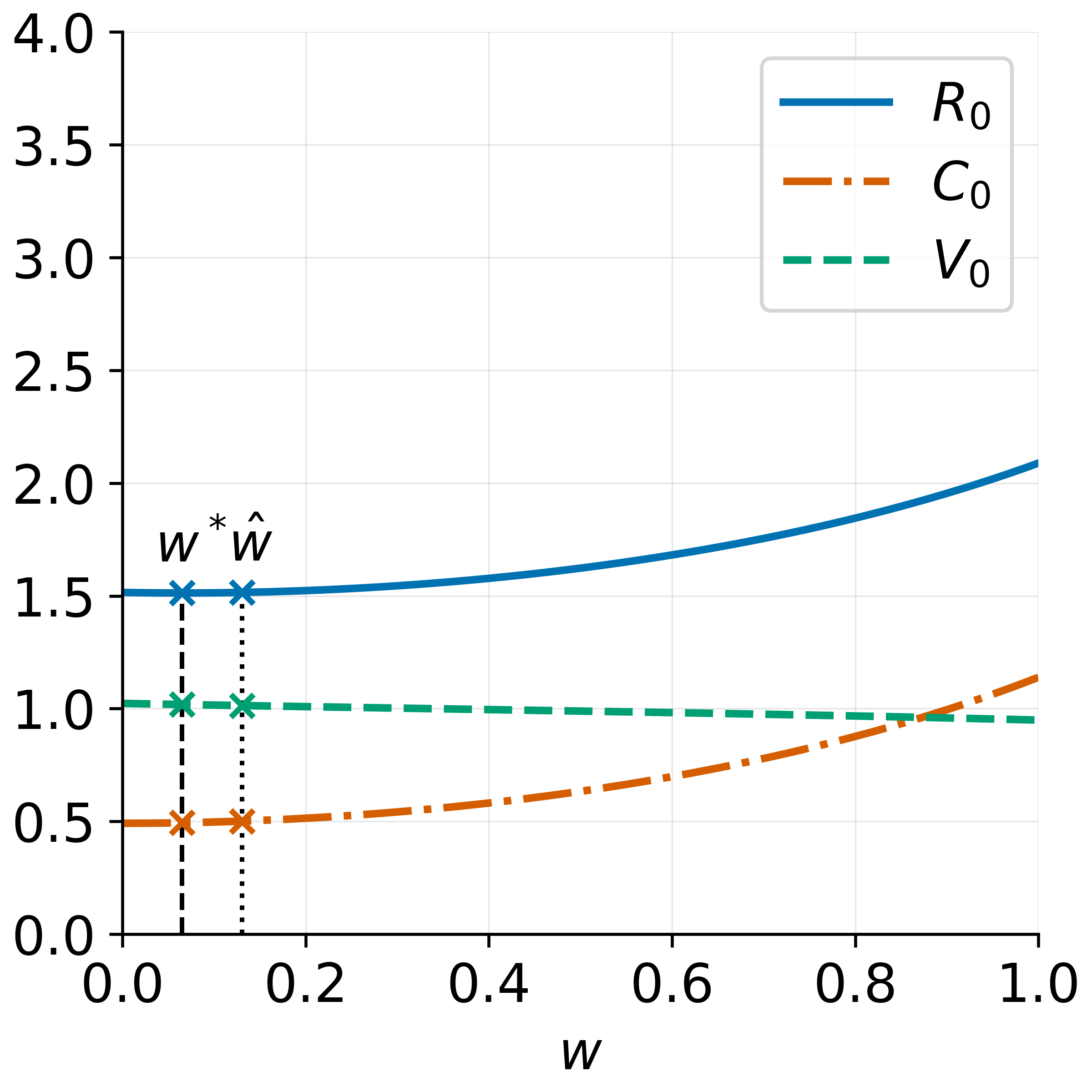}
		\caption{$\nu=0.2$}
		\label{fig:first1}
	\end{subfigure}
	\hfill
	\begin{subfigure}{0.32\textwidth}
		\includegraphics[width=\textwidth]{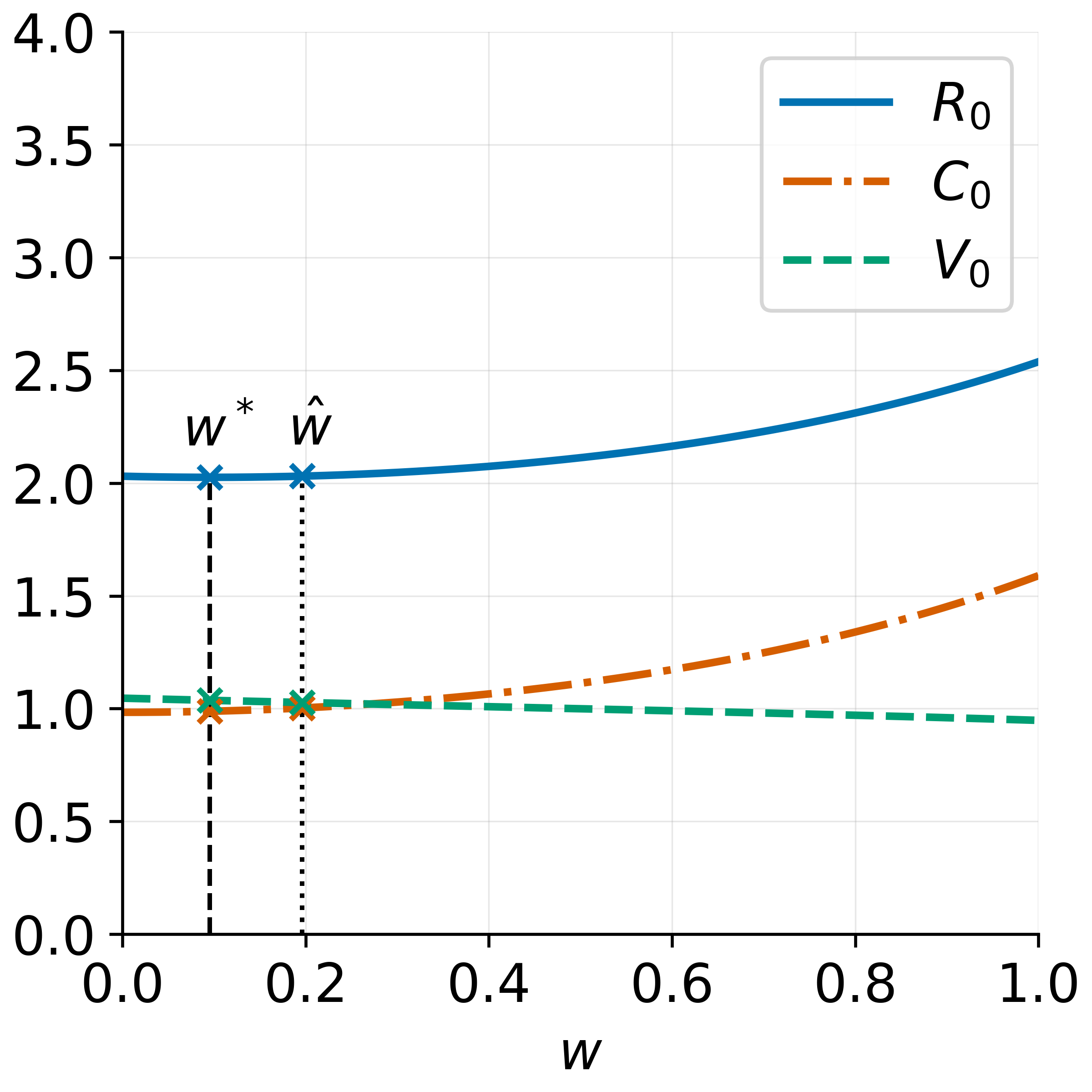}
		\caption{$\nu=0.4$}
		\label{fig:second2}
	\end{subfigure}
	\hfill
	\begin{subfigure}{0.32\textwidth}
		\includegraphics[width=\textwidth]{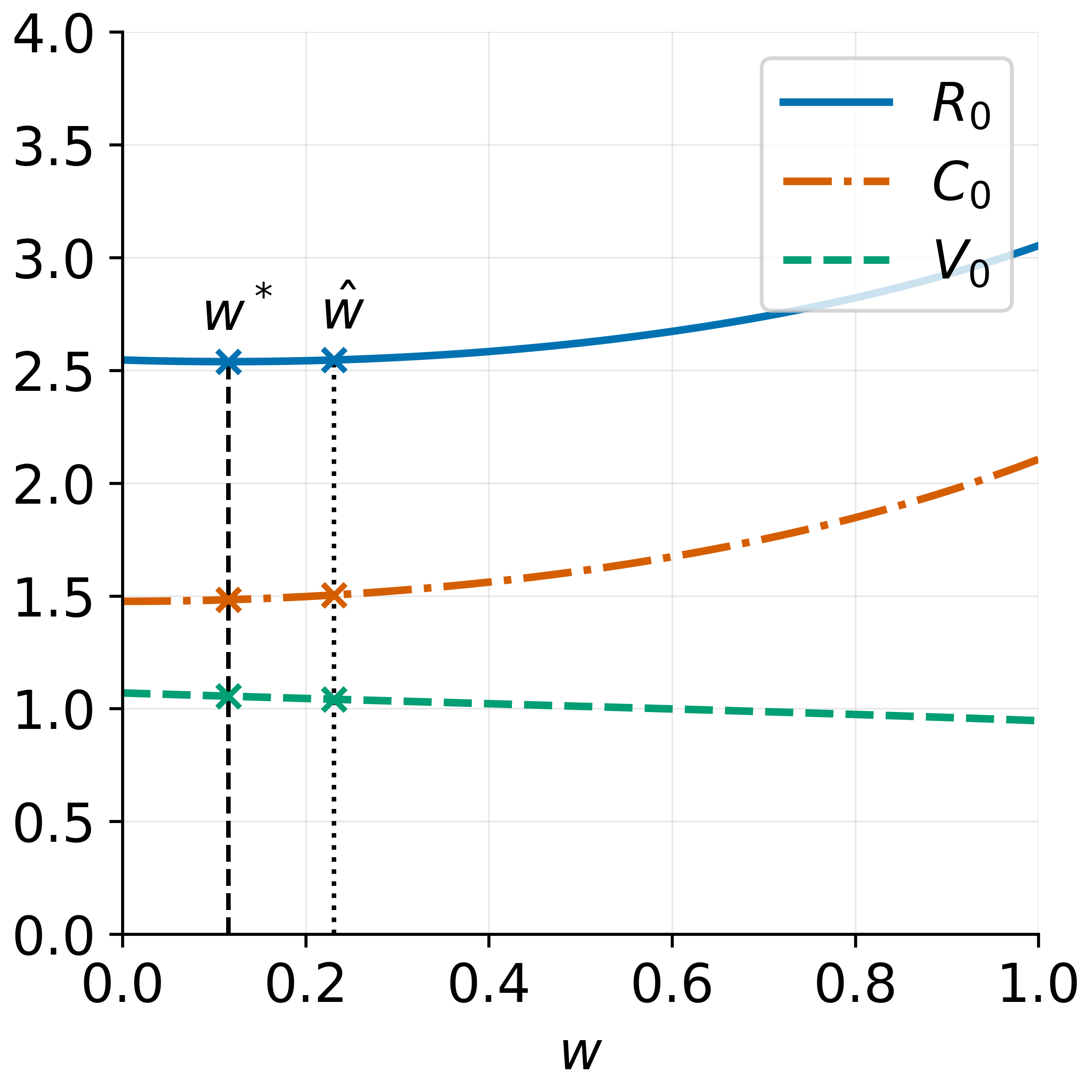}
		\caption{$\nu=0.6$}
		\label{fig:third3}
	\end{subfigure}
	\caption{$R_{0}^{w}, C_{0}^{w}$ and $V_{0}^{w}$ for the Gaussian model with $\mu = 1.05, \sigma=0.2,\gamma  = 1$, $\rho=\VaR_{0.005}$ and various values of $\nu$}
	\label{fig2}
\end{figure}
Note that while the needed insurance premium $V_0^{w}$ reduces with increasing risky asset allocation for all $w$, the needed solvency requirement $C_0^{w}$ from shareholders increases correspondingly, and the minimal overall capital requirement $R_0^{w}$ for $\sigma=0.2$ is achieved for $w^*=0.083$, e.g., only 8.3\% of the assets are invested in the risky asset, which then constitutes the most capital-efficient solution under the present model assumptions. One sees how this value $w^*$ decreases for increasing volatility in the risky asset, which matches the intuition. It is also insightful to see that for larger values of $w$ the solvency requirement $C_0^w$ starts to dominate the needed insurance premium $V_0$, more prominently for larger values of $\sigma$. The value $\widehat{w}$ until which the overall capital requirement $R_0^{w}$ is smaller than $R_0^{0}$ indeed corresponds to the one obtained from Proposition \ref{prop:R0_is_convex}. Figure \ref{fig2} plots the corresponding results for fixed $\sigma=0.2$, but varying standard deviation $\nu$ of the insurance risk. One observes how the overall capital requirement $R_0^w$ increases with $\nu$, and how $C_0^w$ becomes larger than the needed insurance premium $V_0^{w}$ already for smaller values of $w$. 
One also observes from Figure \ref{fig2} that the riskier the insurance risk $X_1$ is (larger value $\nu$), the larger fraction $w^*$ of the initial capital should be invested in the risky asset in order to minimize the total capital requirement $R_0^w$. \\

Finally, at first sight it may seem surprising that the increase of $R_0^{w}$ for increasing standard deviation $\nu$ of the insurance risk is almost exclusively swallowed by the increase of $C_0^{w}$ and leaves the premium $V_0^{w}$ virtually unchanged. The explanation is that the increased standard deviation (for fixed $\E[X_1]$) raises the value of the limited liability option. Concretely, for any absolutely continuous random variable $X_1$ (not only Gaussian!), it follows from \eqref{eq:V0formula} that 
	$$
	C_0^{w}=R_0^{w}\frac{\E[Z_1]}{1+\eta}-\frac{1}{1+\eta}\E[(R_0^{w}Z_1\wedge X_1)]
	$$
	Therefore 
	\begin{align*}
		\frac{\partial C_0^{w}}{\partial R_0^{w}}&=\frac{\E[Z_1]}{1+\eta}-\frac{1}{1+\eta}\E[Z_1\mathds{1}_{\{R_0^{w}Z_1<X_1\}}]\\
		&=\frac{\E[Z_1]}{1+\eta}-\frac{1}{1+\eta}\E[Z_1 \mid R_0^{w}Z_1<X_1]\P(R_0^{w}Z_1<X_1)
	\end{align*}
	With $\rho=\VaR_{\alpha}$ we then get 
	\begin{align}
		\frac{\partial C_0^{w}}{\partial R_0^{w}}
		&=\frac{\E[Z_1]}{1+\eta}-\frac{1}{1+\eta}\E[Z_1 \mid R_0^{w}Z_1<X_1]\alpha \nonumber\\
		&\geq \frac{\E[Z_1]}{1+\eta}-\frac{\E[Z_1]\alpha}{1+\eta}=\frac{\E[Z_1](1-\alpha)}{1+\eta}.\label{bbou}
	\end{align}
	As $\E \left[ Z_1 \right]$ is close to $1 + \eta  $ and $\alpha$ is very small, most of the increase of the required capital is indeed absorbed by the shareholders' contribution $C_0^w$ (see also Remark \ref{rem:llo}).

\subsection{The Lognormal model}
Let us alternatively consider the lognormal model, which is of interest for two reasons. 
First, it allows to see the effects of heavier right tails for the insurance risk when compared to the normal model. 
Secondly, the lognormal assumption on $S_1$ corresponds to a Black-Scholes market assumption for the risky asset, which is not uncommon in this framework. We will choose the parameters of the lognormal model in such a way that the first two moments of both $X_1$ and $S_1$ match the ones from the normal model. That is, for $\mu=1.05$, $\gamma=1$ and each choice of $\sigma$ and $\nu$ we use 
\begin{align*}
s_s^2=\log\bigg(1+\frac{\sigma^2}{\mu^2}\bigg), \quad m_s=\log\bigg(\mu\bigg(1+\frac{\sigma^2}{\mu^2}\bigg)^{-1/2}\bigg)
\end{align*}
and 
\begin{align*}	
s_x^2=\log\big(1+\nu^2\big), \quad m_x=-\frac{1}{2}\log(1+\nu^2)
\end{align*} 
to determine the respective parameters of the lognormal distributions. In the absence of explicit expressions, the figures are then obtained by Monte Carlo simulation. 
All figures are based on two i.i.d.~samples of size 1'000'000  from the standard normal distribution, where the simulated variates are suitably transformed into lognormal variates using the parameters $m_s,m_x,s_s,s_x$. \\

\begin{figure}[h]
	\centering
	\begin{subfigure}{0.32\textwidth}
		\includegraphics[width=\textwidth]{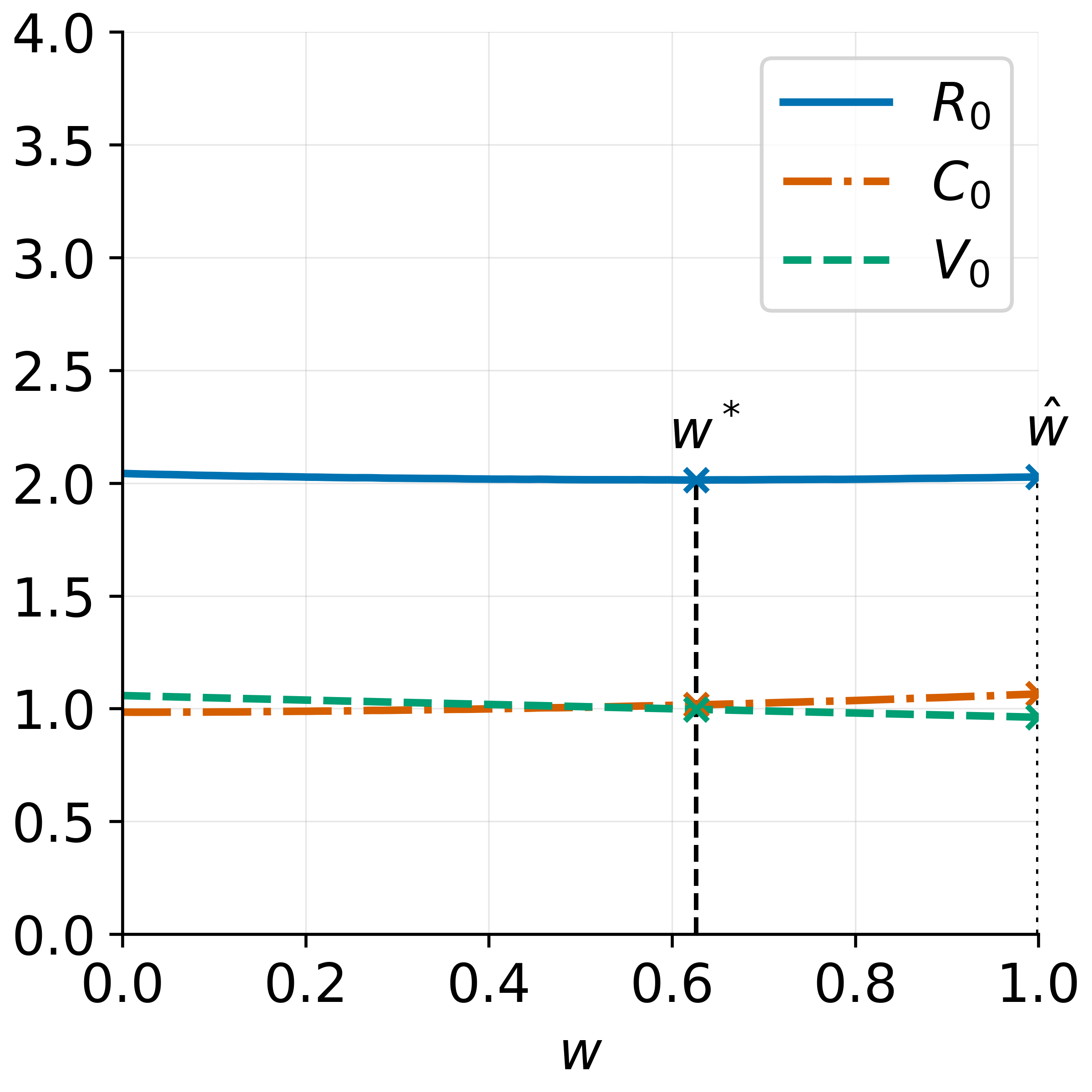}
		\caption{$\text{std}(S_1)=0.1$}
		\label{fig:firsta}
	\end{subfigure}
	\hfill
	\begin{subfigure}{0.32\textwidth}
		\includegraphics[width=\textwidth]{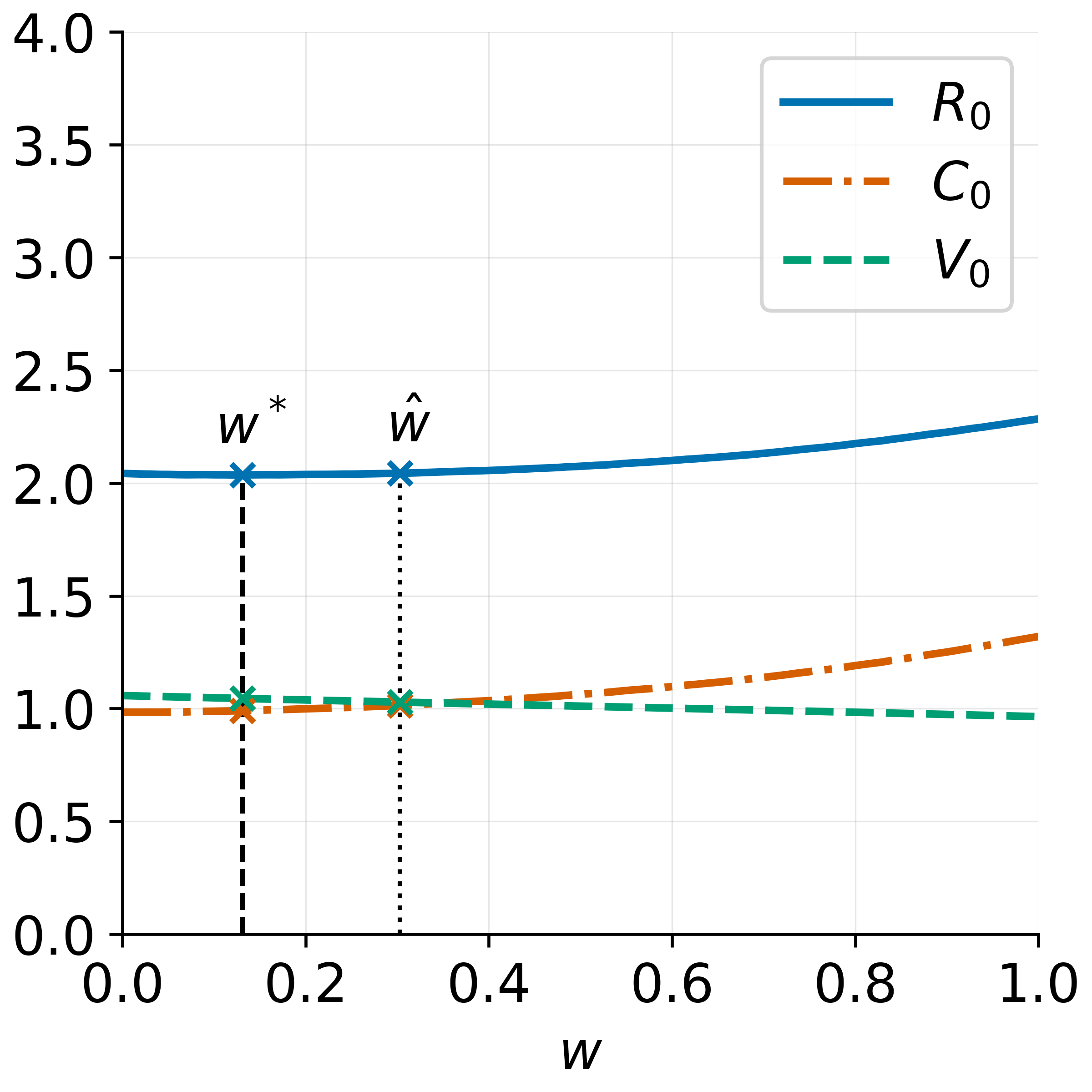}
		\caption{$\text{std}(S_1)=0.2$}
		\label{fig:seconda}
	\end{subfigure}
	\hfill
	\begin{subfigure}{0.32\textwidth}
		\includegraphics[width=\textwidth]{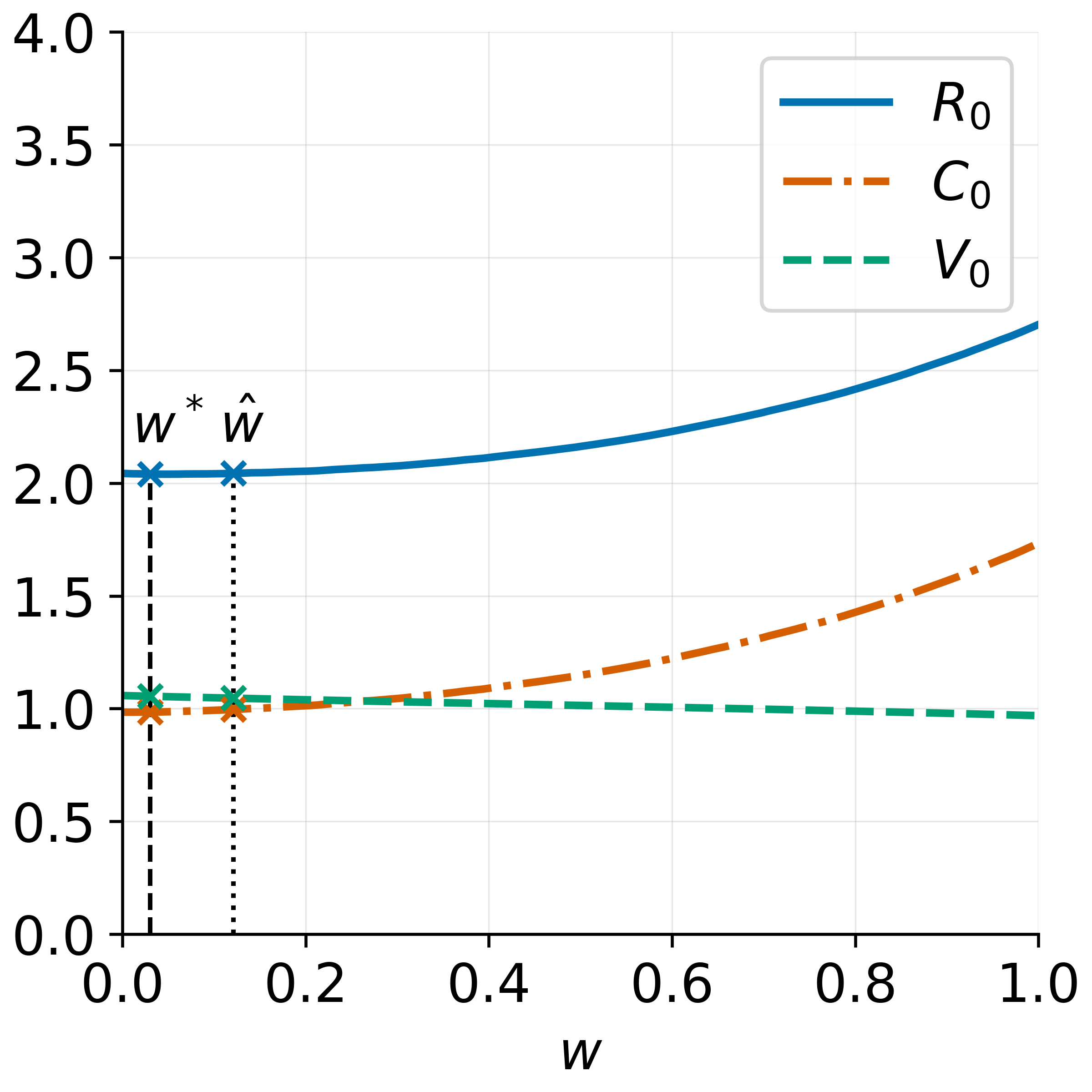}
		\caption{$\text{std}(S_1)=0.3$}
		\label{fig:thirda}
	\end{subfigure}
	\caption{$R_{0}^{w}, C_{0}^{w}$ and $V_{0}^{w}$ for the lognormal model with $\E[S_1] = 1.05$, $\E[X_1]  = 1$, $\text{std}(X)=0.3$, $\rho=\VaR_{0.005}$ and various values of $\text{std}(S_1)$}
	\label{fig3}
\end{figure}
\begin{figure}[h]
	\centering
	\begin{subfigure}{0.32\textwidth}
		\includegraphics[width=\textwidth]{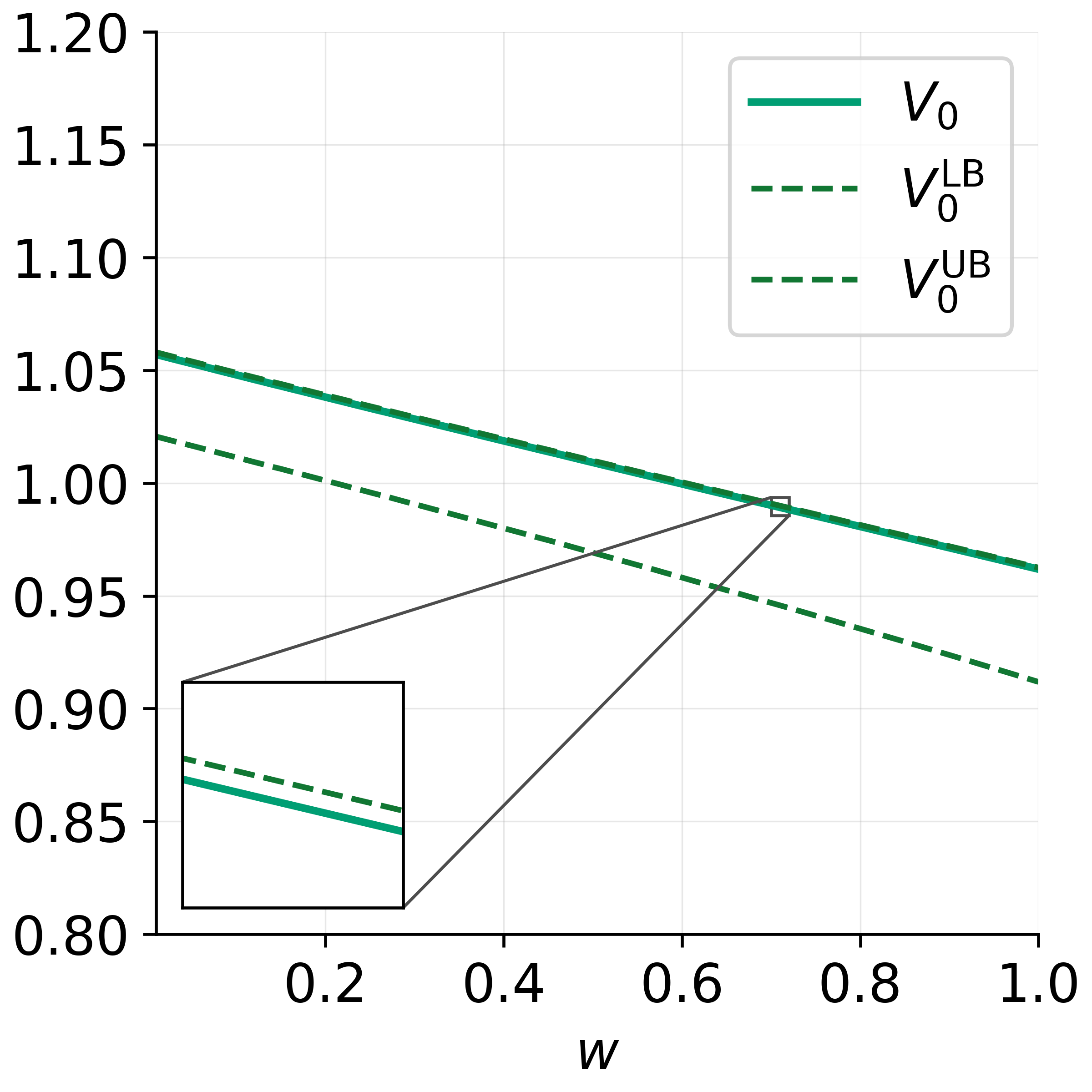}
		\caption{$\text{std}(S_1)=0.1$}
		\label{fig:firsta5}
	\end{subfigure}
	\hfill
	\begin{subfigure}{0.32\textwidth}
		\includegraphics[width=\textwidth]{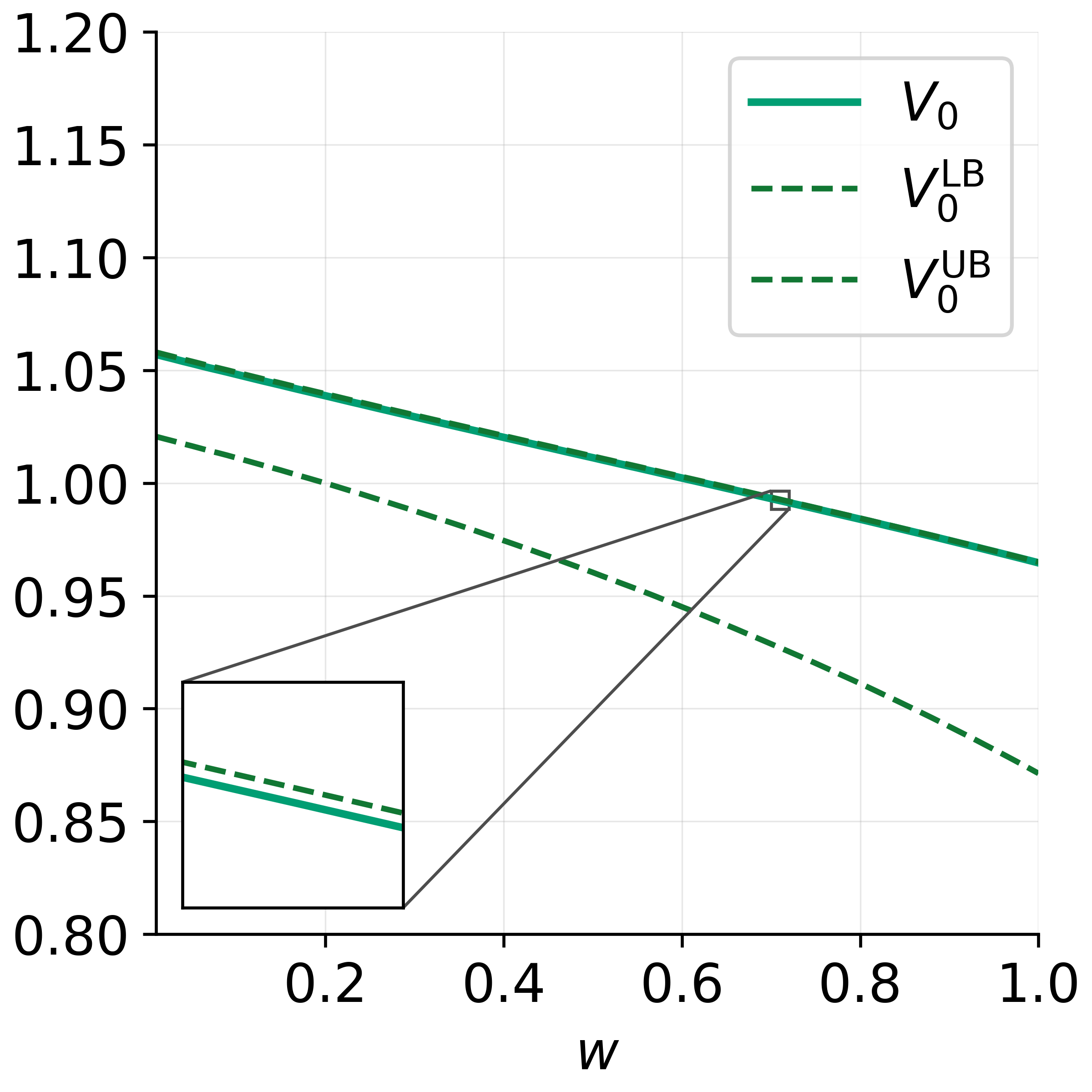}
		\caption{$\text{std}(S_1)=0.2$}
		\label{fig:seconda5}
	\end{subfigure}
	\hfill
	\begin{subfigure}{0.32\textwidth}
		\includegraphics[width=\textwidth]{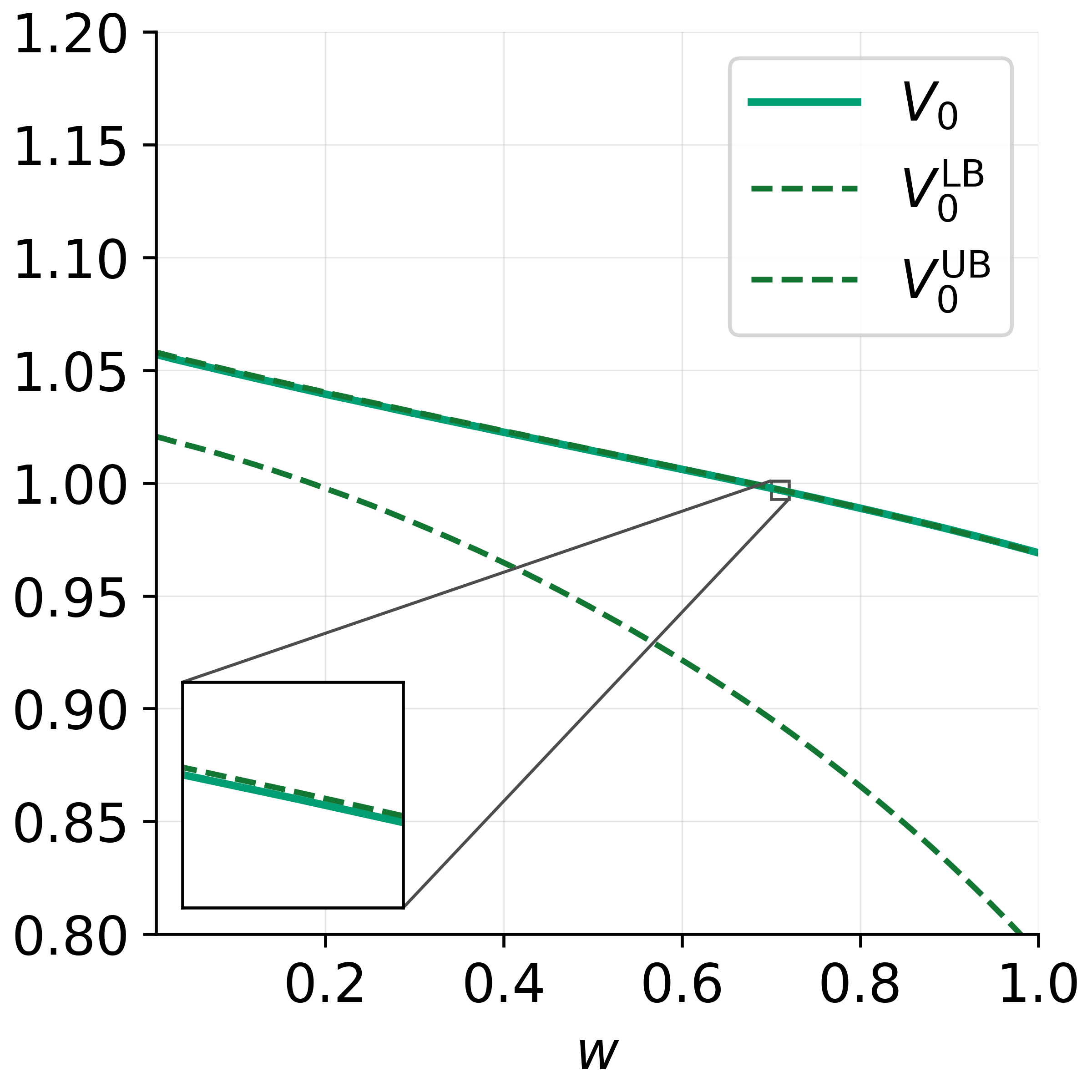}
		\caption{$\text{std}(S_1)=0.3$}
		\label{fig:thirda5}
	\end{subfigure}
	\caption{Upper and lower bounds for $V_{0}^{w}$ in Figure \ref{fig3}}
	\label{fig5}
\end{figure}

Figure \ref{fig3} is the analogue of Figure \ref{fig1} under lognormality for both insurance and financial risk. One sees that the heavier tail increases $R_0^w$, and $C_0^w>V_0^w$ already for smaller weights $w$. In all three situations depicted in Figure \ref{fig3}, both the optimal weight $w^*$ as well as the critical threshold $\widehat{w}$ are increased when compared to the same situation in Figure \ref{fig1} with identical first two moments of $S_1$ and $X_1$, which can also be attributed to the heavier tails. 
For large values of $w$, the safety loading $V_0^w-\E[X_1]$ in the premiums can become very small and even negative. 
The model-independent upper bound for $V_0^w$ from Remark \ref{rem:V0bounds} can be expressed as 
\begin{align*}
V_0^w \leq \E[X_1] + \frac{\eta (R_0^w-\E[X_1])-w(\E[S_1]-1)R_0^w}{1+\eta}
\end{align*}
from which it follows immediately that, for $\E[S_1]>1$,  
\begin{align*}
w>\frac{\eta}{\E[S_1]-1}\frac{R_0^w-\E[X_1]}{R_0^w}
\end{align*}
implies a negative safety loading for any model.

Figures \ref{fig5} zooms into the respective values of the needed premium $V_0^w$ and also depicts the bounds that were derived in Remark \ref{rem:V0bounds} (the upper bound is very tight, as a further magnification in the plots indicates). One observes that the lower bound is too coarse for practical purposes, but the upper bound is remarkably sharp in these cases. As the upper bound is obtained by ignoring the limited liability, this indicates that the value of the limited liability option is quite small (which here is due to the small value $\alpha=0.005$ and the fact that the right tail of $X_1-R_0Z_1$ beyond its $1-\alpha=0.995$ quantile value is not sufficiently heavy to significantly impact $V_0$). \\

Figure \ref{fig4} depicts the analogue of Figure \ref{fig2} under lognormality for both insurance and financial risk, in this case varying the standard deviation of the insurance risk. Again, each of these situations leads to larger overall capital requirement $R_0^w$ compared to the counterparts of Figure \ref{fig2}, and again the solvency capital requirement $C_0^w$ absorbs most of that increase, whereas the premium $V_0^w$ stays almost unchanged for increasing standard deviation. The latter is again due to the increased value of the limited liability option, cf.\ \eqref{bbou} and also Example \ref{ex:llo_pareto} for the case $w=0$.

Figure \ref{fig6} zooms into $V_0^w$ and depicts the bounds that were derived in Remark \ref{rem:V0bounds}. Again the sharpness of the upper bound is noteworthy. \\

\begin{figure}[h]
	\centering
	\begin{subfigure}{0.32\textwidth}
		\includegraphics[width=\textwidth]{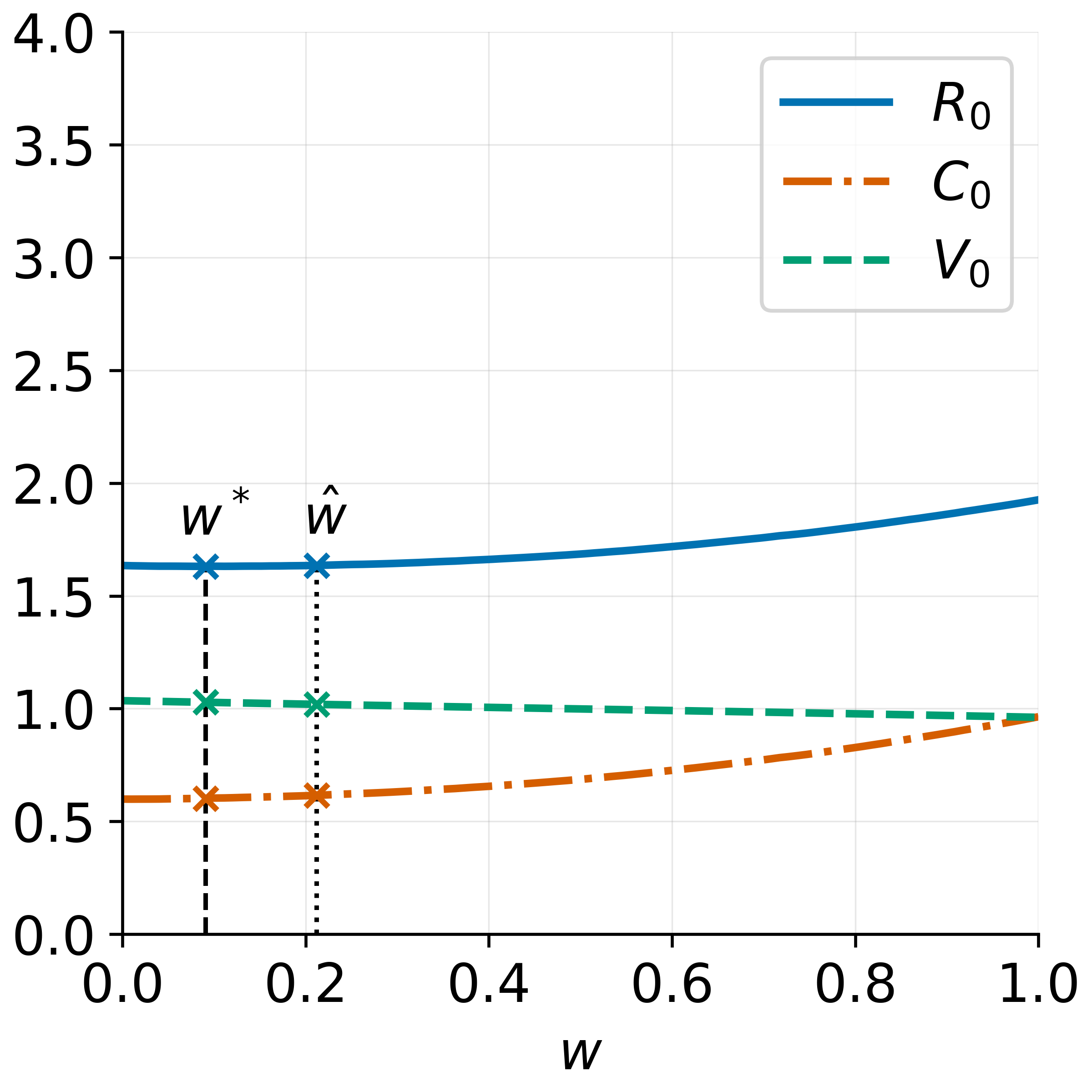}
		\caption{$\text{std}(X_1)=0.2$}
		\label{fig:firstb}
	\end{subfigure}
	\hfill
	\begin{subfigure}{0.32\textwidth}
		\includegraphics[width=\textwidth]{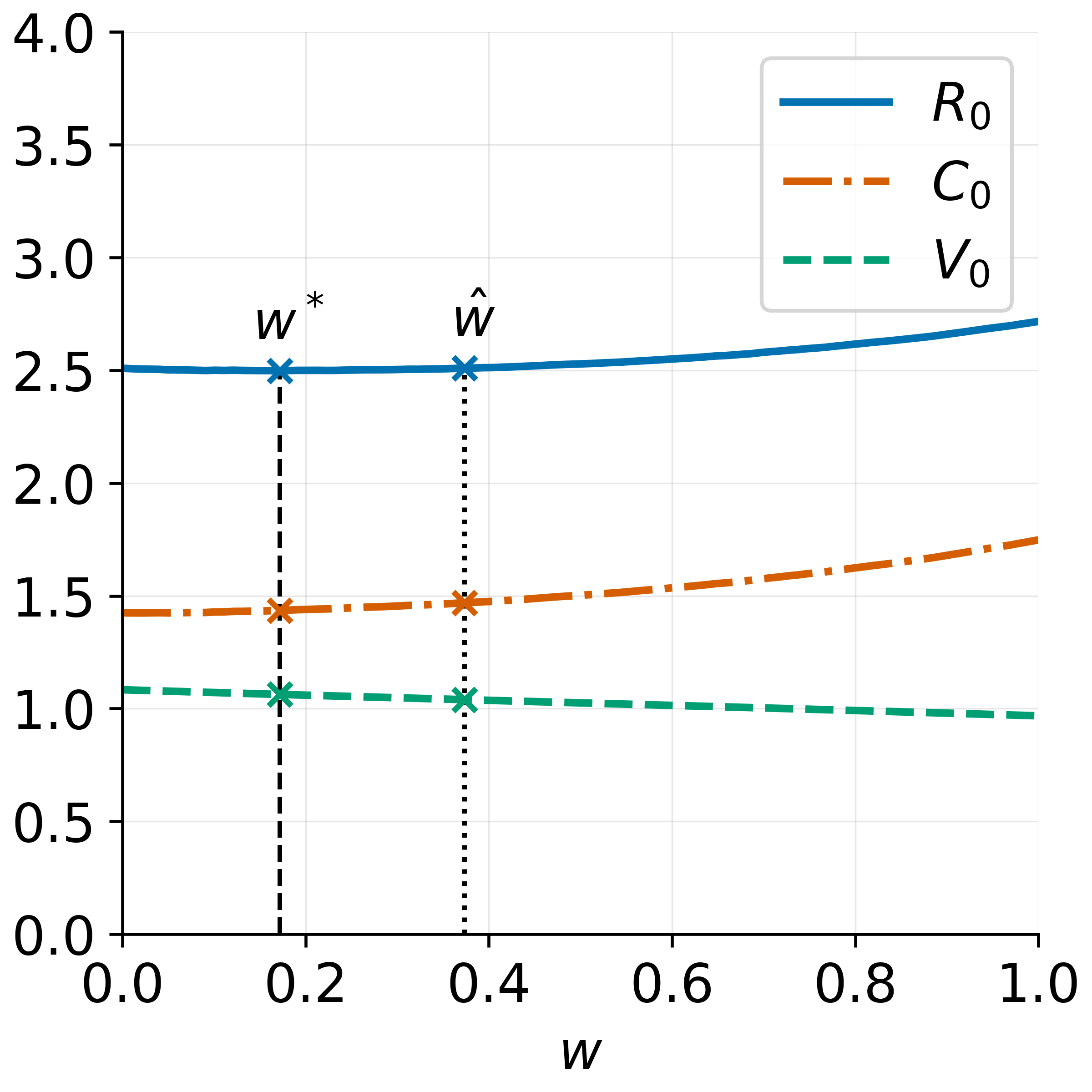}
		\caption{$\text{std}(X_1)=0.4$}
		\label{fig:secondb}
	\end{subfigure}
	\hfill
	\begin{subfigure}{0.32\textwidth}
		\includegraphics[width=\textwidth]{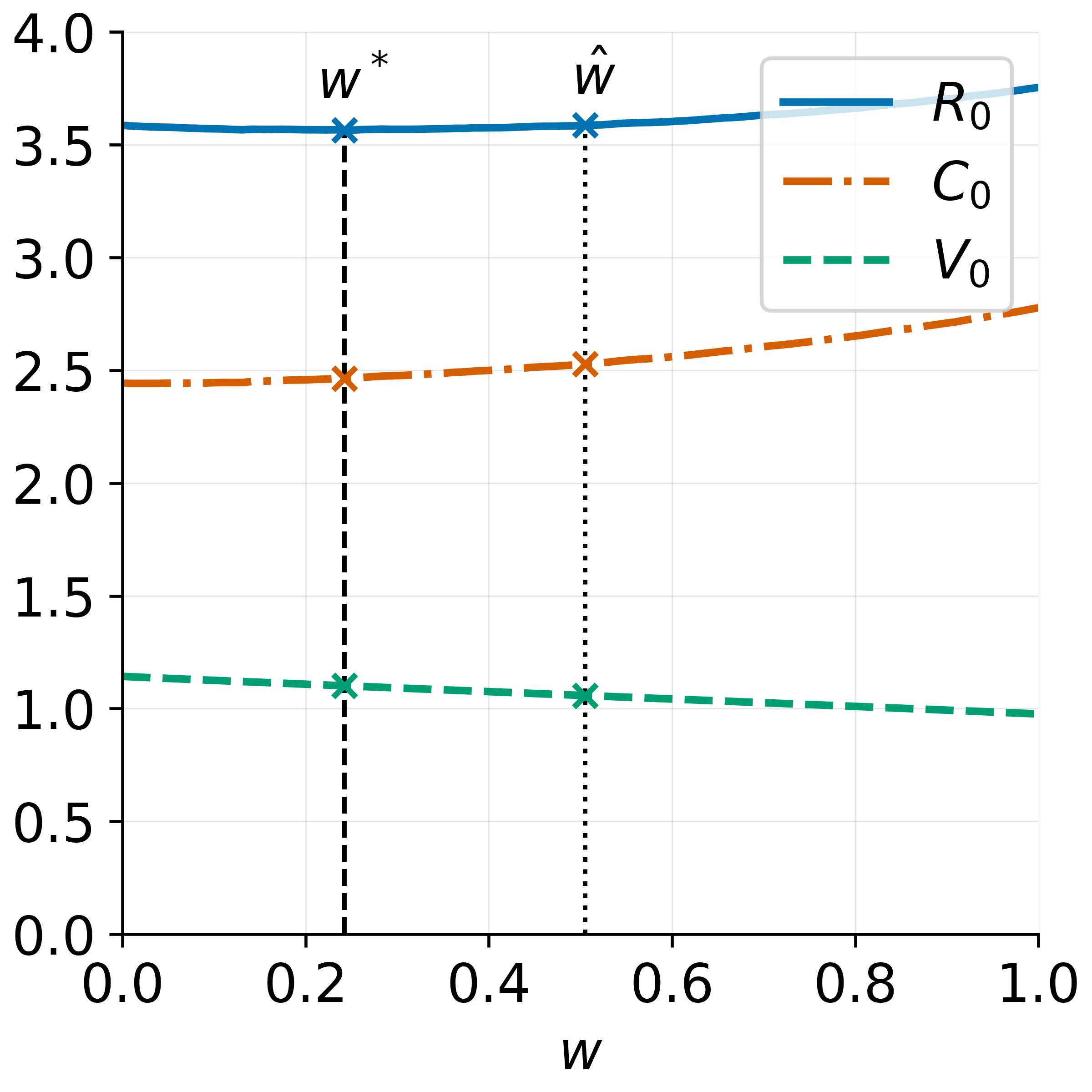}
		\caption{$\text{std}(X_1)=0.6$}
		\label{fig:thirdb}
	\end{subfigure}
	\caption{$R_{0}^{w}, C_{0}^{w}$ and $V_{0}^{w}$ for the lognormal model with $\E[S_1] = 1.05$, $\E[X_1]  = 1$, $\text{std}(S_1)=0.2$, $\rho=\VaR_{0.005}$ and various values of $\text{std}(X_1)$}
	\label{fig4}
\end{figure}

\begin{figure}[h]
	\centering
	\begin{subfigure}{0.32\textwidth}
		\includegraphics[width=\textwidth]{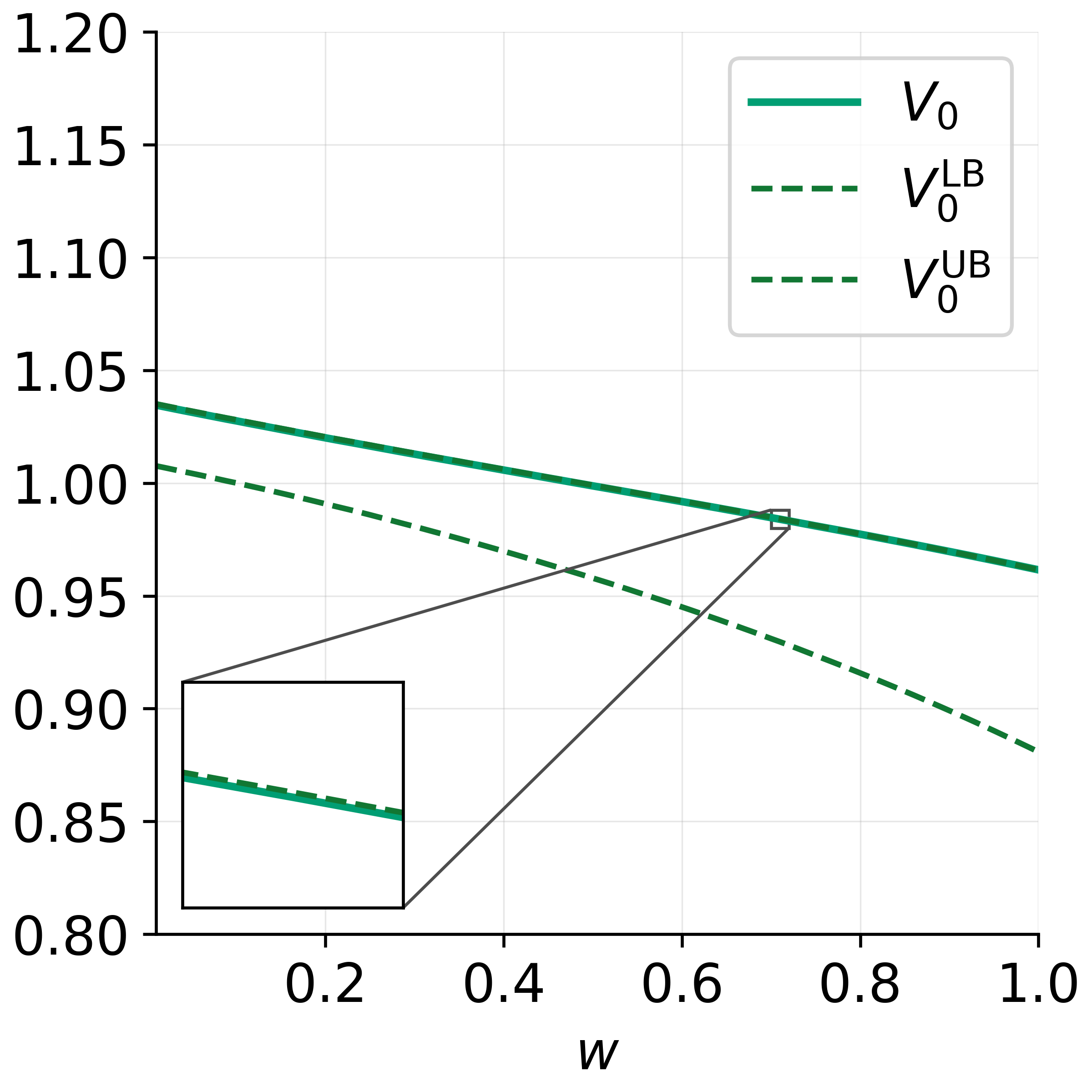}
		\caption{$\text{std}(X_1)=0.2$}
		\label{fig:firstb5}
	\end{subfigure}
	\hfill
	\begin{subfigure}{0.32\textwidth}
		\includegraphics[width=\textwidth]{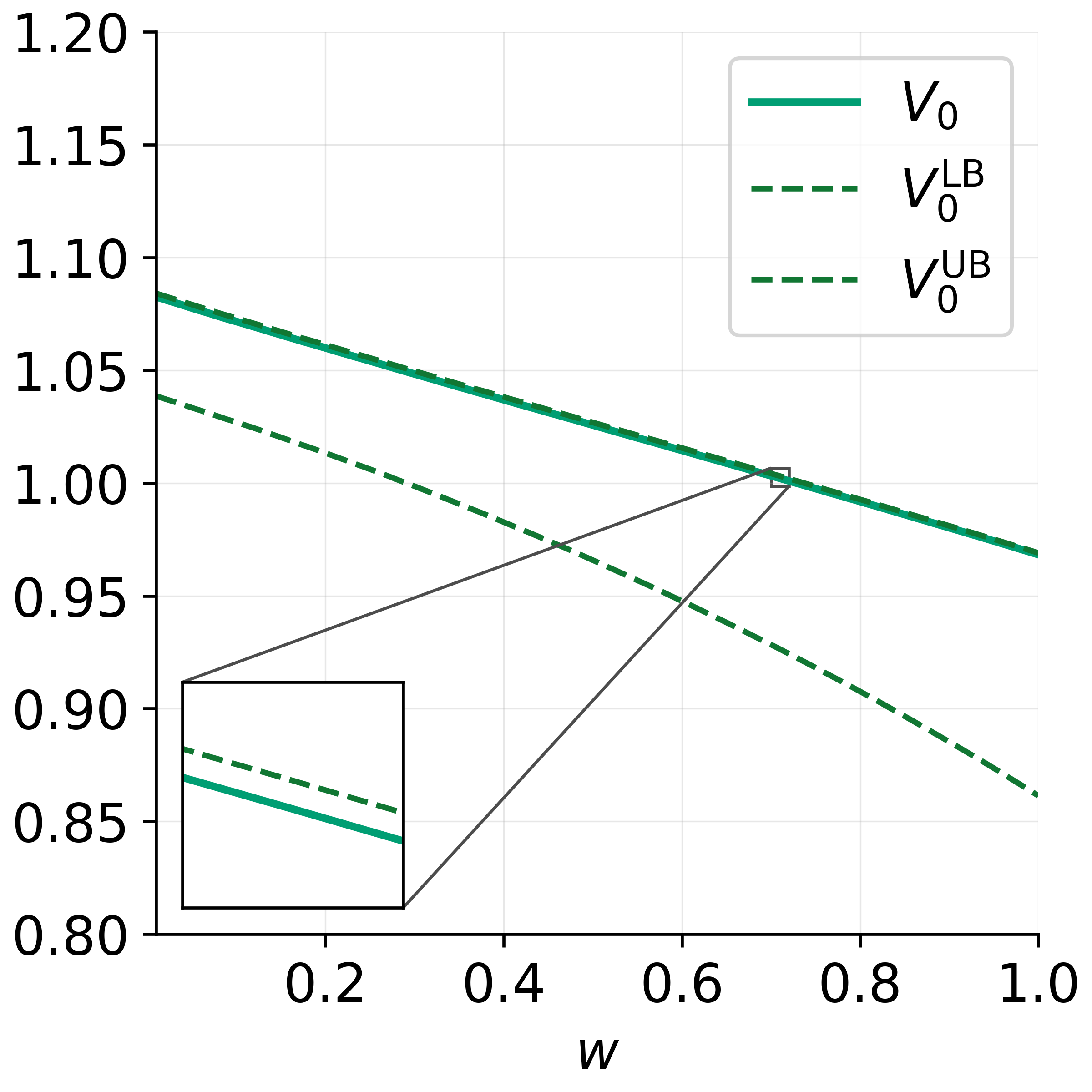}
		\caption{$\text{std}(X_1)=0.4$}
		\label{fig:secondb5}
	\end{subfigure}
	\hfill
	\begin{subfigure}{0.32\textwidth}
		\includegraphics[width=\textwidth]{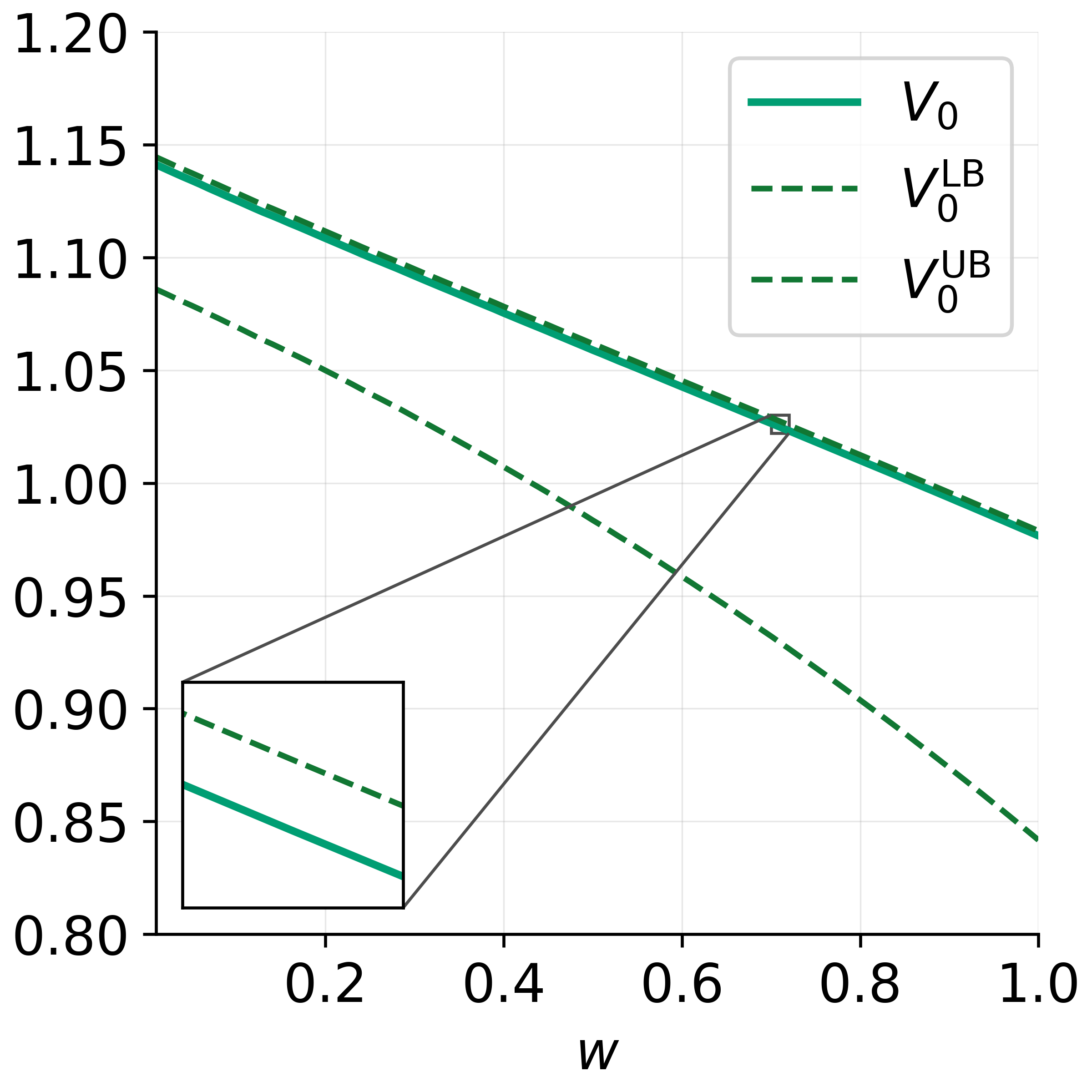}
		\caption{$\text{std}(X_1)=0.6$}
		\label{fig:thirdb5}
	\end{subfigure}
	\caption{Upper and lower bounds for $V_{0}^{w}$ in Figure \ref{fig4}}
	\label{fig6}
\end{figure}

Figures \ref{fig9} and \ref{fig10} depict the analogues of Figures \ref{fig3} and \ref{fig4} for the situation when the expected return for the risky asset is only 2\%. While the magnitudes of the quantities are very similar (visually almost indistinguishable to the case of larger $\mu$), one sees that the optimal weight $w^*$ and the critical threshold $\widehat{w}$ are considerably smaller with smaller expected invested return $\mu$, i.e. there is less incentive to invest a larger amount of the capital into the risky asset. 
\begin{figure}[h]
	\centering
	\begin{subfigure}{0.32\textwidth}
		\includegraphics[width=\textwidth]{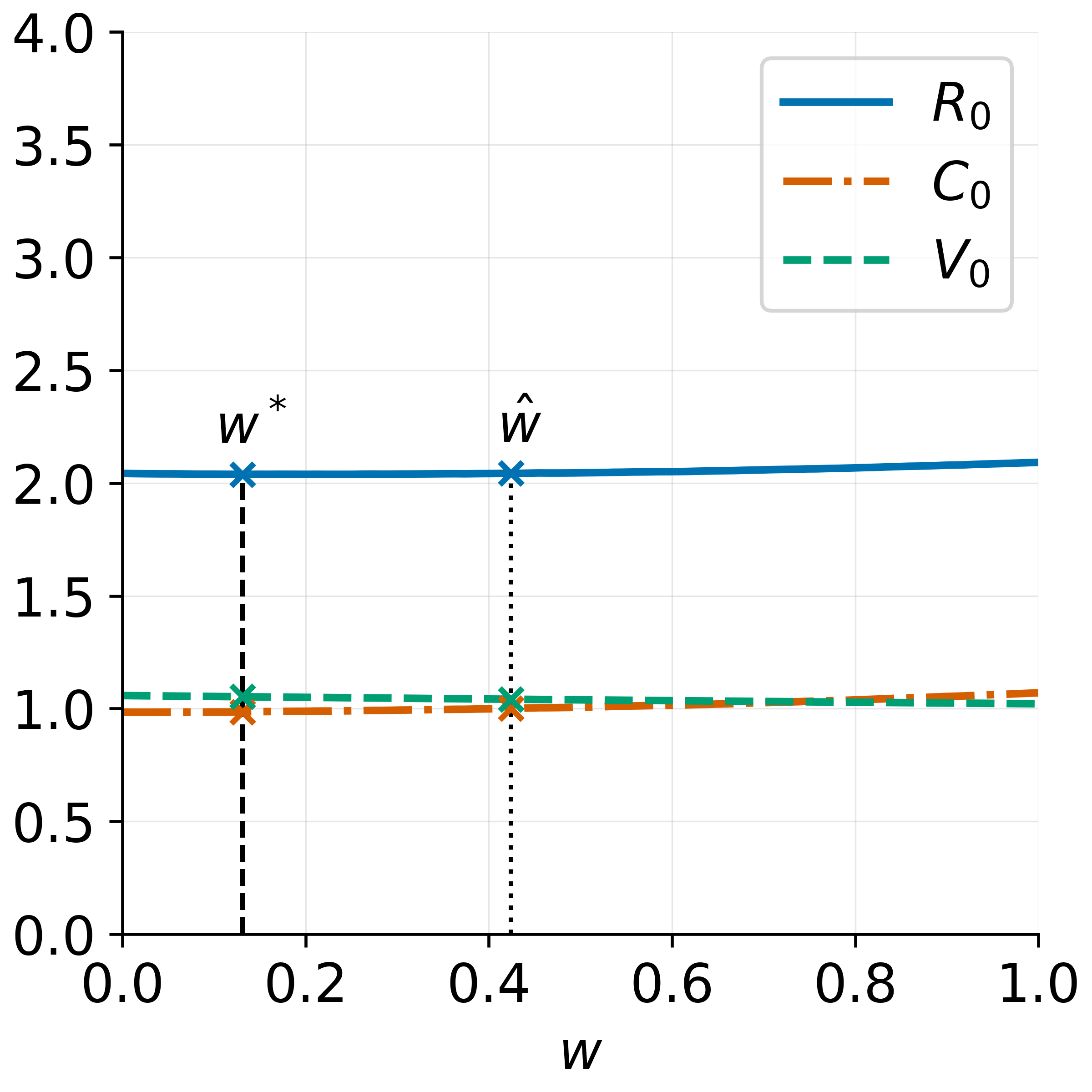}
		\caption{$\text{std}(S_1)=0.1$}
		\label{fig:firstd}
	\end{subfigure}
	\hfill
	\begin{subfigure}{0.32\textwidth}
		\includegraphics[width=\textwidth]{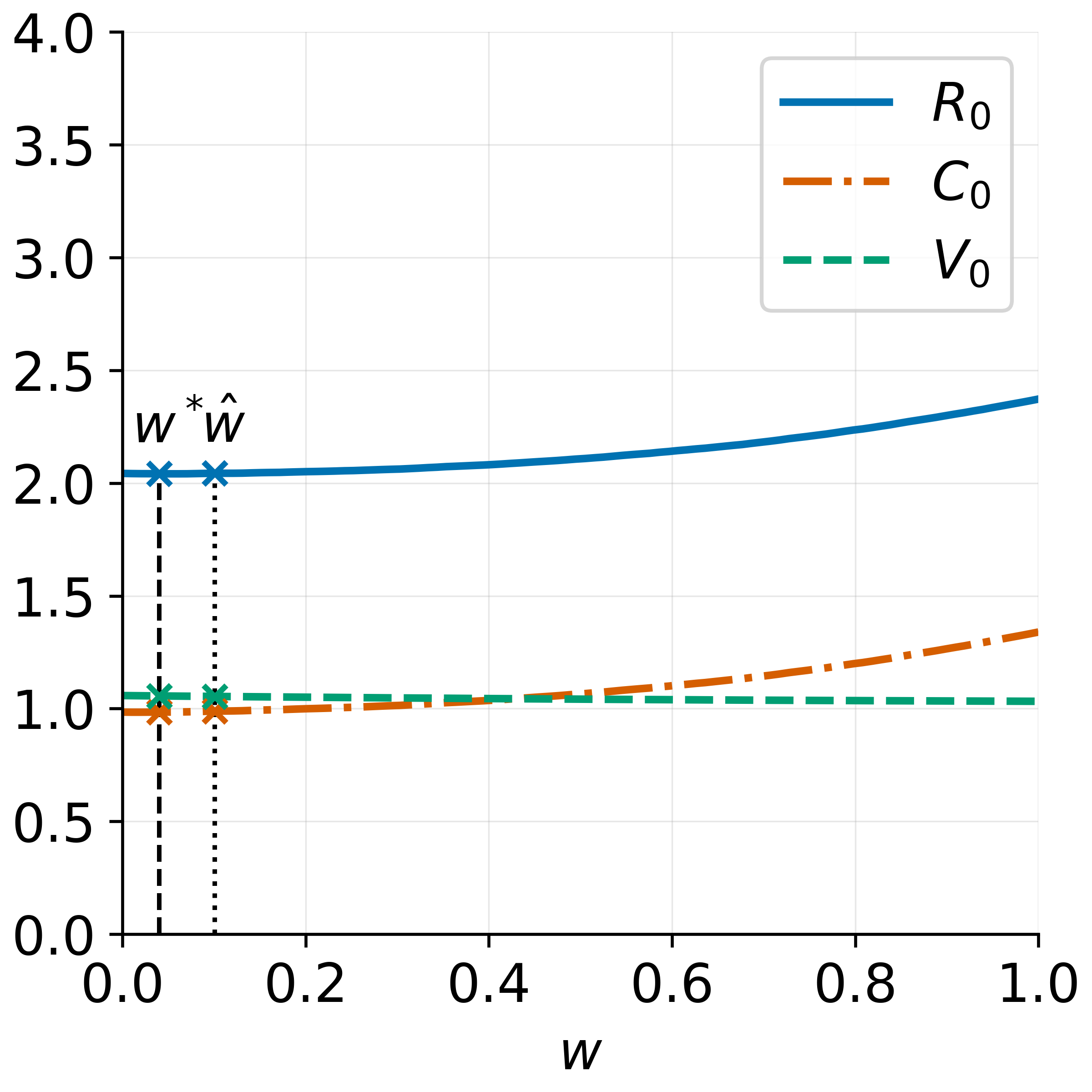}
		\caption{$\text{std}(S_1)=0.2$}
		\label{fig:secondd}
	\end{subfigure}
	\hfill
	\begin{subfigure}{0.32\textwidth}
		\includegraphics[width=\textwidth]{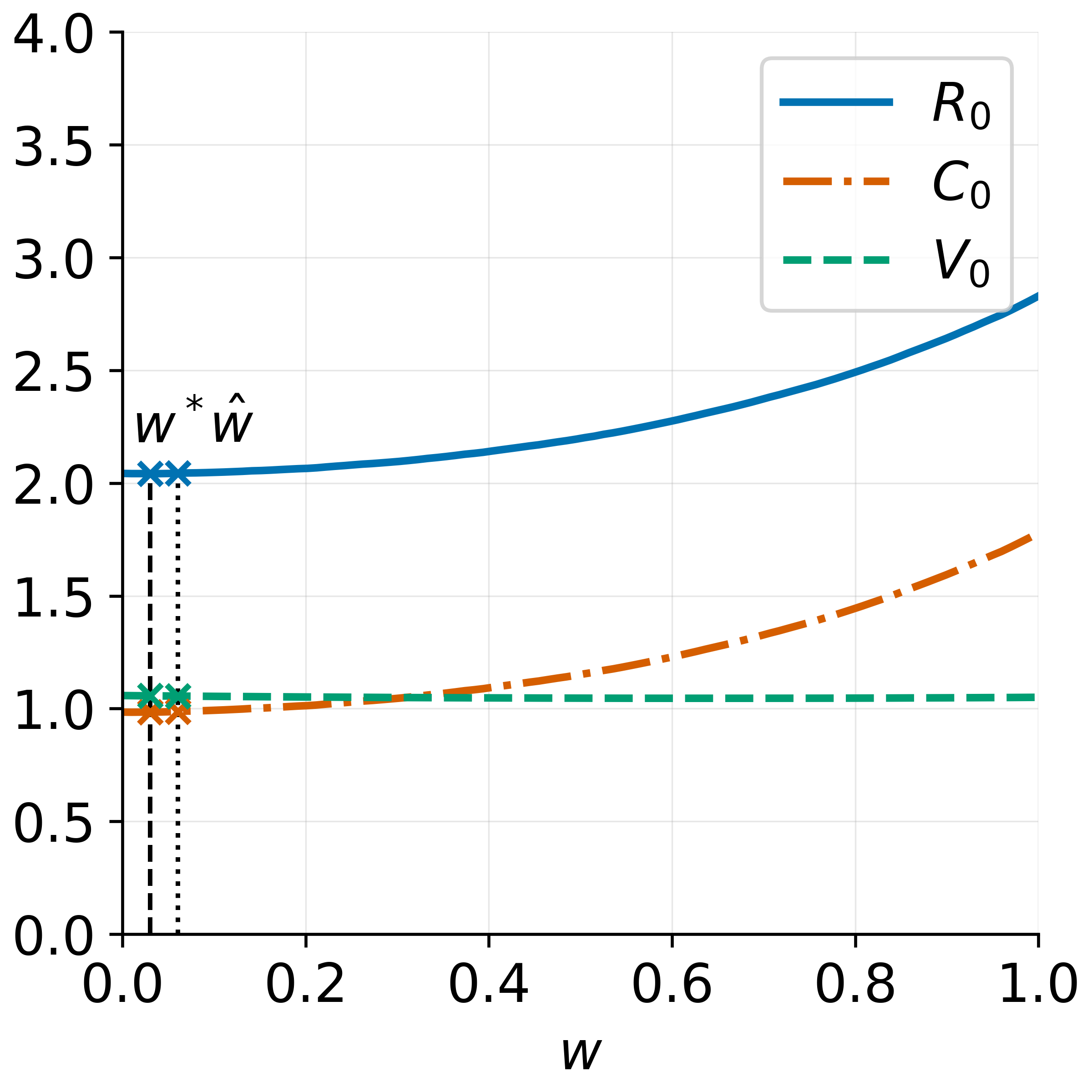}
		\caption{$\text{std}(S_1)=0.3$}
		\label{fig:thirdd}
	\end{subfigure}
	\caption{$R_{0}^{w}, C_{0}^{w}$ and $V_{0}^{w}$ for the lognormal model with $\E[S_1] = 1.02$, $\E[X_1] = 1$, $\text{std}(X)=0.3$, $\rho=\VaR_{0.005}$ and various values of $\text{std}(S_1)$}
	\label{fig9}
\end{figure}
\begin{figure}[h]
	\centering
	\begin{subfigure}{0.32\textwidth}
		\includegraphics[width=\textwidth]{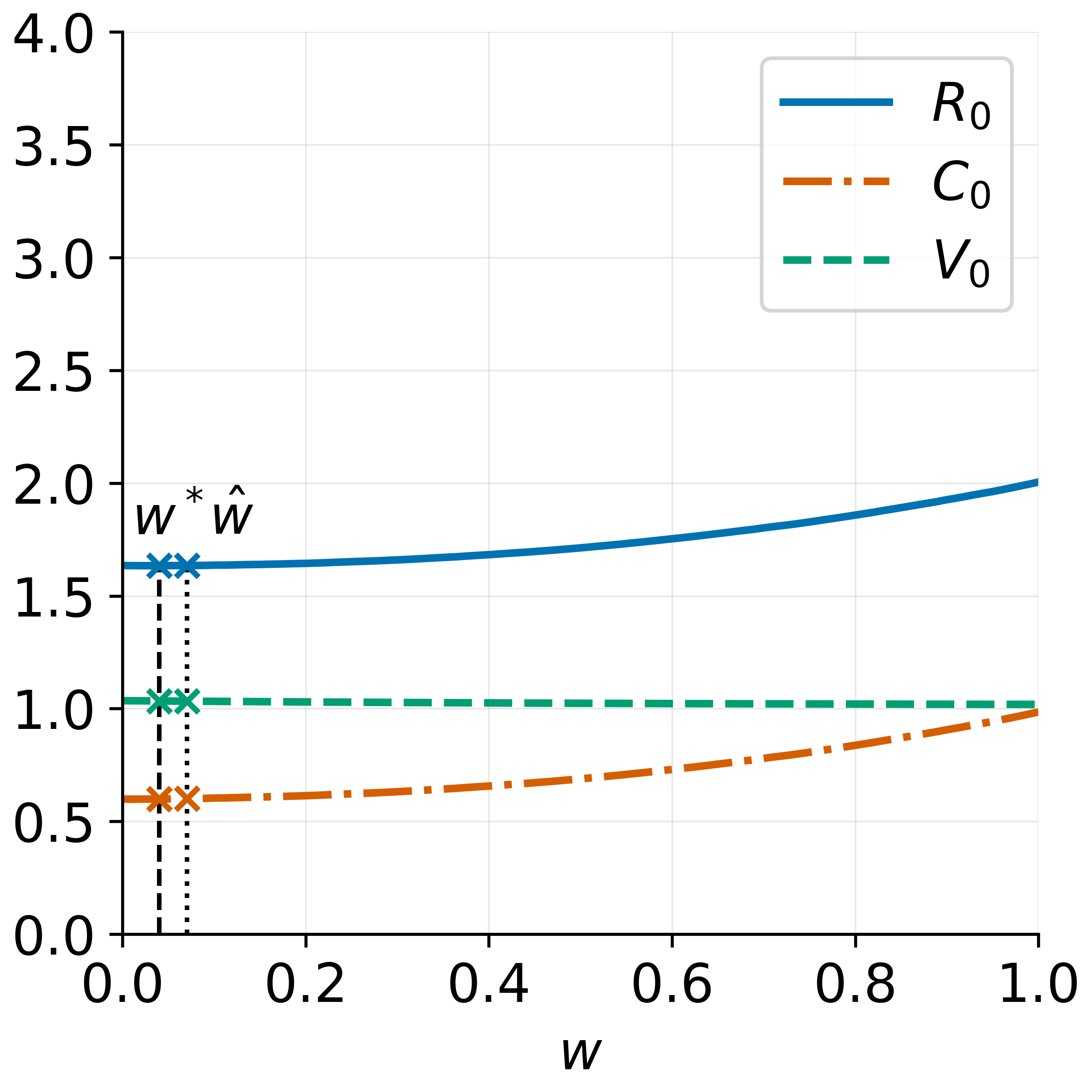}
		\caption{$\text{std}(X_1)=0.2$}
		\label{fig:firstd}
	\end{subfigure}
	\hfill
	\begin{subfigure}{0.32\textwidth}
		\includegraphics[width=\textwidth]{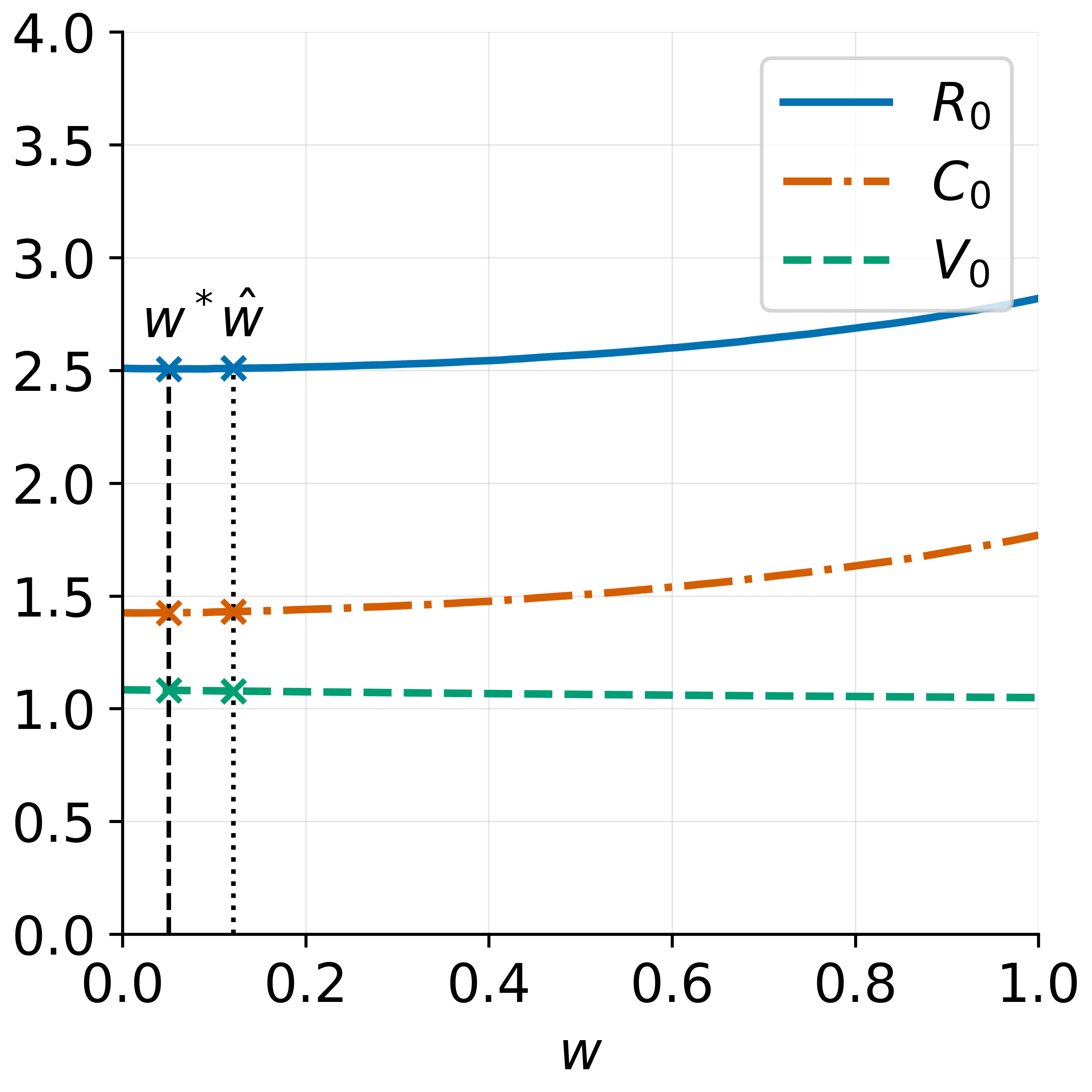}
		\caption{$\text{std}(X_1)=0.4$}
		\label{fig:secondd}
	\end{subfigure}
	\hfill
	\begin{subfigure}{0.32\textwidth}
		\includegraphics[width=\textwidth]{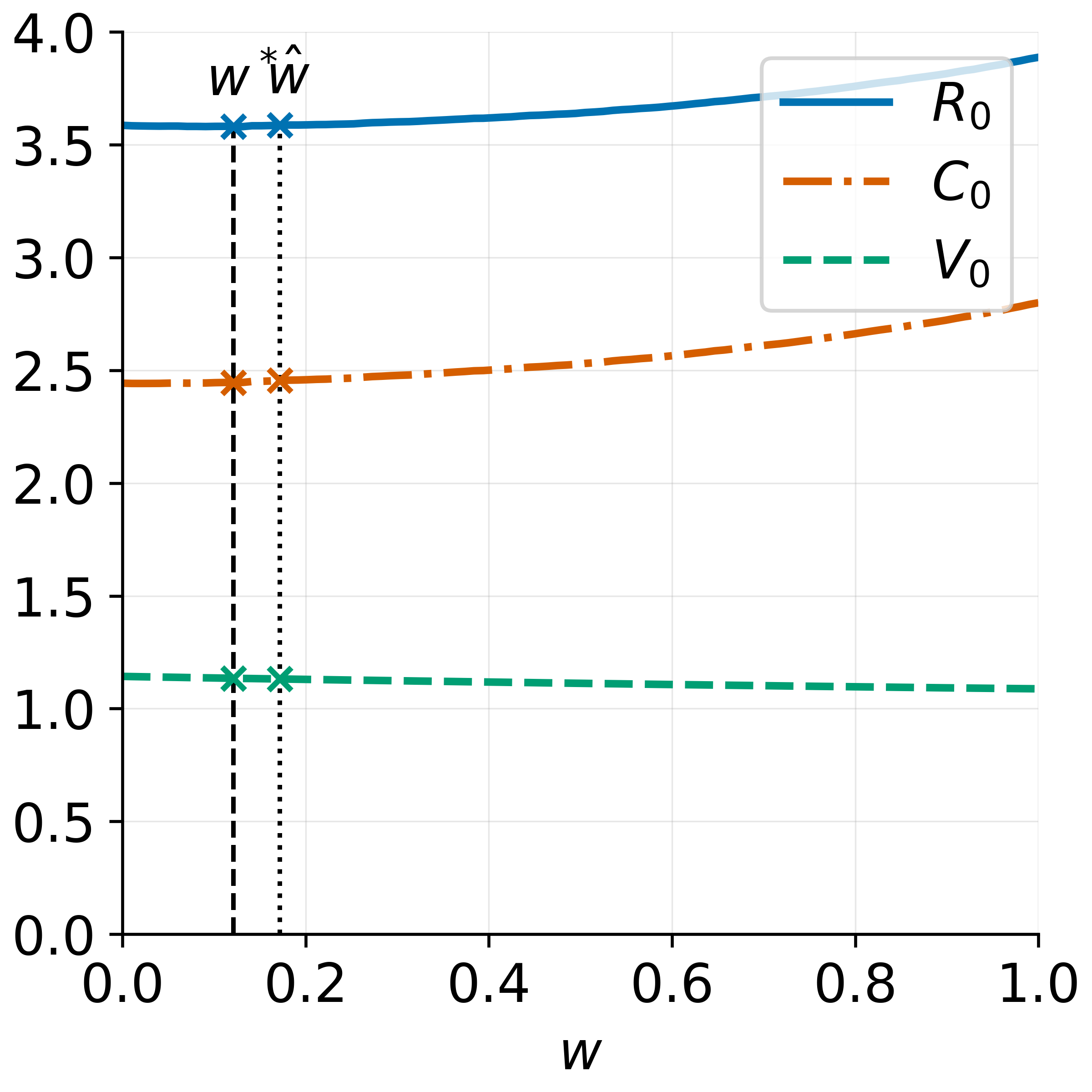}
		\caption{$\text{std}(X_1)=0.6$}
		\label{fig:thirdd}
	\end{subfigure}
	\caption{$R_{0}^{w}, C_{0}^{w}$ and $V_{0}^{w}$ for the lognormal model with $\E[S_1] = 1.02$, $\E[X_1] = 1$, $\text{std}(S_1)=0.2$, $\rho=\VaR_{0.005}$ and various values of $\text{std}(X_1)$}
	\label{fig10}
\end{figure}

\subsection{Lognormal investment returns and Pareto claims}
Figures \ref{fig11} and \ref{fig12} illustrate the possible effect of even heavier tails for the insurance risk, assuming that $X_1$ is Pareto-distributed with cumulative distribution function $F(x)=1-(x/x_m)^{-\beta}$, $x>x_m>0$, with the same first two moments as in the case of lognormal $X_1$ in Figures \ref{fig3} and \ref{fig4}: 
\begin{align*}
\E[X_1]=x_m\frac{\beta}{\beta-1}, \quad \text{std}(X_1)=x_m\frac{\sqrt{\beta}}{(\beta-1)\sqrt{\beta-2}}. 
\end{align*}
For example, the resulting Pareto parameters in Figure \ref{fig12}, where $\E[X_1]=1$, are given by 
\begin{align*}
\beta=1+\sqrt{1+1/\text{std}(X_1)^2}, \quad x_m=(\beta-1)/\beta.
\end{align*}
A comparison of Figure \ref{fig4} (lognormal model) and Figure \ref{fig12} (lognormal $S_1$ and Pareto-distributed $X_1$) 
shows that the main difference is that the heavier Pareto tail increases the overall capital requirement $R_0$, and that the increase in $R_0$ is solely financed by increased capital $C_0$ (the premiums $V_0$ stay virtually unchanged).
In order to see the effects of heavier tails even more significantly, in Figure \ref{fig15} the expected aggregate claim size is kept at $ \E[X_1] = 1$, but the parameter $\beta$ is taken more extreme ($\beta=2$ and $\beta=1.1$), in which case the standard deviation does not exist anymore (but the mean still does). 
One observes how the capital requirement $R_0$ goes up tremendously, but the theoretical premiums $V_0$ are still virtually unchanged.
It should be noted that whereas $R_0$ increases as $\beta$ decreases from $\beta=2$ to $\beta=1.1$, this is not the case for $V_0$ due to the increasing value of the shareholders limited liability option (cf.~Example \ref{ex:llo_pareto}). \\ 
Generally, the optimal investment weight $w^*$ and the critical threshold $\widehat{w}$ are larger in the Pareto case, indicating that for heavier-tailed insurance risks a riskier investment is advantageous (in the very heavy-tailed case of $\beta=1.1$ even to the extent that full investment in the risky asset, $w=1$, does not lead to a higher overall capital requirement $R_0$ than risk-less investment, $w=0$). 
Since the insurance risk dominates here, we gain from an attractive return on the risky investment without substantially affecting the overall capital requirement.  
One also observes that in comparison with lighter-tailed risks, the curves are generally flatter, i.e.\ they are less sensitive to changes in $w$. 

\begin{figure}[h]
	\centering
	\begin{subfigure}{0.32\textwidth}
		\includegraphics[width=\textwidth]{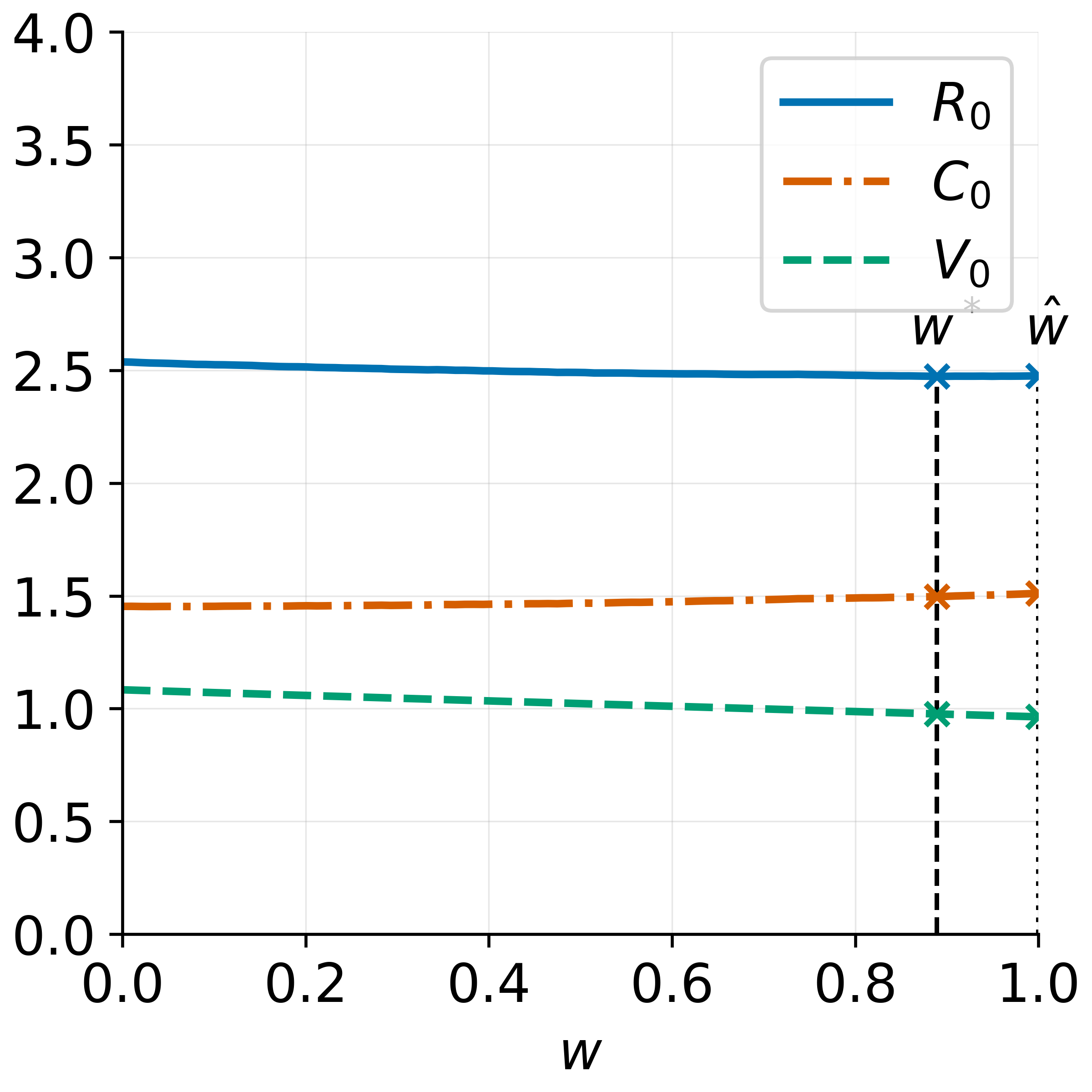}
		\caption{$\text{std}(S_1)=0.1$}
		\label{fig:firsta}
	\end{subfigure}
	\hfill
	\begin{subfigure}{0.32\textwidth}
		\includegraphics[width=\textwidth]{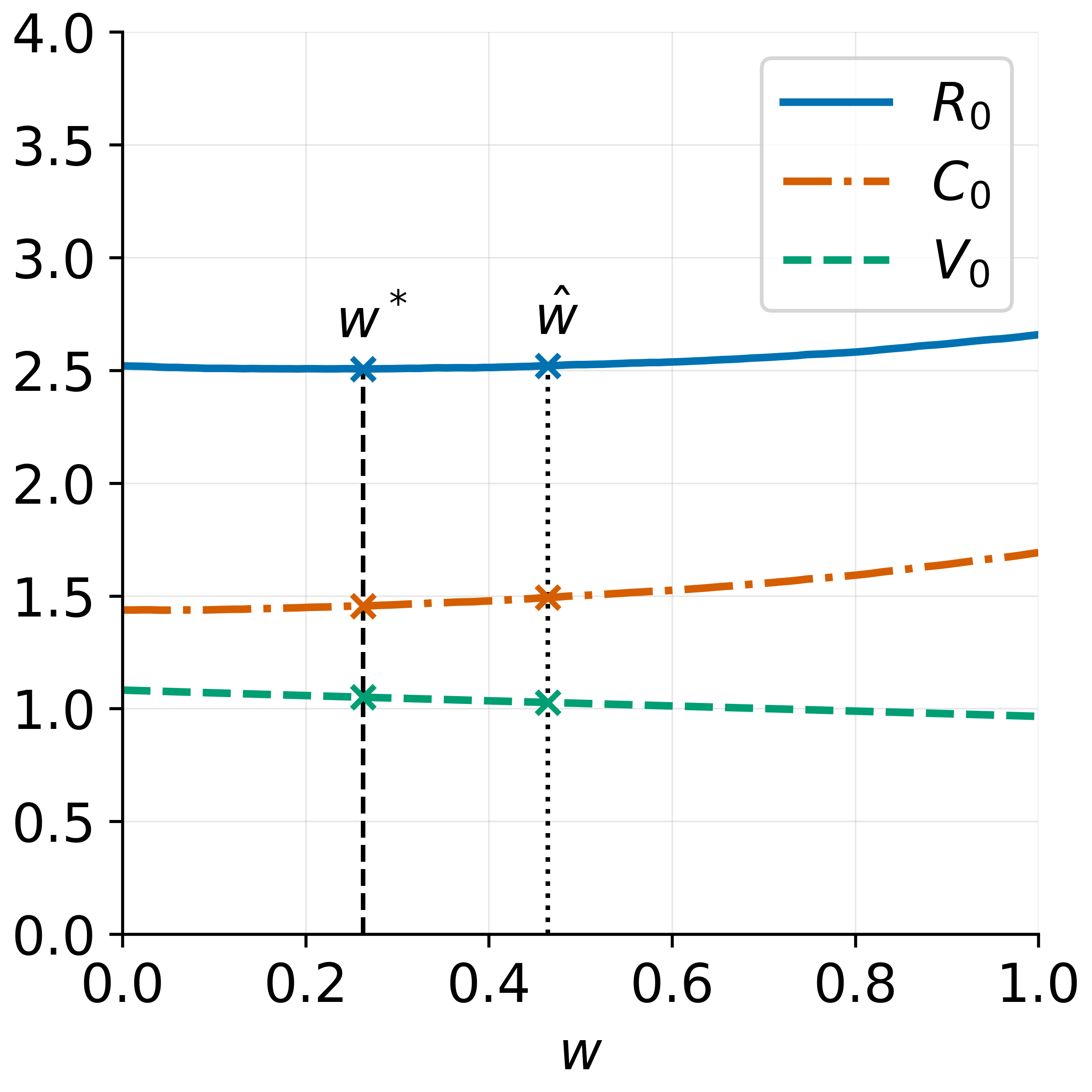}
		\caption{$\text{std}(S_1)=0.2$}
		\label{fig:seconda}
	\end{subfigure}
	\hfill
	\begin{subfigure}{0.32\textwidth}
		\includegraphics[width=\textwidth]{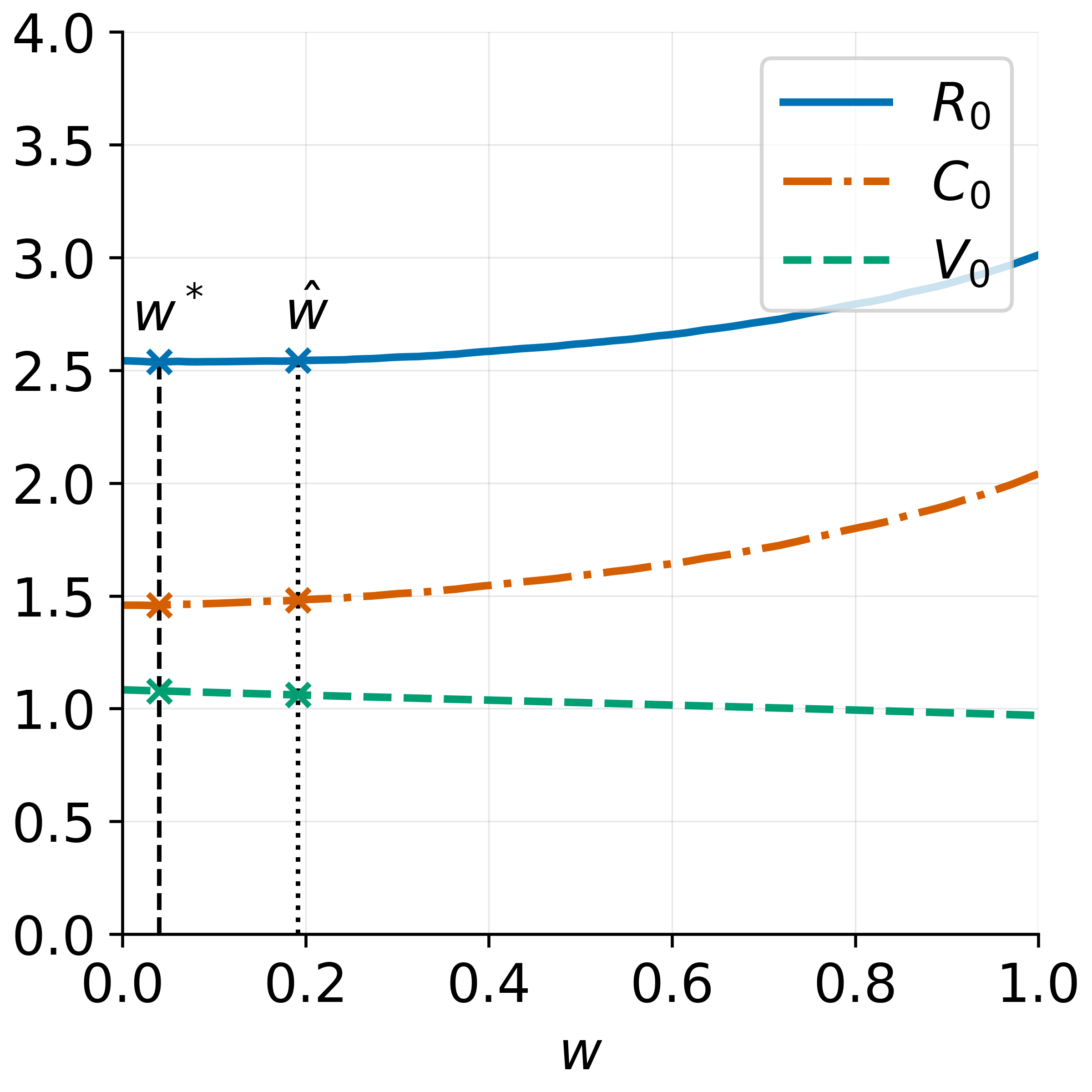}
		\caption{$\text{std}(S_1)=0.3$}
		\label{fig:thirda}
	\end{subfigure}
	\caption{$R_{0}^{w}, C_{0}^{w}$ and $V_{0}^{w}$ for lognormal $S_1$ and Pareto-distributed $X_1$ with $\E[X_1] = 1$, $\text{std}(X_1)=0.3$, $\E[S_1] = 1.05$, $\rho=\VaR_{0.005}$ and various values of $\text{std}(S_1)$}
	\label{fig11}
\end{figure}
\begin{figure}[h]
	\centering
	\begin{subfigure}{0.32\textwidth}
		\includegraphics[width=\textwidth]{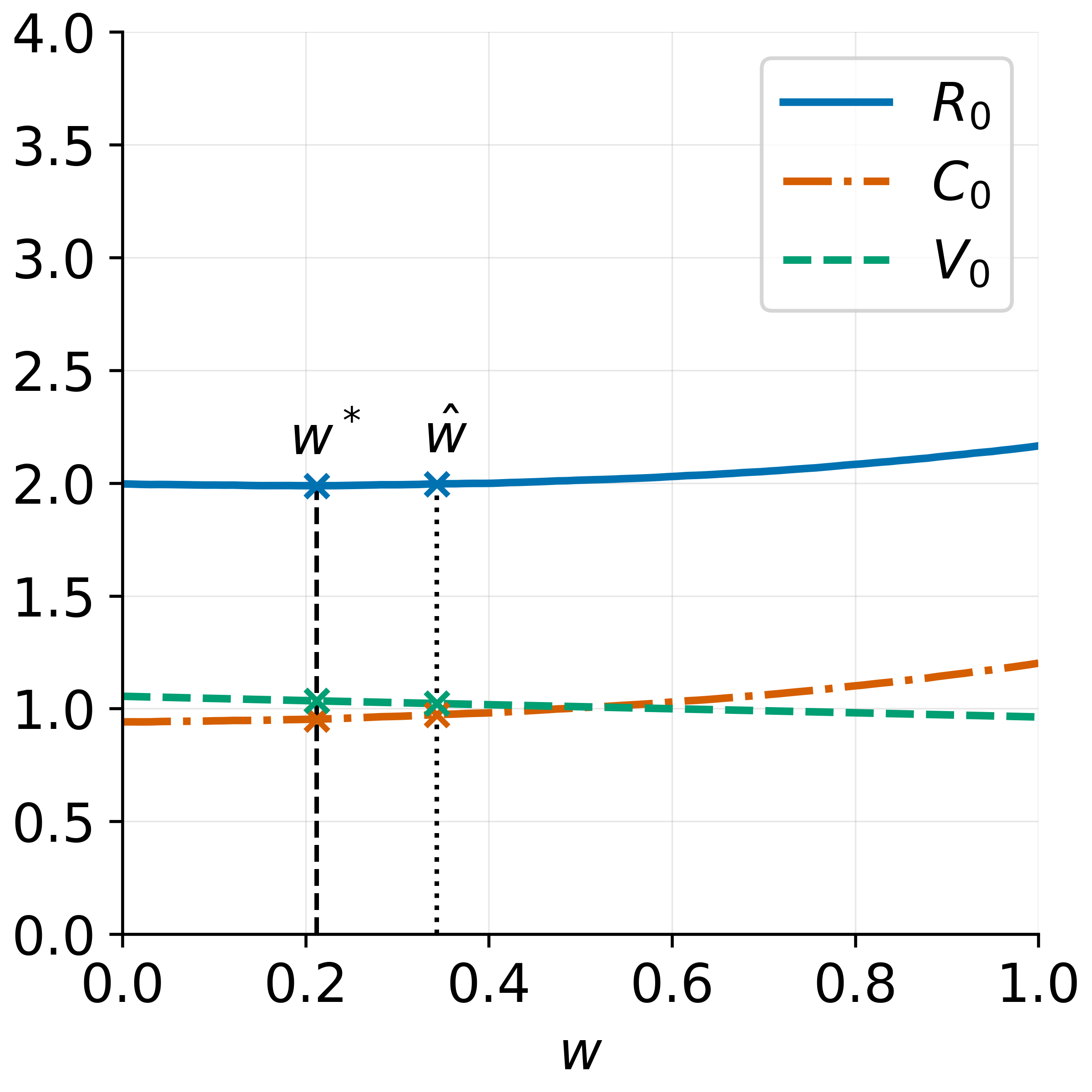}
		\caption{$\text{std}(X_1)=0.2$}
		\label{fig:firstb}
	\end{subfigure}
	\hfill
	\begin{subfigure}{0.32\textwidth}
		\includegraphics[width=\textwidth]{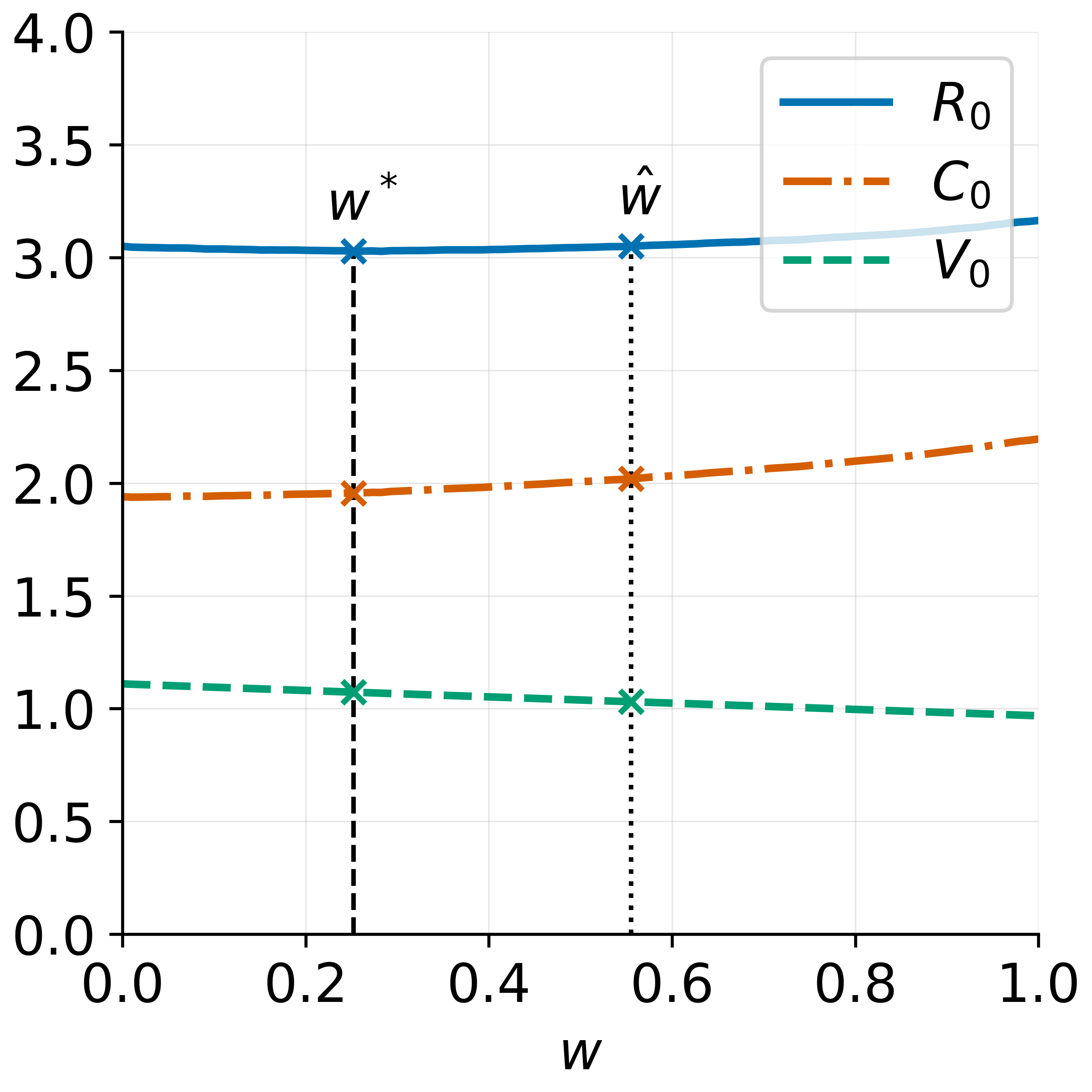}
		\caption{$\text{std}(X_1)=0.4$}
		\label{fig:secondb}
	\end{subfigure}
	\hfill
	\begin{subfigure}{0.32\textwidth}
		\includegraphics[width=\textwidth]{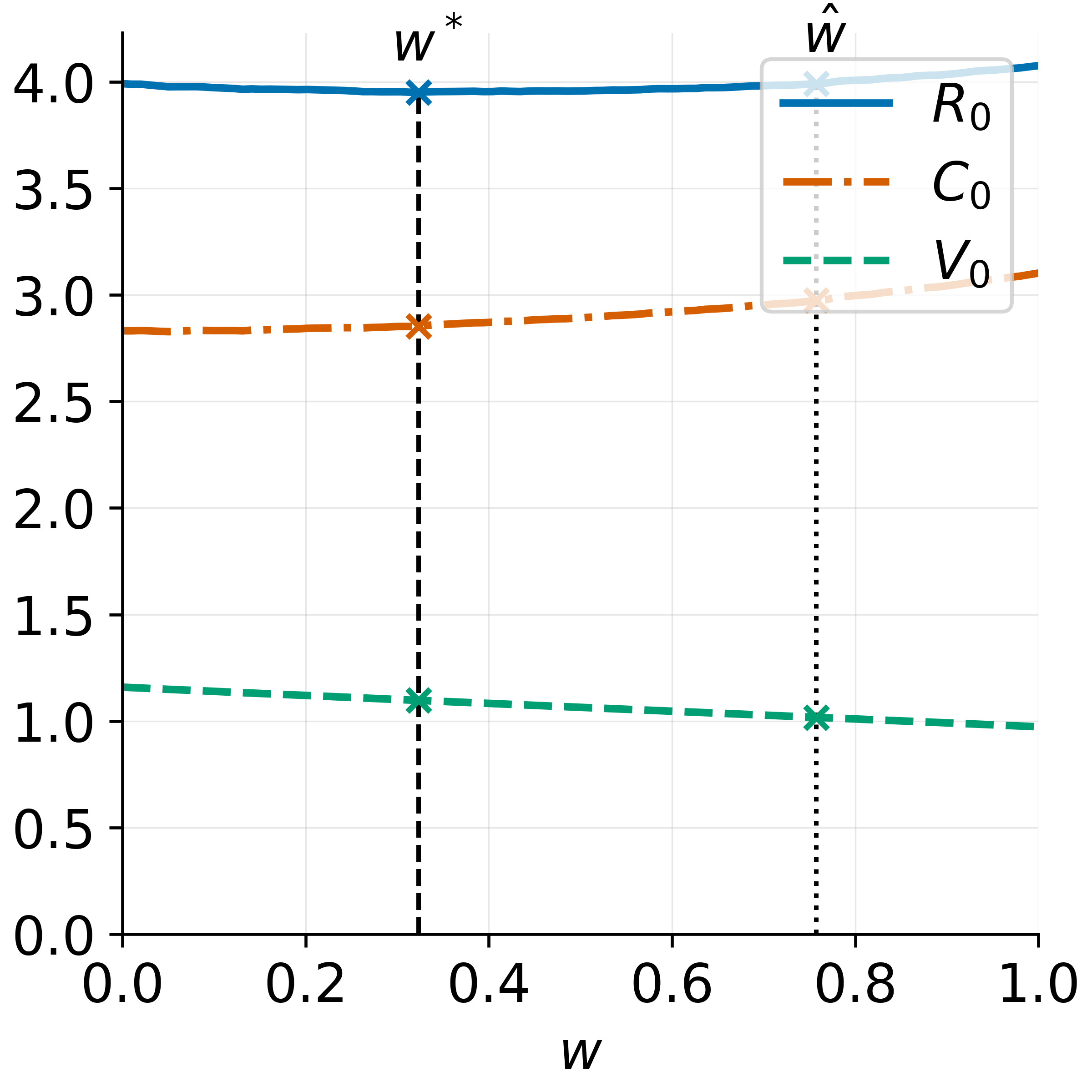}
		\caption{$\text{std}(X_1)=0.6$}
		\label{fig:thirdb}
	\end{subfigure}
	\caption{$R_{0}^{w}, C_{0}^{w}$ and $V_{0}^{w}$ for lognormal $S_1$ and Pareto-distributed $X_1$ with $\E[X_1] = 1$, $\E[S_1] = 1.05$, $\text{std}(S_1)=0.2$, $\rho=\VaR_{0.005}$ and various values of $\text{std}(X_1)$}
	\label{fig12}
\end{figure}

\begin{figure}[h]
	\centering
	\begin{subfigure}{0.32\textwidth}
		\includegraphics[width=\textwidth]{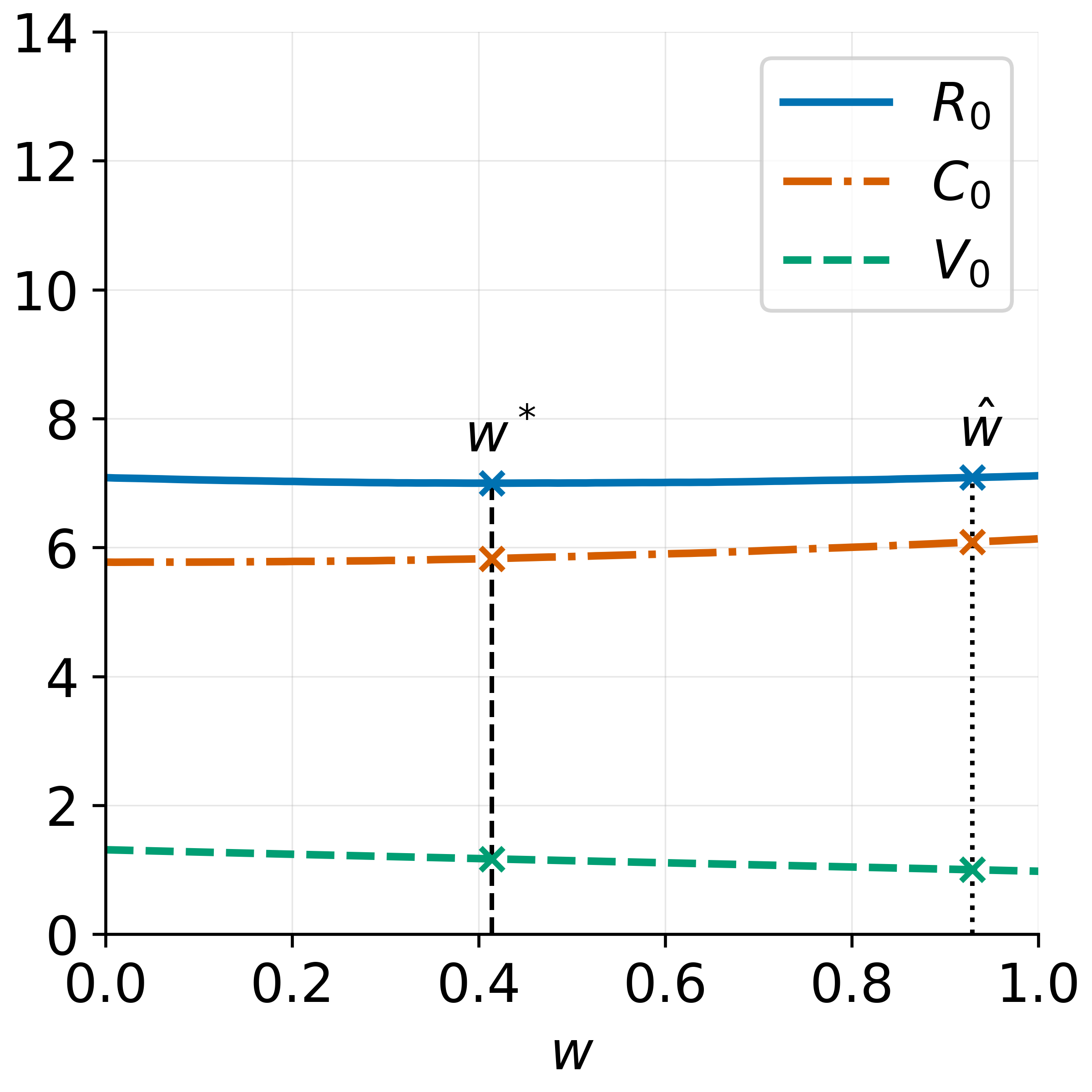}
		\caption{$\beta=2$}
		\label{fig:secondb}
	\end{subfigure}
	\hspace*{0.32cm}	\begin{subfigure}{0.32\textwidth}
		\includegraphics[width=\textwidth]{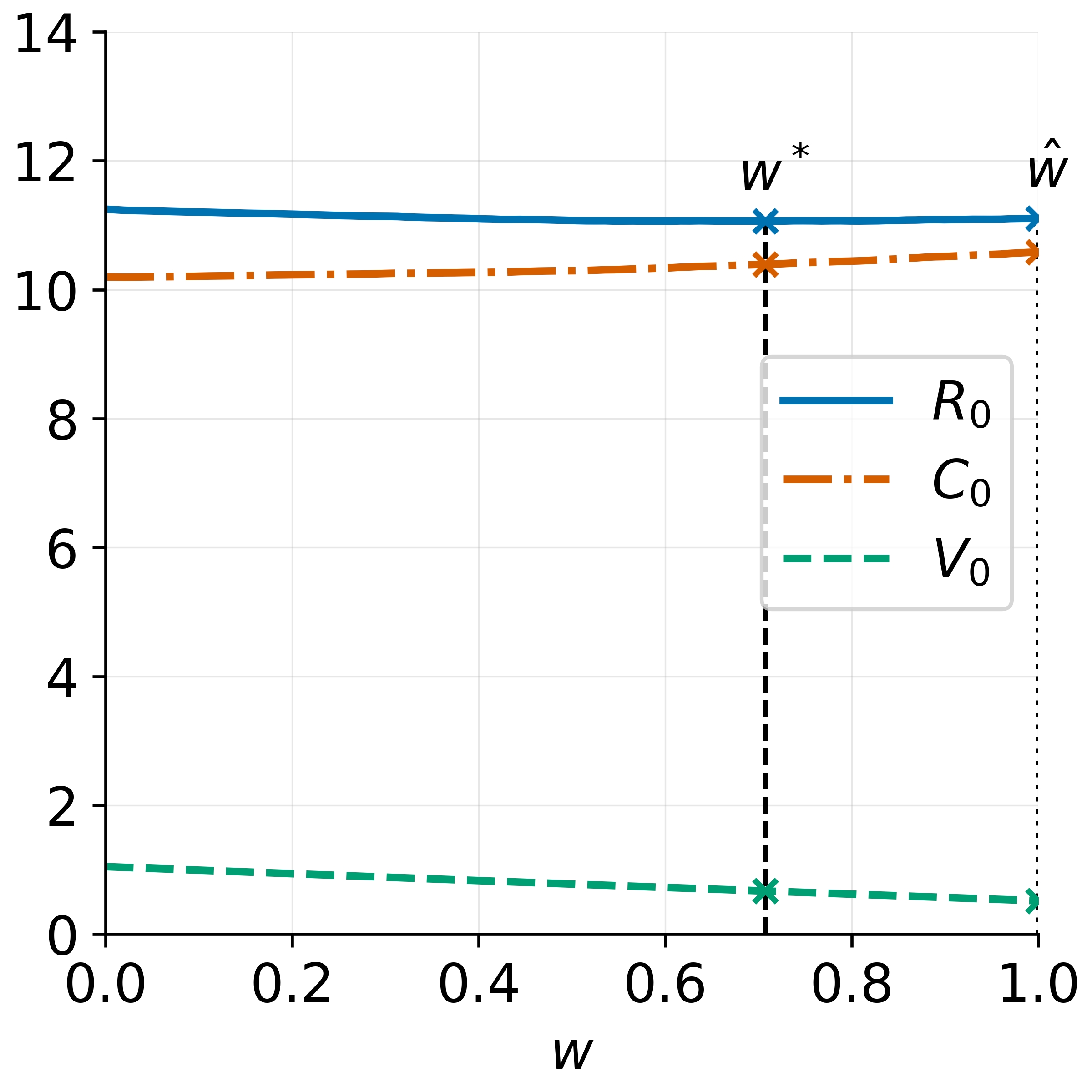}
		\caption{$\beta=1.1$}
		\label{fig:thirdb}
	\end{subfigure}
	\caption{$R_{0}^{w}, C_{0}^{w}$ and $V_{0}^{w}$ for lognormal $S_1$ and Pareto type I-distributed $X_1$ with $\E[X_1] = 1$, $\E[S_1] = 1.05$, $\text{std}(S_1) =0.2$, $\rho=\VaR_{0.005}$ and low levels of Pareto parameter $\beta$}
	\label{fig15}
\end{figure}

Finally, it may be of interest to see the effect of the risk measure on the result. While for the normal model the expected shortfall can be replaced by a value at risk with a different security level, this is not the case for the lognormal model. Figures \ref{fig8} and \ref{fig7b} depict the analogues of Figures \ref{fig3} and \ref{fig4} when $\VaR_{0.005}$ is replaced by $\ES_{0.01}$ (which is the risk measure in the Swiss Solvency Test). As expected, the overall capital requirement $R_0$ increases, but this increase is provided by a larger investment $C_0$ of the shareholders and the insurance premiums $V_0$ stay virtually unchanged. In this case, the optimal weights $w^*$ and the critical threshold $\widehat{w}$ are very similar to the ones obtained under the $\VaR_{0.005}$ criterion. \\

\begin{figure}[h]
	\centering
	\centering
	\begin{subfigure}{0.32\textwidth}
		\includegraphics[width=\textwidth]{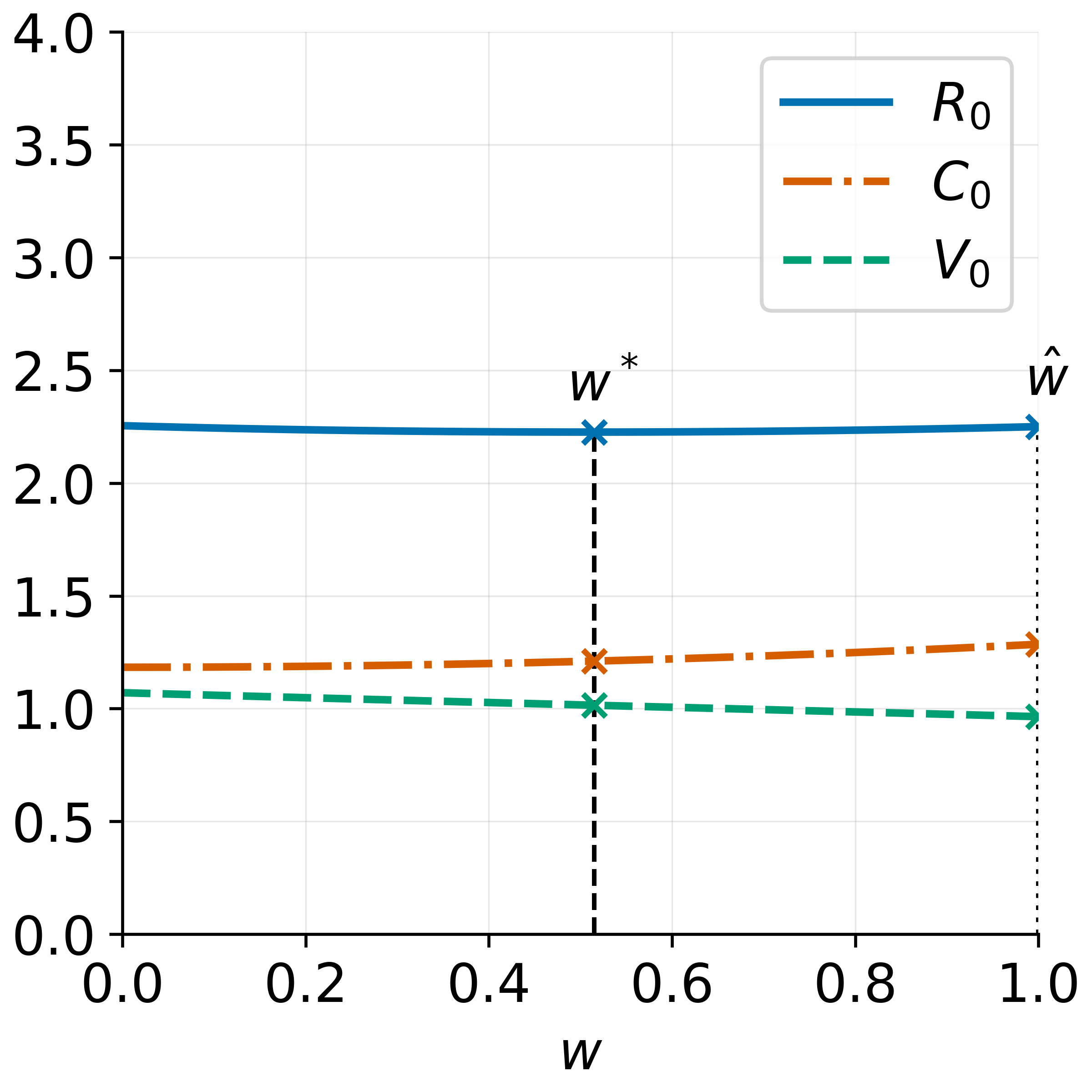}
		\caption{$\text{std}(S_1)=0.1$}
		\label{fig:firstc}
	\end{subfigure}
	\hfill
	\begin{subfigure}{0.32\textwidth}
		\includegraphics[width=\textwidth]{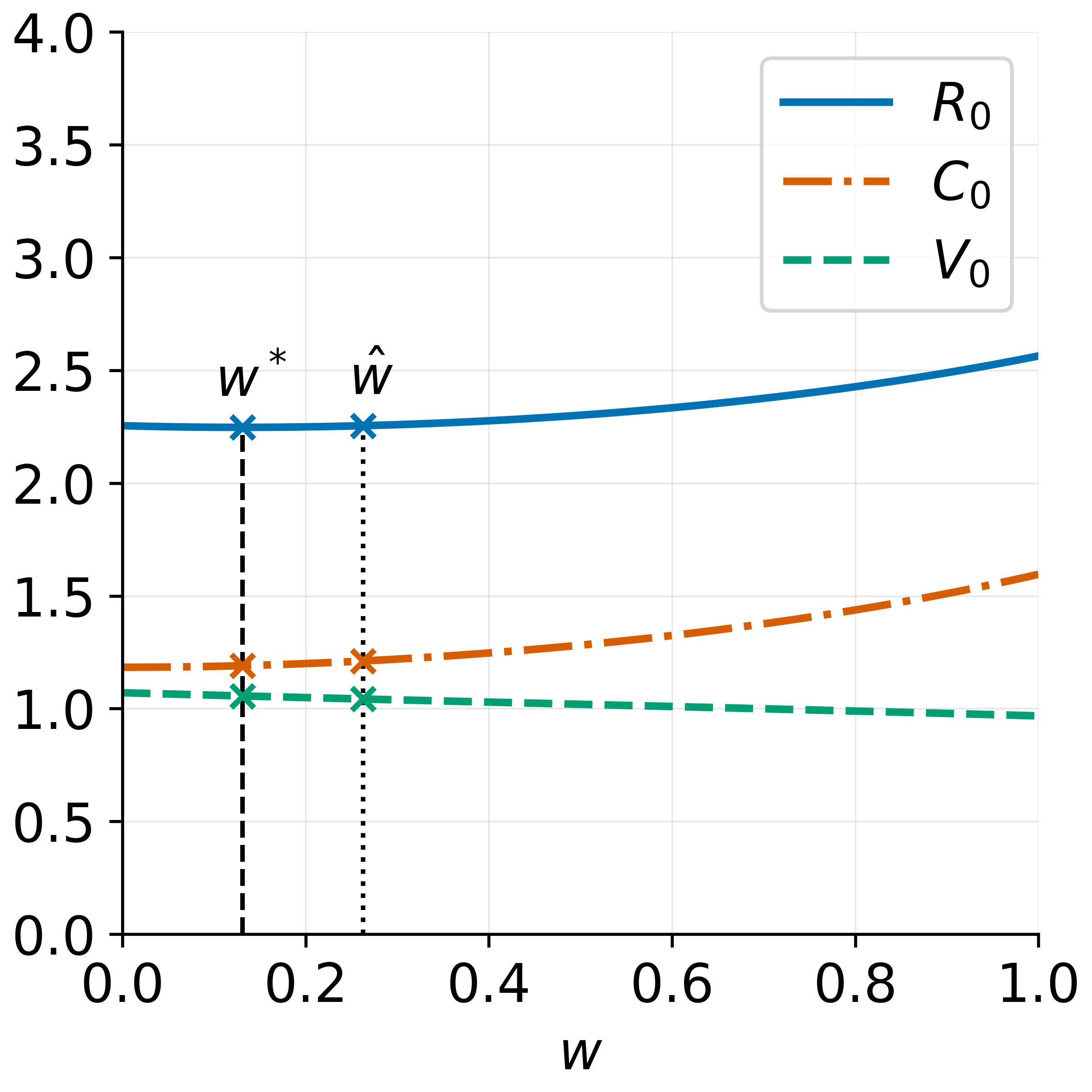}
		\caption{$\text{std}(S_1)=0.2$}
		\label{fig:secondc}
	\end{subfigure}
	\hfill
	\begin{subfigure}{0.32\textwidth}
		\includegraphics[width=\textwidth]{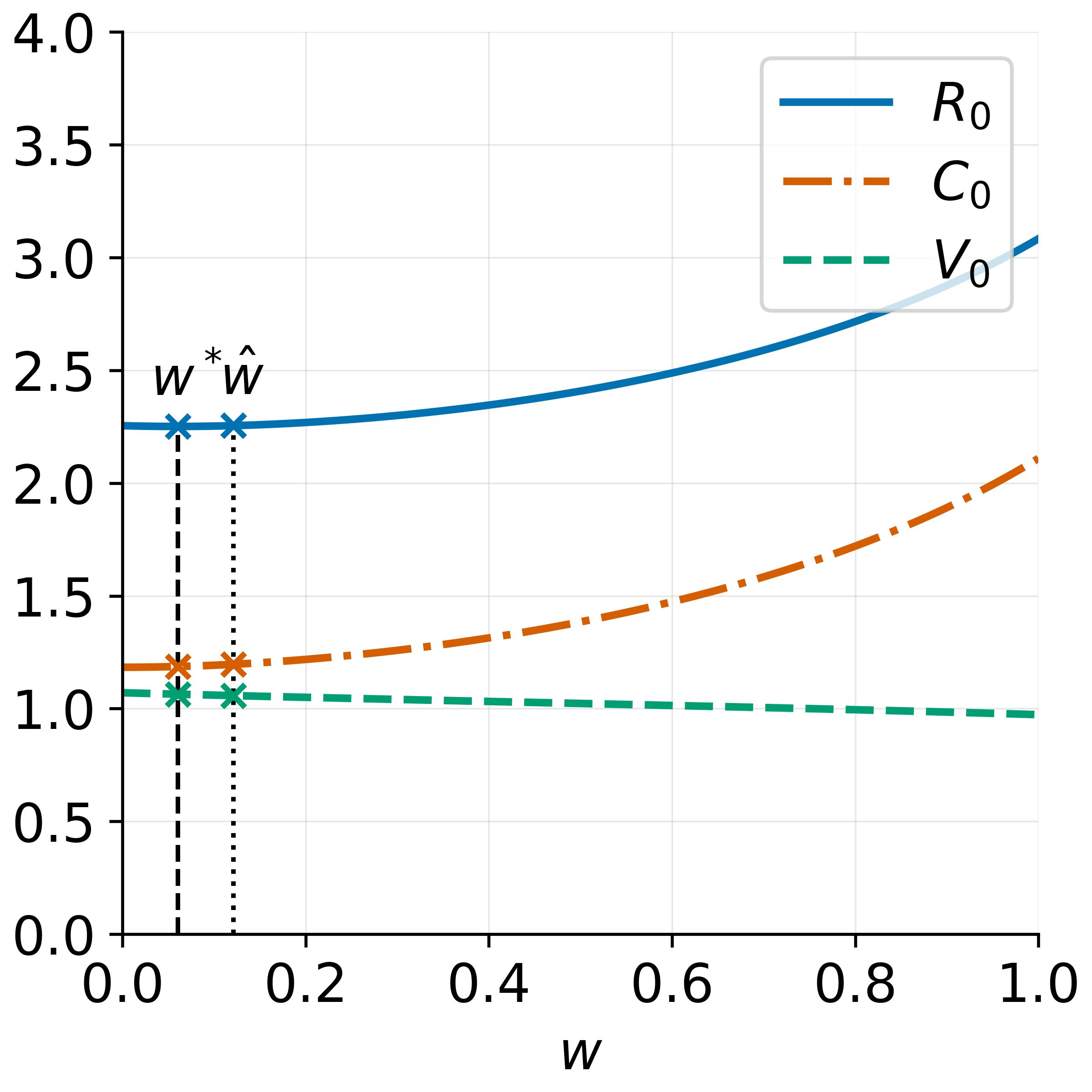}
		\caption{$\text{std}(S_1)=0.3$}
		\label{fig:thirdc}
	\end{subfigure}
	\caption{$R_{0}^{w}, C_{0}^{w}$ and $V_{0}^{w}$ for the lognormal model with $\E[S_1] = 1.05$, $\E[X_1] = 1$, $\text{std}(X_1)=0.3$, $\rho=\ES_{0.01}$ and various values of $\text{std}(S_1)$}
	\label{fig8}
\end{figure}
	
	\begin{figure}[h]
	\centering
	\begin{subfigure}{0.32\textwidth}
		\includegraphics[width=\textwidth]{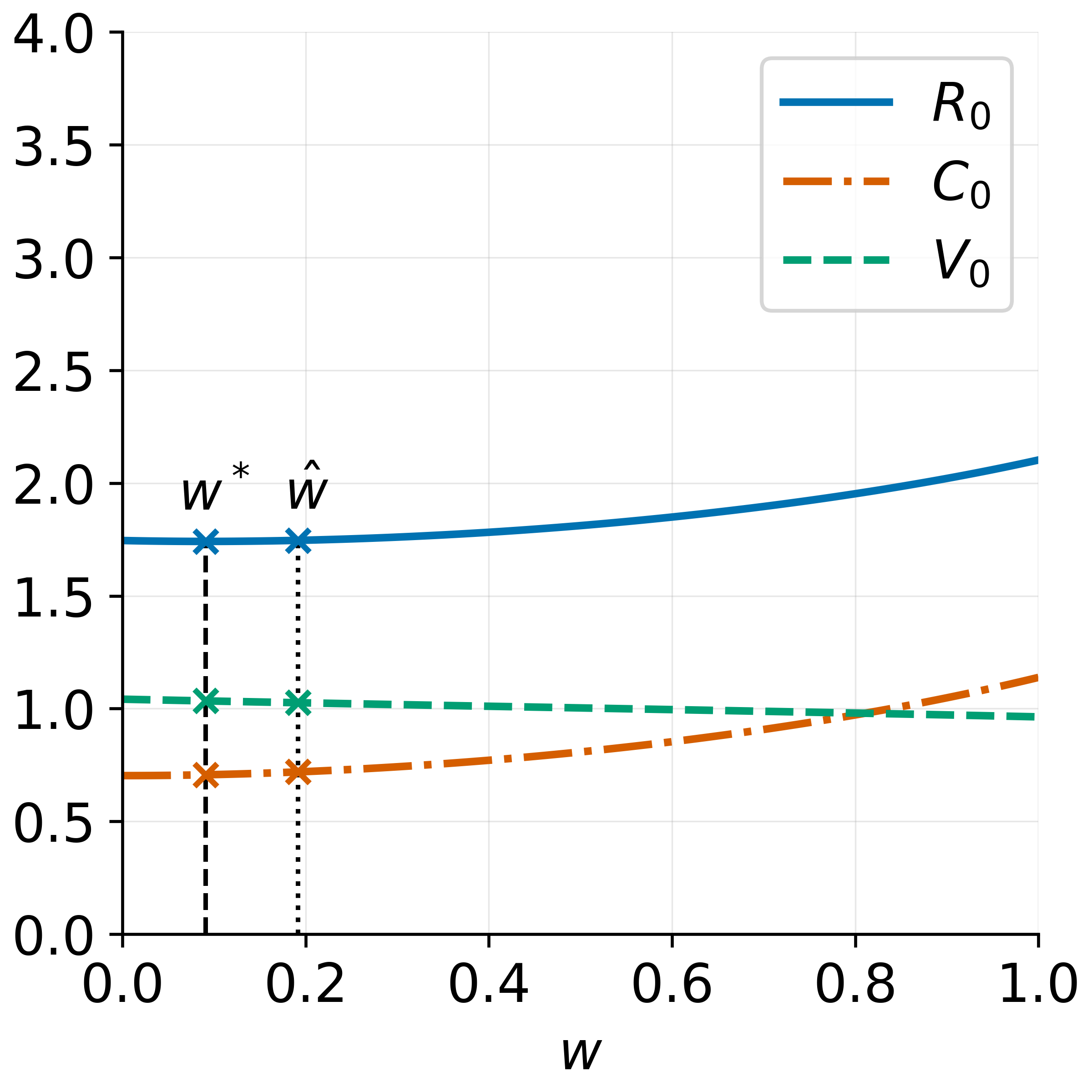}
		\caption{$\text{std}(X_1)=0.2$}
		\label{fig:firstb}
	\end{subfigure}
	\hfill
	\begin{subfigure}{0.32\textwidth}
		\includegraphics[width=\textwidth]{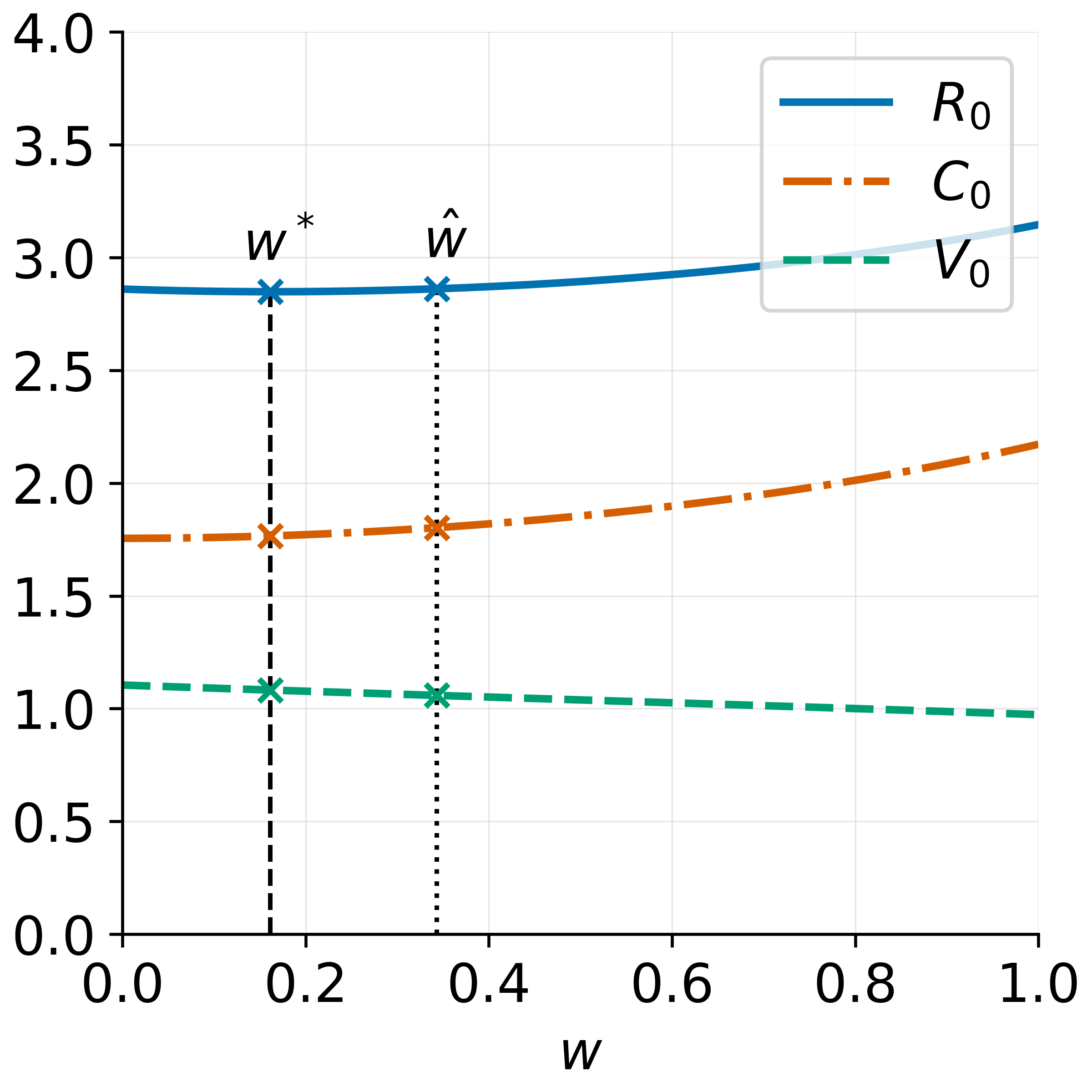}
		\caption{$\text{std}(X_1)=0.4$}
		\label{fig:secondb}
	\end{subfigure}
	\hfill
	\begin{subfigure}{0.32\textwidth}
		\includegraphics[width=\textwidth]{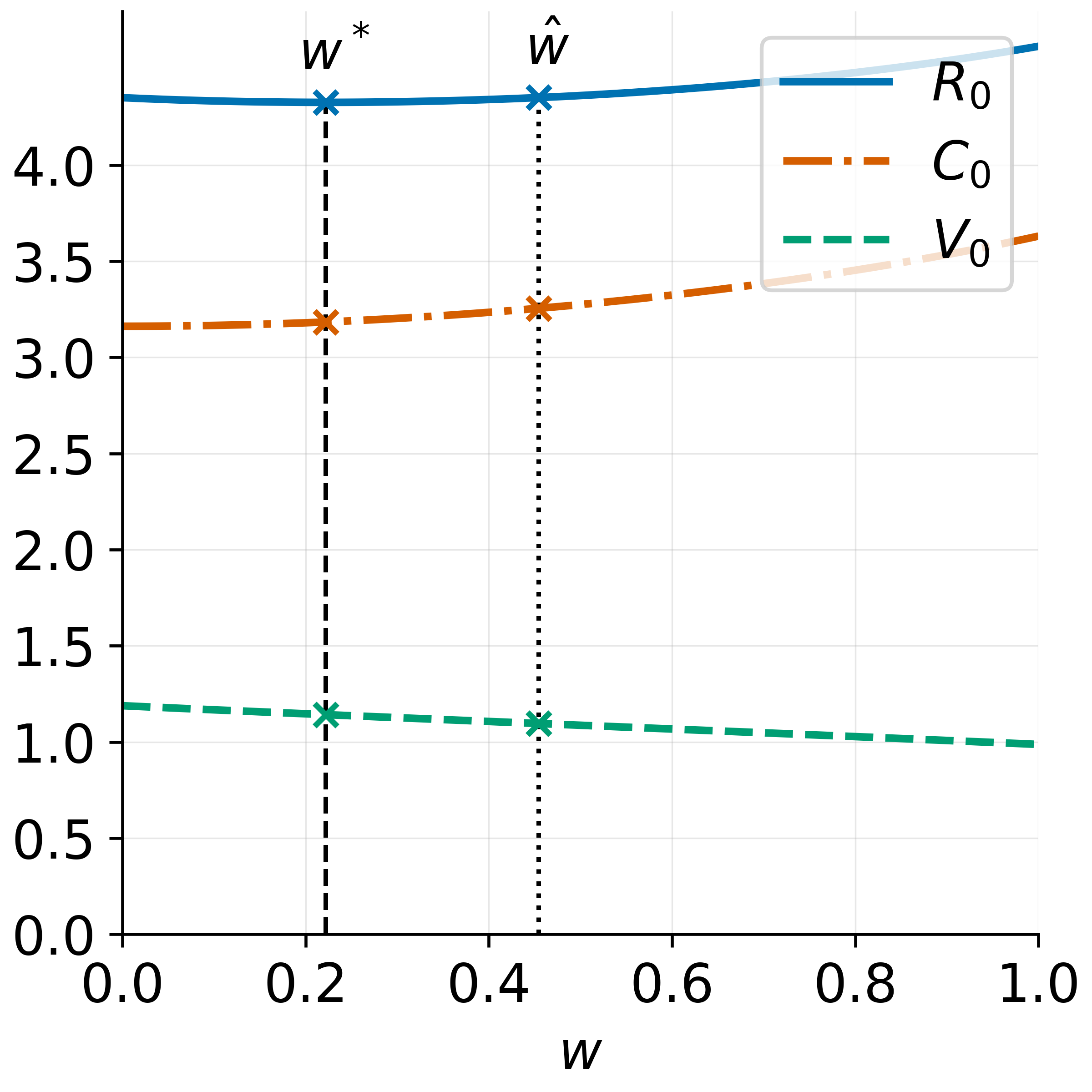}
		\caption{$\text{std}(X_1)=0.6$}
		\label{fig:thirdb}
	\end{subfigure}
	\caption{$R_{0}^{w}, C_{0}^{w}$ and $V_{0}^{w}$ for the lognormal model with $\E[S_1] = 1.05$, $\E[X_1]  = 1$, $\text{std}(S_1)=0.2$, $\rho=\ES_{0.01}$ and various values of $\text{std}(X_1)$}
	\label{fig7b}
\end{figure}

All these numerical examples illustrate that the appropriate blending of (independent) financial risk into the management of insurance risk can drive down the necessary insurance premiums according to actuarial considerations, which may be considered as a promising insight in situations where insurance premiums are considered too high. At the same time, the absolute reduction of $V_0$ through increasing $w$ is still relatively moderate.

\section{Conclusions}\label{sec6}
In this paper we showed, for simple yet reasonably realistic model assumptions, that a partial investment of the solvency capital into risky assets with higher return but also higher volatility can be beneficial for both policyholders and shareholders of a non-life insurance company. While the fact itself is well-known and also exploited in insurance practice, the present analysis allows to assess this effect quantitatively and quite explicitly from a viewpoint that seems not to have been pursued in the literature before. For the concrete static one-period model and under certain assumptions, we established a number of monotonicity results for needed premiums and solvency capital requirements. For the case of normally distributed insurance and asset risks it was possible to even derive explicit formulas for these quantities as well as the limit weight $\widehat{w}$ until which risky investment is a mutual advantage.  
We also illustrated the results numerically for parameters that are motivated by real-life magnitudes. For heavy-tailed risks, it was possible to quantify the increasing importance of the limited liability option of the shareholders, which can translate into reduced premium requirements of policyholders. 
Also, when the effect of heavy-tailed insurance risk on the overall capital requirement dominates the effect of risky investments, then higher proportions of the risky investment provide an additional pooling effect that supercedes the additionally introduced risk and are typically in favor of the policyholders.\\   
The results of this paper may be generalized in various directions. 
For instance, it could be interesting to investigate the sensitivity of the resulting quantities with respect to model uncertainty, both on the dynamics of the financial market and the distribution of insurance claims. Furthermore, although the independence assumption between insurance and financial risk is not an uncommon assumption in practice for certain non-life business lines, in a future study we intend to investigate to what extent dependence may compromise this trade-off of introducing additional risk for the prospect of higher investment returns. 
Also, to extend the analysis to a multi-period framework including explicit runoffs in long-tailed business would be of interest (see e.g.\ \cite{salzmann2010cost,Moehr2011}, possibly also under consideration of a changing realized cost-of-capital rate, cf.\ \cite{AlDa}). Such extensions are left for future research.

\appendix

\section{Proofs}\label{app1}

\begin{proof}[Proof of Proposition \ref{prop1}]
	If $w=0$, then the solution to $\rho(R_0Z_1-X_1)=0$ is $R_0=\rho(-X_1)\geq 0$. Therefore, it is sufficient to only consider $w>0$. 
	Note that $\VaR_{\alpha}(rZ_1-X_1)=F_{X_1-rZ_1}^{-1}(1-\alpha)$ and 
	\begin{align*}
		F_{X_1-rZ_1}(y)=\int_{0}^{\infty}\P(wS_1+1-w\geq (x-y)/r)dF_{X_1}(x), \quad r>0.
	\end{align*}
	Note further that $y\mapsto F_{X_1-rZ_1}(y)$ is continuous and strictly increasing since 
	\begin{align*}
		\frac{\partial}{\partial y}F_{X_1-rZ_1}(y)=\int_{y+r(w-1)}^{\infty}f_{wS_1}\bigg(\frac{x-y}{r}+w-1\bigg)\frac{1}{r}dF_{X_1}(x)>0.
	\end{align*}
	Hence, $\VaR_{\alpha}(rZ_1-X_1)=0$ is equivalent to $F_{X_1-rZ_1}(0)=1-\alpha$. 
	Note that 
	\begin{align*}
		\frac{\partial}{\partial r}F_{X_1-rZ_1}(y)=\int_{y+r(w-1)}^{\infty}f_{wS_1}\bigg(\frac{x-y}{r}+w-1\bigg)\frac{x-y}{r^2}dF_{X_1}(x)>0.
	\end{align*}
	Since $r\mapsto g(r):=F_{X_1-rZ_1}(0)$ is continuous and strictly increasing, $g(0)=F_{X_1}(0)<1-\alpha$ and $g(r)\uparrow 1$ as $r\to\infty$, there exists a unique $r>0$ solving $F_{X_1-rZ_1}(0)=1-\alpha$. 
	
	Write $(r,y)\mapsto g(r,y):=F_{X_1-rZ_1}(y)-1+\beta$. From the equation $g(r,y)=0$ it follows from the implicit function theorem and the partial derivatives above that the existence of a solution $(r^*,y^*)$ means that $y$ can be expressed as a continuous strictly decreasing function of $r$ in a neighborhood of $r^*$, i.e.~$\partial y/\partial r=-(\partial g/\partial r)/(\partial g/\partial y)<0$ in a neighborhood of $(r^*,y^*)$.  
	Hence, $r\mapsto \VaR_{\beta}(rZ_1-X_1)$ is continuous and strictly decreasing. Hence, also $r\mapsto \ES_{\alpha}(rZ_1-X_1)$ is continuous and strictly decreasing. Since, $\ES_{\alpha}(-X_1)>0$ and $\ES_{\alpha}(rZ_1-X_1)\downarrow -\infty$ as $r\to\infty$, there exists a unique $r>0$ such that $\ES_{\alpha}(rZ_1-X_1)=0$. 
\end{proof}

\begin{proof}[Proof of Proposition \ref{prop23}]
The monotonicity property of $\rho$ ($\VaR_{\alpha}$ or $\ES_{\alpha}$) means that if $\widetilde{Y}\geq_{\text{st}} Y$, then $\rho(\widetilde{Y})\leq \rho(Y)$ (see e.g.~Remark 4.58 in F\"ollmer \& Schied \cite{FS4} for the fact that $\rho$ is monotone with respect to first-order stochastic dominance, and note that this order is invariant under convolution). 

(i): Suppose that $\widetilde{X}_1\geq_{\text{st}} X_1$. Let $\widetilde{R}_0$ be the unique solution to $\rho(\widetilde{R}_0Z_1-\widetilde{X}_1)=0$. 
Suppose that $\widetilde{R}_0<R_0$. We will see that this assumption leads to a contradiction. 
By the monotonicity property of $\VaR_{\alpha}$ and $\ES_{\alpha}$ and the independence of $X_1$ and $\widetilde{X}_1$ of $Z_1$, 
\begin{align}\label{eq:proof_x1}
0=\rho(\widetilde{R}_0Z_1-\widetilde{X}_1)\geq \rho(\widetilde{R}_0Z_1-X_1).
\end{align}
Because of \eqref{eq:proof_x1} and uniqueness of the solution $r$ to $\rho(rZ_1-X_1)=0$,  
\begin{align}\label{eq:proof_x2}
\rho(\widetilde{R}_0Z_1-X_1)<0. 
\end{align}
Because of monotonicity of $\rho$ and the assumption $\widetilde{R}_0<R_0$, 
\begin{align}\label{eq:proof_x3}
\rho(\widetilde{R}_0Z_1-X_1)\geq \rho(R_0Z_1-X_1)=0. 
\end{align}
Since \eqref{eq:proof_x2} and \eqref{eq:proof_x3} are contradictory we conclude that the assumption $\widetilde{R}_0<R_0$ is false. Hence, $\widetilde{R}_0\geq R_0$ which completes the proof of the first statement.  

(ii): Suppose now that $\widetilde{Z}_1\geq_{\text{st}} Z_1$. Let $\widetilde{R}_0$ be the unique solution to $\rho(\widetilde{R}_0\widetilde{Z}_1-X_1)=0$. 
Suppose that $\widetilde{R}_0>R_0$. We will see that this assumption leads to a contradiction. 
By the monotonicity property of $\VaR_{\alpha}$ and $\ES_{\alpha}$ and the independence of $X_1$ and $\widetilde{X}_1$ of $Z_1$, 
\begin{align}\label{eq:proof_z1}
0=\rho(\widetilde{R}_0\widetilde{Z}_1-X_1)\leq \rho(\widetilde{R}_0Z_1-X_1).
\end{align}
Because of \eqref{eq:proof_z1} and uniqueness of the solution $r$ to $\rho(rZ_1-X_1)=0$,  
\begin{align}\label{eq:proof_z2}
\rho(\widetilde{R}_0Z_1-X_1)>0. 
\end{align}
Because of monotonicity of $\rho$ and the assumption $\widetilde{R}_0>R_0$, 
\begin{align}\label{eq:proof_z3}
\rho(\widetilde{R}_0Z_1-X_1)\leq \rho(R_0Z_1-X_1)=0. 
\end{align}
Since \eqref{eq:proof_z2} and \eqref{eq:proof_z3} are contradictory we conclude that the assumption $\widetilde{R}_0>R_0$ is false. Hence, $\widetilde{R}_0\leq R_0$ which completes the proof of the second statement.  

(iii): The proof reuses some arguments from the proof of Proposition 1 in \cite{engsner2017insurance}. 
By Statement (i), $\widetilde{X}_1\geq_{\text{st}} X_1$ implies $\widetilde{R}_0\geq R_0$. Note that 	 
\begin{align*}
	(\widetilde{R}_0Z_1-\widetilde{X}_1)^+&\leq_{\text{st}} 
	(R_0Z_1-\widetilde{X}_1)^++(\widetilde{R}_0-R_0)Z_1\\
	&\leq_{\text{st}} (R_0Z_1-X_1)^++(\widetilde{R}_0-R_0)Z_1.
\end{align*}
Hence, $\widetilde{C}_0-C_0\leq (\widetilde{R}_0-R_0)\E[Z_1]/(1+\eta)$. Therefore, 
\begin{align*}
	\widetilde{V}_0-V_0&=\widetilde{R}_0-R_0+C_0-\widetilde{C}_0\\ 
	&\geq (\widetilde{R}_0-R_0)\frac{1+\eta-\E[Z_1]}{1+\eta}\geq 0.
\end{align*}

(iv) By Statement (ii), $\widetilde{Z}_1\geq_{\text{st}} Z_1$ implies $\widetilde{R}_0\leq R_0$. Then
\begin{align*}
	(R_0Z_1-X_1)^+&\leq_{\text{st}} 
	(\widetilde{R}_0Z_1-X_1)^++(R_0-\widetilde{R}_0)Z_1\\
	&\leq_{\text{st}} (\widetilde{R}_0\widetilde{Z}_1-X_1)^++(R_0-\widetilde{R}_0)Z_1.
\end{align*}
Hence, $C_0-\widetilde{C}_0\leq (R_0-\widetilde{R}_0)\E[Z_1]/(1+\eta)$. 
Therefore, 
\begin{align*}
	\widetilde{V}_0-V_0&=\widetilde{R}_0-R_0+C_0-\widetilde{C}_0\\
	&\leq (\widetilde{R}_0-R_0)\frac{1+\eta-\E[Z_1]}{1+\eta} \leq 0.
\end{align*}
\end{proof}

\begin{proof}[Proof of Proposition \ref{prop:icxV0}]
	Set $Z_1:=Z_1^{w}$ and $\widetilde{Z}_1:=Z_1^{\widetilde{w}}$. 
	Set $R_0:=R_0^{w}$ and $\widetilde{R}_0:=R_0^{\widetilde{w}}$. 
	For $\widetilde{R}_0\leq R_0$ and any $Z_1,\widetilde{Z}_1,X_1$, 
	\begin{align*}
		(R_0Z_1-X_1)^+&\leq 
		(\widetilde{R}_0Z_1-X_1)^++(R_0-\widetilde{R}_0)Z_1\\
		&= (\widetilde{R}_0\widetilde{Z}_1-X_1)^++(R_0-\widetilde{R}_0)Z_1+\Delta,
	\end{align*}
	where 
	\begin{align*}
		\Delta:=(\widetilde{R}_0Z_1-X_1)^+-(\widetilde{R}_0\widetilde{Z}_1-X_1)^+.
	\end{align*}
	Note that $\E[\Delta]=\E[\E[\Delta\mid X_1]]$ and that $\E[\Delta\mid X_1]\leq 0$ follows from $Z_1\leq_{\text{icx}}\widetilde{Z}_1$. 
	Hence, 
	\begin{align*}
		C_0\leq \widetilde{C}_0+\frac{(R_0-\widetilde{R}_0)\E[Z_1]}{1+\eta}+\frac{\E[\Delta]}{1+\eta} \leq \widetilde{C}_0+\frac{(R_0-\widetilde{R}_0)\E[Z_1]}{1+\eta}.
	\end{align*}
	Consequently, 
	\begin{align*}
		\widetilde{V}_0-V_0=\widetilde{R}_0-R_0+C_0-\widetilde{C}_0
		\leq (\widetilde{R}_0-R_0)\frac{1+\eta-\E[Z_1]}{1+\eta} \leq 0
	\end{align*}
	given that  $1+\eta\geq \E[Z_1]$. 
\end{proof}

\begin{proof}[Proof of Proposition \ref{prop:icxV02}]
Note that $R_0^{w}Z_1^{w}\geq_{\text{icx}} R_0^0Z_1^0=R_0^0$ is equivalent to that, for all $\beta\in [0,1]$, 
\begin{align*}
\int_{\beta}^{1}F_{R_0^{w}(wS_1+1-w)}^{-1}(u)du-(1-\beta)R_0^0\geq 0.
\end{align*}
The inequality is equivalent to 
\begin{align*}
R_0^{w}\bigg(w\frac{1}{1-\beta}\int_{\beta}^{1}F_{S_1}^{-1}(u)du+1-w\bigg)\geq R_0^0
\end{align*}
and the left-hand side is bounded from below by $R_0^{w}(w\E[S_1]+1-w)$. Moreover, $R_0^{w}Z_1^{w}\geq_{\text{icx}} R_0^0Z_1^0$ implies $C_0^{w}\geq C_0^{0}$. Finally, $R_0^{w}\leq R_0^0$ and $C_0^{w}\geq C_0^{0}$ together imply $V_0^{w}=R_0^{w}-C_0^{w}\leq R_0^0-C_0^0=V_0^0$. 
\end{proof}

\begin{proof}[Proof of Proposition \ref{prop:init_der_base_case}]
Write $r(w):=R_0^w$. For all $w$, $\P(X\leq r(w)(wS_1+1-w))=1-\alpha$. Suppose that $X_1$ has a density function $f_{X_1}$. 
From 
\begin{align*}
\int F_{X_1}(r(w)(ws+1-w))dF_{S_1}(s)=1-\alpha \text{ for all } w
\end{align*}
follows that the derivative of the left-hand side with respect to $w$ is zero. The assumptions allow differentiating under the integral sign, which gives, for all $w$, 
\begin{align*}
\int f_{X_1}(r(w)(ws+1-w))\Big(r'(w)(ws+1-w)+r(w)(s-1)\Big)dF_{S_1}(s)=0. 
\end{align*}
Let 
\begin{align*}
h(w,s):=f_{X_1}(r(w)(ws+1-w))\Big(r'(w)(ws+1-w)+r(w)(s-1)\Big)
\end{align*}
and note that the assumptions imply that $|h(w,s)|\leq c_1s+c_2$ for constants $c_1,c_2>0$. 
Applying the Dominated Convergence Theorem gives 
\begin{align*}
0&=\lim_{w\to 0}\int f_{X_1}(r(w)(ws+1-w))\Big(r'(w)(ws+1-w)+r(w)(s-1)\Big)dF_{S_1}(s)\\
&=\int \lim_{w\to 0}f_{X_1}(r(w)(ws+1-w))\Big(r'(w)(ws+1-w)+r(w)(s-1)\Big)dF_{S_1}(s)\\
&=f_{X_1}(r(0))\Big(r'(0)+r(0)(\E[S_1]-1)\Big)
\end{align*}
Hence, $r'(0)=-r(0)(\E[S_1]-1)<0$. 

Set $\psi(w,s):=r(w)(ws+1-w)$. 
\begin{align*}
&\frac{d}{dw}\E[(r(w)(wS_1+1-w)-X_1)^+]\\
&\quad =\frac{d}{dw} \int_{0}^{\infty}\int_{0}^{\psi(w,s)}\big(\psi(w,s)-x\big)f_{X_1}(x)dx dF_{S_1}(s)\\
&\quad =\int_{0}^{\infty}\int_{0}^{\psi(w,s)}\frac{\partial}{\partial w}\psi(w,s)f_{X_1}(x)dx dF_{S_1}(s)\\
&\quad =\int_{0}^{\infty}F_{X_1}\big(\psi(w,s)\big)\Big(r'(w)(ws+1-w)+r(w)(s-1)\Big)dF_{S_1}(s)
\end{align*}
Since $\lim_{w\to 0}\psi(w,s)=r(0)$, applying the Dominated Convergence Theorem gives 
\begin{align*}
&\lim_{w\to 0}\frac{d}{dw}\E[(r(w)(wS_1+1-w)-X_1)^+]\\
&\quad =\int_{0}^{\infty}F_{X_1}\big(r(0)\big)\Big(r'(0)+r(0)(s-1)\Big)dF_{S_1}(s) \\
&\quad =F_{X_1}(r(0))(r'(0)+r(0)(\E[S_1]-1))\\
&\quad =0.
\end{align*}
Since $V_0^w=R_0^w-C_0^w$ the proof is complete. 
\end{proof}

\begin{proof}[Proof of Proposition \ref{prop:capital_req_normal}] 
Equation \eqref{eqr} is straight-forward. Define the convex function $[0,\infty)\ni x\mapsto f_{\rho}(x)$ by 
\begin{align*}
	f_{\rho}(x):=\rho(xZ_1-X_1)
	=\gamma-x\mu+\Phi^{-1}(1-\alpha)\sqrt{x^2\sigma^2+\nu^2}.
\end{align*}
Then, using $\sqrt{x^2\sigma^2+\nu^2}<x\sigma+\nu$, we get the following bounds for $f_{\rho  }(x)$: 
\begin{align*}
	\gamma + x (\sigma \Phi ^{-1}(1 - \alpha ) - \mu    ) \le
	f_{\rho  }(x)
	\le \gamma  + \nu \Phi ^{-1}(1- \alpha ) + x (\sigma \Phi ^{-1}(1 - \alpha ) - \mu    )  .
\end{align*}
It immediately follows that for $\mu \le \sigma \Phi ^{-1}(1 - \alpha )$, $f_{\rho  }(x) > 0$ for all positive $x$, settling (ii). \\
On the other hand, for $\mu > \sigma \Phi ^{-1}(1 - \alpha )$, there exists an $x> 0$ with  $f_{\rho  }(x) = 0$. From \eqref{eqr} it then follows that the function 
\begin{align*}
	g(x)=(\mu^{2} - \sigma ^{2} \Phi ^{-1}(1-\alpha )^{2}   ) \,x^{2} + (-2 \mu \gamma  )\, x + (\gamma ^{2} - \nu^{2} \Phi ^{-1}(1 - \alpha )^{2} )
\end{align*}
has two real zeroes. 
If $\gamma\leq \nu\Phi^{-1}(1-\alpha)$, then $g(x)$ has a unique positive zero $R_0>0$. 
	 If $\gamma>\nu\Phi^{-1}(1-\alpha)$, the function $g(x)$ has two positive zeroes that are candidates for $R_0$. 
However, $f_{\rho}(R_0) = 0$ is equivalent to $\Phi ^{-1}(1 - \alpha ) \sqrt{\sigma ^2 R_0^{2} + \nu^{2} } = \mu R_0 - \gamma $, so $R_0$ must satisfy $\mu R_0 - \gamma \ge 0 $. Since $g(0)>0$ and $g(\gamma/\mu)<0$, only $R_0$ given in \eqref{er0} fulfills $\mu R_0 - \gamma \ge 0 $.
\end{proof}

\begin{proof}[Proof of Proposition \ref{prop:R0_is_convex}]
Note that by Remark \ref{rem:equivR0} we can write $R_0^{w}=c/h(w)$, where $c:=\gamma^2-\nu^2\Phi^{-1}(1-\alpha)^2$ and 
\begin{align}\label{eq:hfunction}
h(w):=\mu_w\gamma-\Phi^{-1}(1-\alpha)\sqrt{\gamma^2\sigma_w^2+\mu_w^2\nu^2-\sigma_w^2\nu^2\Phi^{-1}(1-\alpha)^2}.
\end{align}
Since $c>0$, $w\mapsto R_0^{w}$ is convex if $w\mapsto h(w)$ is concave (a decreasing convex function composed with a concave function is convex). Set 
\begin{align}\label{eq:gfunction}
g(w):=(\gamma^2-\nu^2\Phi^{-1}(1-\alpha)^2)\sigma^2w^2+(w\mu+1-w)^2\nu^2
\end{align}
Then $h(w)=(w\mu+1-w)\gamma-\Phi^{-1}(1-\alpha)\sqrt{g(w)}$. Hence, $h''(w)<0$ if and only if $(g'(w))^2-2g(w)g''(w)<0$. Direct computations show that this inequality holds if $\gamma^2-\nu^2\Phi^{-1}(1-\alpha)^2>0$, i.e.~without having to impose any additional conditions.  
Moreover, 
\begin{align*}
\frac{dR_0^{w}}{dw}(0)&=-c\frac{h'(0)}{h^2(0)}=-c\frac{(\mu-1)(\gamma-\nu\Phi^{-1}(1-\alpha))}{(\gamma-\nu\Phi^{-1}(1-\alpha))^2}\\
&=-(\mu-1)(\gamma+\nu\Phi^{-1}(1-\alpha))<0.
\end{align*}
From the expression of $R_0^w$ in Remark \ref{rem:equivR0}, the equation $R_0^w=R_0^0=\gamma+\nu\Phi^{-1}(1-\alpha)$ implies that  
\begin{align*}
0=w(\gamma-\nu\Phi^{-1}(1-\alpha))\Big(&w(\gamma+\nu\Phi^{-1}(1-\alpha))\big((\mu-1)^2-\sigma^2\Phi^{-1}(1-\alpha)^2\big)\\
&+2(\mu-1)\nu\Phi^{-1}(1-\alpha)\Big).
\end{align*}
If $\mu\geq 1+\sigma\Phi^{-1}(1-\alpha)$, then there is no strictly positive solution $w$, which implies that $R_0^w<R_0^0$ for all $w>0$. 
If $\mu<1+\sigma\Phi^{-1}(1-\alpha)$, then there is a strictly positive solution $w$ which satisfies 
\begin{align*}
w(\gamma+\nu\Phi^{-1}(1-\alpha))\big(\sigma^2\Phi^{-1}(1-\alpha)^2-(\mu-1)^2\big)=2(\mu-1)\nu\Phi^{-1}(1-\alpha),
\end{align*}
from which the conclusion follows using the convexity of $w\mapsto R_0^{w}$.
\end{proof}

\begin{lemma}\label{lem:C0lemma}
If $G\sim N(0,1)$ and $e,f$ are constants, with $f>0$, such that $\rho(e+fG)=0$, then for any translation invariant and positively homogeneous risk measure $\rho$, one has
\begin{align*}
\E[(e+fG)^+]=e\bigg(\Phi(\rho(G))+\frac{\phi(\rho(G))}{\rho(G)}\bigg),  
\end{align*}
where $\Phi$ and $\phi$ denote the standard normal distribution function and density function, and $\overline{\Phi}=1-\Phi$.  
\end{lemma}
\begin{proof}The result follows from
\begin{align*}
\E[(e+fG)^+]&=\int_{-e/f}^{\infty} (e+fg)\phi(g)dg\\
&=e\overline{\Phi}(-e/f)+f\int_{-e/f}^{\infty}g\phi(g)dg\\
&=e\overline{\Phi}(-e/f)+f\phi(-e/f),
\end{align*}
From $-e+f\rho(G)=0$ follow $-e/f=-\rho(G)$ and $f=e/\rho(G)$. Finally, $\overline{\Phi}(-\rho(G))=\Phi(\rho(G))$ and $\phi(\rho(G))=\phi(-\rho(G))$.   
\end{proof}

\begin{proof}[Proof of Proposition \ref{prop:C0V0_normal_model}]
	From Lemma \ref{lem:C0lemma} it follows that 
	\begin{align*}
		C_0=\frac{R_0\mu-\gamma}{1+\eta}\bigg(1-\alpha+\frac{\phi(\Phi^{-1}(1-\alpha))}{\Phi^{-1}(1-\alpha)}\bigg),
	\end{align*}
	where we used that $\Phi(\rho(G))=\Phi(\Phi^{-1}(1-\alpha))=1-\alpha$. 
	The Mill's ratio bounds
	\begin{align*}
		\frac{z}{z^2+1}<\frac{\overline{\Phi}(z)}{\phi(z)}<\frac{1}{z}, \quad z > 0
	\end{align*}
	imply (for $\alpha<1/2$)
	\begin{align*}
		\alpha\Phi^{-1}(1-\alpha)<\phi(\Phi^{-1}(1-\alpha))<\alpha\bigg(\Phi^{-1}(1-\alpha)+\frac{1}{\Phi^{-1}(1-\alpha)}\bigg).
	\end{align*}
\end{proof}

\begin{proof}[Proof of Proposition \ref{prop:C0greater_normal_model}]
From 
\begin{align*}
	0=\VaR_{\alpha}(R_0^{w}Z_1-X_1)=\gamma-R_0^{w}\mu_{w}+\Phi^{-1}(1-\alpha)\sqrt{(R_0^{w})^2\sigma_{w}^2+\nu^2} 
\end{align*}
follows that 
\begin{align*}
	R_0^{w}\mu_{w}-R_0^{0}=\Phi^{-1}(1-\alpha)\Big(\sqrt{(R_0^{w})^2\sigma_{w}^2+\nu^2}-\nu\Big)>0.
\end{align*}
Hence, 
\begin{align*}
	\frac{1 + \eta  }{1 + \delta (\alpha )}	(C_{0}^{w} - C_{0}^{0}) &= (R_{0}^{w}\mu_{w}   - R_{0}^{0}) > 0. 
\end{align*}
Consequently, $C_{0}^{w} > C_{0}^{0}$. 
\end{proof}

\bibliographystyle{abbrv}
\bibliography{coc}

\end{document}